\documentclass[12pt]{article}

\usepackage[latin1]{inputenc}

\usepackage{amstext,amsmath,amssymb,amsfonts,bbm,amsmath,extarrows}
\usepackage{authblk}
\usepackage{caption}
\usepackage{xcolor} 
\usepackage{braket}  
\usepackage{amsthm}   
\usepackage{tikz}

\usetikzlibrary{arrows.meta}
\usetikzlibrary{shapes.misc}

\tikzset{cross/.style={cross out, draw=black, minimum size=2*(#1-\pgflinewidth), inner sep=0pt, outer sep=0pt},cross/.default={1pt}}

\usepackage{hyperref}
\hypersetup{
    colorlinks=false,
    menubordercolor=red,
    linkbordercolor=blue
}

\usepackage{cite}

\allowdisplaybreaks[4]

\usepackage{float}  
\usepackage{subfigure}
\usepackage{subcaption}

\usepackage[lmargin=40pt,rmargin=40pt,tmargin=70pt,bmargin=60pt,headheight=26pt,
]{geometry}

\definecolor{bluegray}{rgb}{0.4, 0.6, 0.8} 

\newcommand{\be}{\begin{equation}}
\newcommand{\ee}{\end{equation}}

\newcommand{\R}{\mathbb{R}}

\newcommand*\diff{\mathop{}\!\mathrm{d}}

\numberwithin{equation}{section}	

\newtheorem{remark}{Remark}
\newtheorem{proposition}{Proposition}

\newtheorem{lemma}{Lemma}
\newtheorem{theorem}{Theorem}

\theoremstyle{remark}

\begin{document}

\title{\bf From path integral quantization to stochastic quantization: a pedestrian's journey}

 \author[1]{Dario Benedetti}
 \author[2]{Ilya Chevyrev}
 \author[3]{Razvan Gurau}

\affil[1]{\normalsize\it 
CPHT, CNRS, \'Ecole polytechnique, Institut Polytechnique de Paris, 91120 Palaiseau, France
\authorcr email: dario.benedetti@polytechnique.edu \authorcr \hfill}
\affil[2]{\normalsize\it SISSA (International School for Advanced Studies), via Bonomea 265, 34136 Trieste, Italy
 \authorcr email: ichevyrev@gmail.com \authorcr \hfill }

 \affil[3]{\normalsize\it Heidelberg University, Institut f\"ur Theoretische Physik, Philosophenweg 19, 69120 Heidelberg, Germany
 \authorcr email: gurau@thphys.uni-heidelberg.de \authorcr \hfill}

\date{}

\maketitle

\hrule\bigskip

\begin{abstract}
We give two novel proofs that the path integral and stochastic quantizations of generic scalar Euclidean quantum field theories are equivalent. Our proofs rely on Taylor interpolations indexed by forests, in the fashion of constructive field theory. The first proof works at the level of individual terms in the Feynman expansion, with the forests appearing as spanning forests in Feynman graphs. The second one works at the level of the path integral and avoids the full expansion of the Feynman perturbation series.
\end{abstract}

\hrule\bigskip

\tableofcontents

 
\section{Introduction}
\label{sec:intro}

In a seminal paper \cite{Parisi:1980ys}, Parisi and Wu argued that the path integral formulation of Euclidean quantum field theory (EQFT) can be obtained as the equilibrium limit of a stochastic dynamic of Langevin type with respect to a fictitious time, in line with the interpretation of EQFT as equilibrium statistical mechanics.
This idea stimulated a great deal of activity in the physics community, mainly in the 1980s \cite{Damgaard_Huffel:87_SQ}, and has received renewed attention in the mathematics community over the last decade. 
In particular, with the advent of new tools in the study of stochastic partial differential equations, such as regularity structures \cite{Hairer14}, paracontrolled distributions \cite{GIP15},
and renormalization group \cite{Kupiainen_16,Duch_25}, the way was opened to a new approach to the rigorous construction and analysis of the EQFT measure for various models, such as $\Phi^4_3$~\cite{MW17Phi43, MoinatWeber20, AK20, GH21, Hairer_Steele_22}. The stochastic methods are complementary to the constructive field theory approach based on the path integral formulation \cite{GlimmJaffe,Rivasseau:1991ub} and, at a technical level,
these two rigorous formulations of EQFT have remained somewhat separated.
One of the aims of this paper is to bring these two approaches closer together.

The contributions of this paper are two novel proofs that the path integral and stochastic quantizations are equivalent.

The first proof works at the level of the Feynman graph expansion in the path integral formulation of a scalar EQFT \cite{Rivasseau:1991ub} and shows that it is equivalent to the tree expansion in stochastic quantization \cite{Parisi:1980ys}. 
Although this result is known in the literature, the currently available proofs are either not direct (using a Fokker--Planck equation \cite{Floratos:1982xj}) or somewhat sketchy \cite{Grimus:1982wn,Nakano:1982rz,Nakazato:1983gm,Egorian:1983wz,Araki:1984qa,Kirschner:1984sx,Huffel:1985ma} (see also \cite[Sec.~3.3]{Damgaard_Huffel:87_SQ} and references therein).
In particular, the classical proof of \cite{Huffel:1985ma} relies on a momentum representation of the covariance, momentum conservation and an induction. Besides the fact that the induction is not trivial, and it does not provide a constructive route to matching the two expansions, the proof in \cite{Huffel:1985ma} does not directly apply to theories in which the momentum is not conserved, at least not in the way it has been written.
Most other proofs \cite{Nakano:1982rz,Nakazato:1983gm,Egorian:1983wz,Araki:1984qa,Kirschner:1984sx} rely instead on a hidden supersymmetric structure of the stochastic quantization, that is made manifest by rewriting the stochastic correlation functions in terms of a formal path integral involving auxiliary scalar and Grassmann fields \cite{Martin:1973zz,DeDominicis:1976,Jansen:1981}, which however goes against the spirit of using stochastic partial differential equations as a substitute for formal path integrals.
In contrast, our first proof avoids such constructions and proceeds by a Taylor expansion at the level of Feynman amplitudes,
leading to an explicit representation of each graph contribution as a sum over spanning forests weighted by stochastic amplitudes.

Our second proof of equivalence of the two quantization procedures relies on a novel strategy involving a tree-indexed Taylor interpolation directly at the level of the path integral without expressing it as a sum over Feynman graphs. Such Taylor interpolations are the backbone of constructive field theory \cite{Rivasseau:1991ub,salmhofer2007renormalization} and ours is close in spirit to the original Battle--Federbush formula\cite{Battle:1982pv,Battle:1984cj}.
In particular, tree-indexed series similar to those we consider are a first step towards non-perturbative constructions of the measure. 
However, contrary to the usual constructive field theory interpolations which compute the connected expectations, our interpolation computes directly the expectations (that is the full Schwinger functions) of the model.

Both proofs share a number of features:
\begin{itemize}
    \item[-] they work in an arbitrary theory with a positively defined covariance; in particular, contrary to the proof of \cite{Huffel:1985ma}, our arguments do not rely on momentum conservation and accommodate covariances such as $-\Delta + x^2$, as in the Grosse--Wulkenhaar model \cite{Grosse:2004yu}, which has recently been treated in stochastic quantization \cite{Song:2025iyg};\footnote{The work \cite{Song:2025iyg} proves existence (but not uniqueness) of an invariant measure for the stochastic quantization equations of the 2-dimensional Grosse--Wulkenhaar model.
    The question of whether this measure  coincides with the standard path integral measure, as constructed e.g. in \cite{Wang_18_GW_model}, appears to be unaddressed.}
    \item[-] they are directly generalizable to curved space;\footnote{The construction of EQFT on curved spaces is an important open problem that has recently seen some advances within the stochastic approach (see \cite{Bailleul:2023wyv} and references therein).}
    \item[-] they are constructive and avoid inductive arguments; in particular, both rely on explicit Taylor interpolations, which makes the link between Feynman-type contributions and forest/tree expansions transparent.
\end{itemize}

We point out that we do not consider renormalization in this paper: our statements should be understood in a formal sense or in settings where the relevant quantities are well-defined.
In theories with divergences, our results would need to be complemented by an appropriate analysis of regularization and subtractions.

\paragraph{Plan of the paper.} In Sec.~\ref{sec:pertQFT}, we begin by recalling the basic settings of EQFT in the path integral and stochastic formulations for a generic scalar field theory and we state our main theorem.
In Sec.~\ref{sec:prelim}, we take a step back and introduce a number of prerequisites on combinatorial structures, Gaussian measures, and diffusion equations.
In Sec.~\ref{sec:pert}, we give precise definitions of the perturbative expansions to all orders 
in the path integral and stochastic settings and we reformulate our
main theorem at the perturbative level in terms of graph amplitudes.
In Sec.~\ref{sec:first_proof}, we prove the perturbative formulation of our theorem.
In Sec.~\ref{sec:pathint}, we prove the main theorem as formulated in  Sec.~\ref{sec:pertQFT} non-perturbatively, directly at the level of the path integral.
Readers interested in the latter proof may follow the route of Secs.~\ref{sec:pertQFT}, \ref{sec:prelim}, ~\ref{sec:stoch-exp}, and Sec.~\ref{sec:pathint}; the perturbative expansion in Sec.~\ref{sec:pert} (with the exception of Sec.~\ref{sec:stoch-exp}), and the graph-by-graph proof in Sec.~\ref{sec:first_proof}, are not needed for this route.
In the Appendices we provide explicit examples of the perturbative expansions and the Taylor interpolation at lowest orders.

\paragraph{Acknowledgments.}

I.C. acknowledges support from the European Research Council (ERC) via the Starting Grant SQGT 101116964.
R.G. acknowledges support from the Deutsche Forschungsgemeinschaft (DFG, German Research Foundation) under Germany's Excellence Strategy EXC 2181/1 - 390900948 (the Heidelberg STRUCTURES Excellence Cluster).
The authors are grateful to the Mathematisches Forschungsinstitut Oberwolfach (MFO) for its hospitality during the workshop ``Renormalisation and Randomness'' held in September 2025, at which this project was initiated.

\section{Path integral and stochastic quantization}
\label{sec:pertQFT}

We denote $d$ the dimension of the Euclidean space $\mathbb{R}^d$,
and for simplicity we only consider real scalar fields $\phi:\mathbb{R}^d\to \mathbb{R}$.
In order to simplify the notation, we sometimes denote arguments of functions as indices, for instance $\phi_x \equiv \phi(x)$, $C(x,y) = C_{xy}$ and $\delta_{xy} =\delta^d(x-y)$ denotes the Dirac delta in the appropriate dimension; we sometimes denote $ \int d^d x \equiv \int dx \equiv \int_x $ and use the shorthand notation $\phi A \phi \equiv  \int_{xy} \; \phi_x A_{xy} \phi_y$.

\paragraph{Euclidean quantum field theory.} 
EQFT is defined by a formal path integral measure of the form:
\[
     \diff \nu(\phi)  =  [\diff\phi]  \;\frac{ 1 }{ Z(0) } \;e^{-S(\phi)}  \;, \qquad Z(0) = \int [\diff\phi] \;e^{-S(\phi)} \; ,
\]
where the partition function $Z(0)$ is a normalization constant and the action $S(\phi)$ is some real functional. 
The aim of EQFT is to compute the moments ($n$-point functions or correlation functions) of the theory, which are related to physical observables.
In constructive field theory \cite{GlimmJaffe,Rivasseau:1991ub,salmhofer2007renormalization} one is interested in making sense of the above definition and check that the moments have some specific properties, for instance that they obey the Osterwalder--Schrader axioms \cite{Osterwalder:1973dx,Osterwalder:1974tc}.

A common choice is to separate the action into a free part, quadratic in the field, and an interaction:
\be \label{eq:action}
 S(\phi) = \frac{1}{2}\phi A \phi + 
 V(\phi) \;,
\ee
where the operator $A$ in the free part of the action is assumed to be a closed self-adjoint positive-definite operator, whose inverse is the free \emph{covariance} (or propagator) of the theory $C_{xy} = (A^{-1})_{xy}$.

The perturbation, or interaction, $V(\phi)$ is a possibly nonlocal functional of $\phi$:
\begin{equation}\label{eq:geint}
V(\phi) = \sum_{q\ge 1} \frac{1}{q+1} \int_{ x^1 , x^2, \ldots, x^{q+1} } \;  V (x^1, x^2 , \ldots , x^{q+1}) \;
   \phi_{x^1} \ldots \phi_{x^{q+1} } \; ,
\end{equation}
where by a slight abuse of notation we denote by $V$ both the full potential and its kernels; there shall be no confusion because they always appear with their specific arguments.
We assume that the \emph{vertex kernels} $  V (x^1, \dots , x^{q+1} )$ are cyclically symmetric in their arguments so that:\footnote{Of course for a single scalar the kernels can always be fully symmetrized in their arguments.
We also remind that for a generic functional $F(\phi)$, the functional derivative $\frac{\delta F[\phi]}{\delta\phi_x}$ is defined by the relation $\lim_{\epsilon\to 0}\frac{F(\phi+\epsilon f)-F(\phi)}{\epsilon}=\int_x \frac{\delta F[\phi]}{\delta\phi_x} f(x)$,
i.e. it is the gradient of $F$ at $\phi$ w.r.t. to the $L^2(\R^d)$ inner product,
and it satisfies  usual properties such as linearity, product and chain rules, and moreover $\frac{\delta\phi_x}{\delta\phi_y}=\delta_{xy}$.
} 
\begin{equation}\label{eq:geintderiv}
  \frac{\delta V}{ \delta \phi_{x} } = \sum_{q\ge 1}  \int_{ x^2, \ldots, x^{q+1}} \;  V (x, x^2, \dots , x^{q+1}) \; \phi_{x^2} \ldots \phi_{x^{q+1} }  \; .
\end{equation}
For convenience we have included in the interaction only monomials of degree at least 2
in the interaction, ensuring that $\phi=0$ is a solution of the classical equations of motion $(A\phi)_x+ \frac{\delta V}{ \delta \phi_x}  = 0$, but with minimal effort linear terms can also be taken into account.

For trivial vertex kernels $V (x^1, \dots , x^{q+1} )=0 $, we fall back on the free theory which is a normalized Gaussian measure with covariance $C$. A local monomial interaction is obtained when only one vertex kernel is nonvanishing and completely local:
\[
V (x^1, x^2\dots , x^{q+1}) = g \,  \prod_{i=2}^{q+1} \delta^d(x^1-x^i) \;, \qquad V(\phi) =\frac{g}{q+1} \int_x \phi_x^{q+1} \;,
\]
while a derivative coupling is obtained for instance for:
\[
V(x^1,x^2,x^3) = g \, \frac{ \Delta_{x^1} + \Delta_{x^2} + \Delta_{x^3} }{3} \, \delta^d(x^1-x^2)\delta^d(x^1-x^3)
 \;, \qquad V(\phi) = \frac{g}{3} \int_{x} \phi_x^2 \,  \Delta \phi_x \;,
\]
where $\Delta  = \partial_\mu \partial^\mu $ denotes the Laplacian. Note that we have normalized the interaction monomial of degree $q+1$ by a $1/(q+1)$ factor instead of $1/(q+1)!$, which is more common in the physics literature, as this simplifies the combinatorics later. The usual $\phi^4$ model with mass $m^2$ corresponds to:
\[
A= -\Delta+m^2  \;, \qquad V(x^1,x^2,x^3,x^4) = 
  g \, \delta^d(x^1-x^2)\delta^d(x^1-x^3)\delta^d(x^1-x^4) \;. 
\]

\paragraph{Path integral quantization.} In order to compute the correlation functions in the path integral quantization one considers the partition function with source:
\[
Z(J) = \int  [\diff  \phi] \; e^{-S(\phi) +J\phi } \; , \qquad J\phi  \equiv \int_x J_x \phi_x \;,
\]
such that $Z(J) / Z(0)$ is the generating function of the $n$-point  functions of the model:
\[
\Braket{\phi_{x^1} \dots \phi_{x^n}  }
 =\frac{1}{Z(0)} \; \frac{\delta^n Z}{ \delta J_{x^1} \dots \delta J_{x^n} } \Bigg{|}_{J=0} = 
\frac{1}{Z(0)} \int  [\diff \phi]  \;e^{ - S(\phi)}  \; \phi_{x^1} \dots \phi_{x^n} 
 \; .
\]
The aim of constructive field theory \cite{GlimmJaffe,Rivasseau:1991ub,salmhofer2007renormalization} is to trade such ill defined path integral expressions for convergent expansions and subsequently prove they respect the appropriate axioms. As one can make sense of free theories (Gaussian measures), one usually treats the interacting path integral measure as a perturbed Gaussian measure:
\begin{equation}\label{eq:nu}
  \diff \nu(\phi) =\diff \mu_{C} (\phi) \; \frac{1}{Z(0)} \; e^{-V(\phi)} \;, 
\end{equation}
where $\diff \mu_C$ denotes the normalized Gaussian measure with covariance $C = A^{-1}$ (see Sec.~\ref{sec:Gauss}),
and attempts to construct the $n$-point functions in some convergent expansion around the Gaussian case. 

\paragraph{Stochastic quantization.}
The stochastic quantization procedure computes correlation functions by a different method. 
The starting point of stochastic quantization is the stochastic gradient descent equation, which is a non-linear Langevin equation with respect to the fictitious time $t$:
\be \label{eq:Langevin}
\begin{split}
 & \partial_t \Phi_{(t,x)} = -\frac{\delta S}{ \delta \phi_x} \Big{|}_{ \phi_x =  \Phi_{(t,x)} } + \xi_{(t,x)}  \;,\qquad 
 \Phi_{(t_{\rm in},x)} = \Phi^{\rm in}_{x}  \;; \\
 & \braket{\xi_{(t,x)} \; \xi_{(t',x')}}_{\xi} =   2  \; \delta_{tt'} \delta_{xx'} \;,
\end{split}
\ee
where $\delta_{tt'}=\delta(t-t')$ and $\delta_{xx'}=\delta^d(x-x')$ are Dirac deltas in the appropriate dimension, and thus $\xi$ is a white noise with covariance $2$ times the identity. Note that we have set the initial condition of our equation to be $\Phi^{\rm in}_{x} $ at the initial time $t_{\rm in}$: below we will always take the large-time limit, which has the effect of washing out this initial condition and leads to $\Phi^{\rm in}_{x}$-independent results, under the assumption that the Langevin dynamic relaxes to a unique equilibrium measure. In practice, the large-time limit can be achieved either by taking $t\to \infty$ while keeping $t_{\rm in}$ fixed, or by 
keeping $t$ fixed and taking $t_{\rm in}\to -\infty$. 

The stochastic correlation functions are computed by taking expectation values of the stochastic field $\Phi_{(t,x)}$, solution of \eqref{eq:Langevin}, with respect to the white noise:
\[
 \Braket{\Phi_{(t^1, x^1)} \dots \Phi_{(t^n, x^n)} }_{\xi}  \; .
\]
Given the same decomposition of the action \eqref{eq:action}, such correlators can again be computed by treating the interaction $\delta V/\delta \phi$ as a small perturbation, resulting in a different perturbative series.

\begin{remark}
While the Gaussian measure $\diff \mu_{C} (\phi)$ can directly be defined for non-invertible covariances $C$, a fact exploited for example when introducing field copies below  (see Sec.~\ref{sec:Gauss}), the same is not quite true
for the Langevin dynamic.
The standard way to define the dynamic for non-invertible $C$ is to quotient by its null-space and consider the Langevin dynamic in the quotient space.
Equivalently, if we let $C^{-1}_+$ denote the Moore--Penrose pseudoinverse of $C$ and let $P$ denote the orthogonal projection onto the row space (orthogonal complement of the null space) of $C$,
then the Langevin dynamic reads
\[
\partial_t \Phi_{(t,x)} = -C^{-1}_+ \Phi - P \frac{\delta V}{ \delta \phi_x} \Big{|}_{ \phi_x =  \Phi_{(t,x)} } + P \xi_{(t,x)}  \;.
\]
\end{remark}

\subsection{Main result}

The two quantization procedures, stochastic and \`a la path integral, are supposed to be equivalent, in the appropriate limit. Our main result is the following theorem. 

\begin{theorem}[Main theorem]\label{thm:main}
The quantum field theoretical $n$-point functions coincide with the stochastic $n$-point functions in the limit of large, coincident, fictitious times: 
\[
\Braket{\phi_{x^1} \dots \phi_{x^n}  } = \lim_{t\to \infty} \Braket{\Phi_{(t, x^1)} \dots \Phi_{(t, x^n)} }_{\xi} 
=\lim_{t_{\rm in}\to -\infty} \Braket{\Phi_{(t, x^1)} \dots \Phi_{(t, x^n)} }_{\xi} \; .
 \]

 Our first result is that this equality holds at all orders in a (formal) perturbative expansion in the vertex kernels.
 
Our second result is that this equality holds at the level of formal path integrals.
\end{theorem}

We give a perturbative reformulation of the first result in Theorem \ref{thm:main_precise} in terms of Feynman diagrams,
which we prove in Sec.~\ref{sec:first_proof}.
As we explain in Sec.~\ref{sec:intro}, while several versions of this result are known \cite{Grimus:1982wn,Nakano:1982rz,Nakazato:1983gm,Egorian:1983wz,Araki:1984qa,Kirschner:1984sx,Huffel:1985ma},
our setting is, to the best of our knowledge, more general than previously considered, and our proof is new and explicit.

\paragraph{Outline of proof.}
Perturbation theory naturally leads to the introduction of different families of combinatorial objects associated to abstract graphs, in particular with the possibility of carrying labels, an embedding in the plane,
or a rooted tree structure.
Such structures are reviewed in  Sec.~\ref{sec:prelim}.
Then, in Sec.~\ref{sec:pert}, we express the perturbative expansion of the quantum field theoretical $n$-point functions in terms of rooted combinatorial maps (i.e.\ embedded graphs),  and the one of the stochastic $n$-point functions in terms of forests of rooted plane trees (i.e.\ combinatorial maps with no loops) spanning rooted combinatorial maps. Therefore, for the equivalence proof, the main goal is to prove that for any given combinatorial map, its associated field theoretic amplitude is reproduced by the sum over the stochastic amplitudes associated to its spanning forests.

Our proof in Sec.~\ref{sec:first_proof} starts from the quantum field theoretic side and constructively shows how to reconstruct the stochastic expansion, unlike other proofs that typically go in the opposite direction.  The proof has two main ingredients. First, we provide a canonical way to construct the spanning forests of  an arbitrary rooted combinatorial map  ${\cal G}$. Second, we exploit the fundamental theorem of calculus (or Taylor interpolation), in the form:
\[
f(0) = f(a) + \int_a^0 dt \;\frac{d}{dt} f(t) \;,
\] 
in order to introduce a fictitious time $t$ and a heat kernel in the quantum field theoretic propagators. 
In particular, this is used to write, for example,
\[
\frac{1}{A} = \int_{-\infty}^t dt' \;\frac{d}{dt'}  \left( \frac{e^{-(t-t')A}}{A} \right) \;,
\]
which corresponds to the choice $f(t)=A^{-1}\,e^{tA}$ and $a=-\infty$.
By iteratively using such formula we reconstruct the sum over the stochastic amplitudes associated to the spanning forests of ${\cal G}$.
Notice that choosing $a=-\infty$ corresponds to the $t_{\rm in}\to -\infty$ limit of  Theorem \ref{thm:main}.

In Sec.~\ref{sec:pathint} we give a proof of the second result of Theorem \ref{thm:main} at the level of functional integrals. 
In essence, the proof consists of three steps.
First, by repeated use of the ``copies trick" for Gaussian integrals, reviewed in Sec.~\ref{sec:Gauss}, and of the Taylor interpolation above, we re-express the field theory path integral defining a correlator as a forest-indexed sum over multiple Gaussian path integrals. Next,  by Hubbard-Stratonovich transformations, we trade such multiple integrals for a sole white noise path integral. 
Lastly, we regroup the integrands into solutions of the stochastic PDE, thus establishing the equality between the two without expanding the full perturbation theory.

\section{Preliminaries}
\label{sec:prelim}

We start by some prerequisites: we revisit some combinatorial structures which will be relevant later on, review some facts about Gaussian measures and discuss briefly generic diffusion equations.

\subsection{Combinatorial structures}\label{subsec:combi}

\paragraph{Graphs.}
An \emph{abstract graph}, or simply a graph, $G$ is a pair consisting in a finite set of \emph{vertices} $V(G)$ and a finite collection of \emph{edges}, that is unordered pairs of vertices $E(G) = \{ \{ v,w \} | v,w\in V(G) \}$.
Multiple edges are allowed, that is the same pair $\{ v,w\} $ can appear several times in $E(G)$, as well as self-loops (tadpoles) that is edges which start and end on the same vertex, $\{ v,v \}$.
The set of vertices of $G$ is partitioned into the subset of internal vertices $V^{\rm int}(G)$ and the subset of external ones $V^{\rm ext}(G)$, and all the external vertices are univalent, that is, they belong to exactly one edge.
In principle some of the internal vertices of $G$ can also be univalent, however (due to the fact that our interaction is assumed to not have linear terms) graphs with univalent internal vertices will not arise below.
We denote $d_G(v)$ the \emph{degree} of the vertex $v$ in the graph $G$, that is the number of edges incident to $v$ where a tadpole counts for two incidences.

A \emph{path} from the vertex $a$ to the vertex $b$ is a sequence of edges starting at $a$ and ending at $b$
such that any two consecutive edges share a vertex, $\{a,v_1\} , \{v_1,v_2\}, \dots , \{v_{n-1},v_n\}, \{v_n,b\}$. A graph is called \emph{connected} if any two vertices are connected by a path.
A \emph{cycle} in a graph is a path from a vertex to itself.
A \emph{forest} is a graph with no cycles and a \emph{tree} is a connected forest.
By Cayley's formula, there are $p^{p-2}$ trees over $p$ labeled vertices.

\paragraph{Combinatorial maps.}
A \emph{rooted combinatorial map} \cite{cori1992maps}\footnote{See also
\url{https://ncatlab.org/nlab/show/combinatorial+map} for a brief introduction.}  or rooted embedded graph is a quadruple ${\cal G} = ({\cal D},r,\sigma,\alpha)$ such that (see Fig.~\ref{fig:combimap}):
\begin{itemize}
    \item[-] ${\cal D}$ is a finite set, called the set of \emph{half-edges} (or darts) and $r\in{\cal D}$ is designated as the \emph{root} half-edge,
    \item[-] $\sigma: {\cal D}\to {\cal D}$ is a permutation (that is a bijection),
    \item[-] $\alpha: {\cal D}\to {\cal D}$ is an involution $\alpha^{-1}=\alpha$ with no fixed points, $\alpha(h)\neq h, \forall h\in {\cal D}$,
    \item[-] the group generated by $\sigma$ and $\alpha$ acts transitively on ${\cal D}$, that is for any two half-edges $h_1$ and $h_2$ there exists a finite word $\eta= \prod_{i=1}^q \eta_i$ with $\eta_i \in \{\sigma,\sigma^{-1}, \alpha, \alpha^{-1}\}$ such that $h_2 = \eta(h_1)$.
\end{itemize}

\begin{figure}[ht!]
\begin{center}
\begin{tikzpicture}[scale=0.4]

\begin{scope}[every node/.style={circle, thick,draw,font=\tiny,minimum height=1em}]

     \node (V1) at (0,0) {};
 
     \node (V2) at (12,0) {};

     \node (V3) at (6,4) {};
         
\end{scope}

\begin{scope}[>={Stealth[black]},
              every node/.style={fill=white,circle},
              every edge/.style={draw=black,thick}]
              
    \draw[black, thick,|-] (-8,0) -- (-6,0) node[below] {$r$};
    \draw[black, thick] (V1) -- (0,-2) node[below ] {$v^2$};
    \draw[black, thick] (V1) -- (0,2) node[above] {$v^4$}; 
    \draw[black, thick] (V1) -- (2,0) node[above] {$v^3$}; 
    \draw[black, thick] (V1) -- (-2,0) node[below] {$v^1$};
    \draw[black, thick] (V2) -- (12,2) node[left] {$w^3$}; 
    \draw[black, thick] (V2) -- (12,-2) node[left] {$w^1$}; 
    \draw[black, thick] (V2) -- (14,0) node[right] {$w^2$}; 
    \draw[black, thick]  (V3) -- (6,2) node[left] {$u^1$};   
    \draw[black, thick]  (V3) -- (8,4) node[above] {$u^2$};   

    \draw[dashed,thick] (-6,0) -- (-2,0); 
    \draw[dashed,thick] (0,-2) to[out=-60,in=-90] (12,-2);
    \draw[dashed,thick] (2,0) to[out=0,in=-90] (6,2);
    \draw[dashed,thick] (0,2) to[out=60,in=180] (8,7);
    \draw[dashed,thick] (8,7) to[out=0,in=90] (12,2);
    \draw[dashed,thick] (8,4) to[out=0,in=0] (14,0);
\end{scope}
\end{tikzpicture}

    \end{center}
    \caption{A combinatorial map with half-edges ${\cal D} =\{r,v^1,v^2,v^3,v^4,w^1,w^2,w^3,u^1,u^2\}$, root $r$,
    and permutations, written in cycle notation, $\sigma = (r) (v^1v^2v^3v^4)(w^1w^2w^3)(u^1u^2)$ and $\alpha=(rv^1) (v^2w^1) (v^3u^1)(v^4w^3) (u^2w^2) $. The root $r$ is a fixed point $(r)$ of the  permutation $\sigma$.
    }
    \label{fig:combimap}
\end{figure}
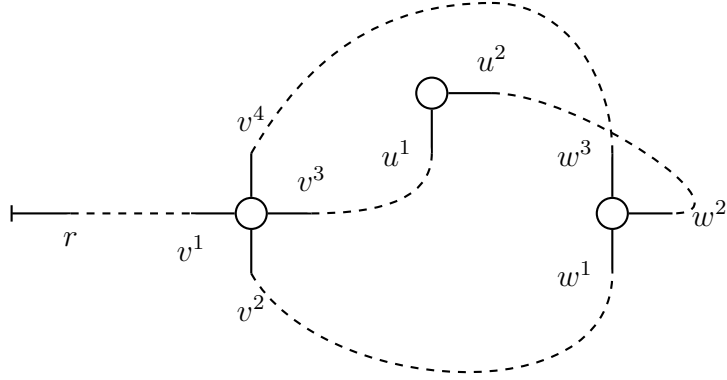

The cycles of the permutation $\sigma$, denoted $(h,\sigma(h),\sigma^2(h), \dots)$, are the \emph{vertices} of the combinatorial map ${\cal G}$, that is 
the permutation $\sigma$ encodes the successor: the half-edge $\sigma(h)$ is the successor of the half-edge $h$ when turning counterclockwise. 
The cycles $(h,\alpha(h))$ of the involution $\alpha$ are the \emph{edges} of ${\cal G}$\footnote{The cycles of $\sigma \alpha$ are called the \emph{faces} of the map and play no role in this paper.}. The transitivity of the group generated by $\sigma$ and $\alpha$ means that one can go from any half-edge to any other one by turning around vertices and stepping from one vertex to another along the edges, that is the map is \emph{connected}.

We will only consider maps such that the permutation $\sigma$ has fixed points and we will designate all such fixed points as external vertices. For maps with external vertices, the root half-edge $r$ is always chosen as a fixed point of $\sigma$.
In this case, we call $(r)$ the \emph{root vertex} and henceforth identify $(r)$ with $r$.

We denote $d_{\cal G}(v)$ the degree of the vertex $v$ in ${\cal G}$, that is the number of half-edges that belong to the cycle $v$. The set of vertices of the map is 
denoted  $V({\cal G})$ and splits into the internal vertices  $V^{\rm int}({\cal G})$ with degree at least two and the univalent external vertices $V^{\rm ext}({\cal G})$. Each internal vertex $v$ has $d_{\cal G}(v)$ ordered half-edges $v=(v^1\dots v^{d_{\cal G}(v)} ) $.
We denote $E({\cal G})$ the set of edges of ${\cal G}$, that is the set of unordered pairs $\{h,\alpha(h)\}$ of half-edges. For the map depicted in Fig.~\ref{fig:combimap}
we have:
\[
\begin{split}
V({\cal G}) &= \Big\{ (r), (v^1v^2v^3v^4), (w^1w^2w^3), (u^1u^2)\Big\} \;, 
\crcr
E({\cal G}) & = \Big\{ \{ r, v^1 \} , \{ v^2, w^1\} \{ v^3, u^1 \} , \{ v^4, w^3\} ,\{ u^2, w^2\}  \Big\} \;.
\end{split}
\]

There is a many to one mapping between combinatorial maps and abstract graphs consisting in ``forgetting the half-edge''. To every combinatorial map ${\cal G}$ one associates the abstract graph $G$ obtained by associating a vertex $v\in V(G)$ to every vertex $(v^1\dots v^{d_{\cal G}(v)} ) \in V({\cal G})$ (and the external vertices of ${\cal G}$ become external vertices of $G$) and an edge $\{v,w\} \in E(G)$ to every edge $\{v^p,w^q\} \in E({\cal G})$. 
As the maps are rooted, one of the external vertices of $G$ is marked as root.
The graph associated to the map in Fig.~\ref{fig:combimap} is:
\[
V(G) = \Big\{ r, v, w, u \Big\} \;, \;\;
E(G) = \Big\{ \{ r, v\} , \{v, w\} , \{ v,u\} , \{v, w \} , \{u,w\}  \Big\} \; .
\]
Note that the graph $G$ is always connected, as by definition the combinatorial maps are connected.

Below we will also encounter disconnected combinatorial maps.
A \emph{possibly disconnected combinatorial map} ${\cal G}$ is the disjoint union of connected rooted combinatorial maps ${\cal G} = {\cal G}^1 \sqcup {\cal G}^2 \ldots $ with internal vertices $V^{\rm int}({\cal G} ) = V^{\rm int}({\cal G}^1) \sqcup  V^{\rm int}({\cal G}^2) \sqcup \ldots$ which are at least bivalent and having univalent external vertices $V^{\rm ext}({\cal G}^1) \sqcup  V^{\rm ext}({\cal G}^2) \sqcup \ldots$ such that each component ${\cal G}^\nu$ contains at least one external vertex, and one of these external vertices is designated as root of ${\cal G}^\nu$.
In other words, ${\cal G}$ is a multiset (unordered collection) of combinatorial maps.
We denote $\pmb{\cal G}^{(r^1) \dots (r^n) }$ the set of possibly disconnected combinatorial maps with $n$ external vertices $r^1, \ldots , r^n$.

\paragraph{Unlabeled combinatorial maps.}

We have so far discussed labeled rooted combinatorial maps. 
A rooted map morphism between ${\cal G} = ({\cal D},r,\sigma,\alpha)$ and ${\cal G}' =({\cal D}',r',\sigma',\alpha')$ is a bijection $\tau:{\cal D} \to {\cal D}' $
such that $\tau(r) = r'$, $\tau\sigma = \sigma' \tau$ and $\tau\alpha = \alpha'\tau$
and we say that two maps are equivalent ${\cal G}\sim {\cal G}'$
if there exists a rooted map morphism between them. 
In particular, for every relabeling permutation $\tau:{\cal D} \to {\cal D}$ with $\tau(r)=r$, the maps $({\cal D},r,\sigma,\alpha)$ and $({\cal D},r,\sigma',\alpha')$
with $\sigma' = \tau \sigma \tau^{-1}$ and $\alpha' = \tau\alpha\tau^{-1}$ are equivalent. Taking the combinatorial map in Fig.~\ref{fig:combimap} and the transposition $\tau=(v^4w^2)$ we obtain the equivalent combinatorial map $({\cal D},r,\sigma',\alpha')$ with:
\[
\sigma' = (r) (v^1v^2v^3w^2)(w^1v^4w^3)(u^1u^2) \;, \qquad
\alpha' = (rv^1) (v^2w^1) (v^3u^1)(w^2w^3) (u^2v^4)  \; .
\]
As any element in ${\cal D}$ can be written as $\eta(r)$ for some element in the group generated by $\sigma$ and $\alpha$, rooted combinatorial maps have trivial automorphism group, that is if $\tau:{\cal D} \to {\cal D}$ is a bijection such that $\tau(r) = r$, $\sigma \tau = \tau \sigma$ and $\alpha \tau = \tau \alpha $ then $\tau$ is the identity. 

An \emph{unlabeled} rooted combinatorial map is an equivalence class $[{\cal G}]$ of combinatorial maps under morphisms.
Unlabeled rooted combinatorial maps can be canonically embedded in two-dimensional surfaces: one embeds any of the labeled combinatorial maps in the equivalence class $[{\cal G}]$ and then forgets the labels of the half edges, except the root. The unlabeled combinatorial map corresponding to the labeled combinatorial map in Fig.~\ref{fig:combimap} is represented in Fig.~\ref{fig:combimapunlab}. 

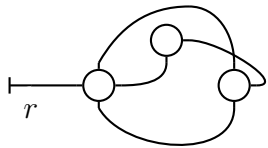
\begin{figure}[ht!]
\begin{center}
\begin{tikzpicture}[scale=0.15]

\begin{scope}[every node/.style={circle, thick,draw,font=\tiny,minimum height=1em}]

     \node (V1) at (0,0) {};
 
     \node (V2) at (12,0) {};

     \node (V3) at (6,4) {};
         
\end{scope}

\begin{scope}[>={Stealth[black]},
              every node/.style={fill=white,circle},
              every edge/.style={draw=black,thick}]
              
    \draw[black, thick,|-] (-8,0) -- (-6,0) node[below] {$r$};
    \draw[black, thick] (V1) -- (0,-2);
    \draw[black, thick] (V1) -- (0,2); 
    \draw[black, thick] (V1) -- (2,0); 
    \draw[black, thick] (V1) -- (-2,0);
    \draw[black, thick] (V2) -- (12,2); 
    \draw[black, thick] (V2) -- (12,-2); 
    \draw[black, thick] (V2) -- (14,0); 
    \draw[black, thick]  (V3) -- (6,2);   
    \draw[black, thick]  (V3) -- (8,4);   

    \draw[black,thick] (-6,0) -- (-2,0); 
    \draw[black,thick] (0,-2) to[out=-60,in=-90] (12,-2);
    \draw[black,thick] (2,0) to[out=0,in=-90] (6,2);
    \draw[black,thick] (0,2) to[out=60,in=180] (8,7);
    \draw[black,thick] (8,7) to[out=0,in=90] (12,2);
    \draw[black,thick] (8,4) to[out=0,in=0] (14,0);
\end{scope}
\end{tikzpicture}

    \end{center}
    \caption{Unlabeled rooted map corresponding to the map in Fig.~\ref{fig:combimap}.}
    \label{fig:combimapunlab}
\end{figure}

Unlabeled combinatorial maps are associated to unlabeled graphs (that is equivalence classes of graphs under graph morphisms). There are as many unlabeled combinatorial maps that lead to the same unlabeled abstract graph $G$ as there are distinct embeddings of $G$ seen as a one-complex into two-dimensional surfaces.\footnote{In order to properly embed $G$, that is with no intersections of the edges, the underlying surface must have a high enough genus.} 
Some examples of unlabeled rooted abstract graphs are given in Fig.~\ref{fig:examplegraphs}, where the coefficients displayed count how many unlabeled rooted  combinatorial maps correspond to the same abstract graph.

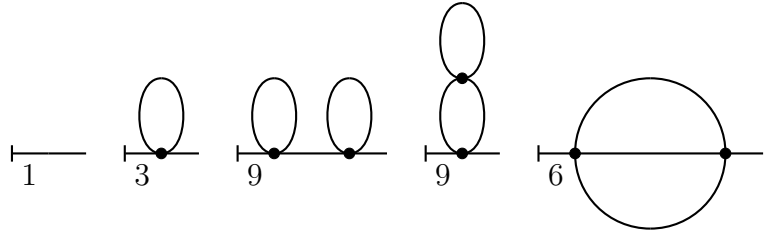
\begin{figure}[ht!]
\begin{center}
\begin{tikzpicture}[scale=0.5]

\draw[black, thick,|-] (-3,0) -- (-2,0) node[below left] {1};
\draw[black, thick] (-2,0) -- (-1,0);

\draw[black, thick,|-] (0,0) -- (1,0) node[below left] {3};
\filldraw[black] (1,0) circle (4pt);
\draw[black, thick] (1,0) -- (2,0);
\draw[black, thick] (1,0) to[out=0,in=0] (1,2);
\draw[black, thick] (1,0) to[out=180,in=180] (1,2);

\draw[black, thick,|-] (3,0) -- (4,0) node[below left] {9};
\filldraw[black] (4,0) circle (4pt);
\draw[black, thick] (4,0) to[out=0,in=0] (4,2);
\draw[black, thick] (4,0) to[out=180,in=180] (4,2);
\draw[black, thick] (4,0) -- (6,0);
\filldraw[black] (6,0) circle (4pt);
\draw[black, thick] (6,0) to[out=0,in=0] (6,2);
\draw[black, thick] (6,0) to[out=180,in=180] (6,2);
\draw[black, thick] (6,0) -- (7,0);

\draw[black, thick,|-] (8,0) -- (9,0) node[below left] {9};
\filldraw[black] (9,0) circle (4pt);
\draw[black, thick] (9,0) to[out=0,in=0] (9,2);
\draw[black, thick] (9,0) to[out=180,in=180] (9,2);
\draw[black, thick] (9,0) -- (10,0);
\filldraw[black] (9,2) circle (4pt);
\draw[black, thick] (9,2) to[out=0,in=0] (9,4);
\draw[black, thick] (9,2) to[out=180,in=180] (9,4);

\draw[black, thick,|-] (11,0) -- (12,0) node[below left] {6};
\filldraw[black] (12,0) circle (4pt);
\draw[black, thick] (12,0) to[out=90,in=180] (14,2);
\draw[black, thick] (14,2) to[out=0,in=90] (16,0);
\draw[black, thick] (12,0) to[out=-90,in=180] (14,-2);
\draw[black, thick] (14,-2) to[out=0,in=-90] (16,0);
\draw[black, thick] (12,0) -- (16,0);
\filldraw[black] (16,0) circle (4pt);
\draw[black, thick] (16,0) -- (17,0);

\end{tikzpicture}
\end{center}
    \caption{Examples of unlabeled abstract graphs rooted at the external vertex on the left. For each graph we have displayed the number of distinct possible embedding in the plane, that is the number of unlabeled rooted combinatorial maps associated to it.}
    \label{fig:examplegraphs}
\end{figure}

In order to simplify the notation, when no confusion can arise, we refer to unlabeled combinatorial maps simply as combinatorial maps and denote them as ${\cal G}$ instead of $[{\cal G}]$. In order to identify the vertices and edges in an unlabeled map ${\cal G}$, one picks any of the labeled maps in the equivalence class.

\subsubsection{Tree picking}
\label{sec:tree_picking}
One can always pick a spanning tree in an unlabeled  rooted combinatorial map by a first available successor (``keep to the right'') algorithm:
starting from the root, at each step one picks the edge connecting the current vertex with its oldest available successor (that is, the rightmost available edge) and steps to this successor vertex.
Here ``available edge'' means that the edge has not already been chosen at a previous step and that adding it will not create cycles.
In case the current vertex has no available successors, we move back to the parent of the current vertex and continue until we exhibit a spanning tree.
Examples of this are presented in Fig.~\ref{fig:trevalentwithtree}.

\begin{figure}[ht!]
\begin{center}
\begin{tikzpicture}[scale=0.3]

\begin{scope}[every node/.style={circle, thick,draw,font=\tiny,minimum height=1em}]

     \node (V11) at (0,3) {};
 
     \node (V21) at (6,3) {};
     \node (V22) at (4,6) {};
     \node (V23) at (8,6) {};

    \node (V31) at (12,3) {};
     \node (V32) at (10,6) {};
     \node (V33) at (14,6) {};

    \node (V3n1) at (18,3) {};
     \node (V3n2) at (16,6) {};
     \node (V3n3) at (20,6) {};

     \node (V41) at (24,3) {};
     \node (V42) at (24,6) {};
     \node (V43) at (24,9) {};

     \node (V51) at (30,3) {};
     \node (V52) at (30,6) {};
     \node (V53) at (30,9) {};
        
\end{scope}

\begin{scope}[>={Stealth[black]},
              every node/.style={fill=white,circle},
              every edge/.style={draw=black,thick}]
              
    \draw[black, thick,-|] (V11) -- (0,0) node[below] {$r$};
    \draw[dashed,thick]  (V11)  to[out=60,in=0] (0,5);
    \draw[dashed,thick] (0,5)  to[out=180,in=120] (V11);

   \draw[black, thick,-|] (V21) -- (6,0) node[below] {$r$};
   \draw[black, dashed] (V21) -- (V22);
   \draw[black, thick] (V21) -- (V23);
   \draw[black, thick] (V22) edge [bend left=20]  (V23);
   \draw[black, dashed] (V23) edge [bend left=20]  (V22);

   \draw[black, thick,-|] (V31) -- (12,0) node[below] {$r$};
   \draw[black, dashed] (V31) -- (V32);
   \draw[black, thick] (V31) -- (V33);
   \draw[black, dashed] (V32) to[out=90,in=180]  (V33);   
   \draw[black, thick] (V32) to[out=0,in=90]  (V33);

   \draw[black, thick,-|] (V3n1) -- (18,0) node[below] {$r$};
   \draw[black, thick] (V3n1) -- (V3n2);
   \draw[black, thick] (V3n1) -- (V3n3);
    \draw[dashed,thick]  (V3n2)  to[out=60,in=0] (16,8);
    \draw[dashed,thick] (16,8)  to[out=180,in=120] (V3n2);
       \draw[dashed,thick]  (V3n3)  to[out=60,in=0] (20,8);
    \draw[dashed,thick] (20,8)  to[out=180,in=120] (V3n3);

   \draw[black, thick,-|] (V41) -- (24,0) node[below] {$r$};
   \draw[black, thick] (V41) to[out=60,in=-60] (V42);
   \draw[black, dashed] (V41) to[out=120,in=-120] (V42);
   \draw[black, thick] (V42) -- (V43);
    \draw[dashed,thick]  (V43)  to[out=60,in=0] (24,11);
    \draw[dashed,thick] (24,11)  to[out=180,in=120] (V43);

  \draw[black, thick,-|] (V51) -- (30,0) node[below] {$r$};
   \draw[black, thick] (V51) to[out=0,in=0] (V53);
   \draw[black, dashed] (V51) to[out=180,in=-180] (V53);
   \draw[black, thick] (V53) -- (V52);
   \draw[dashed,thick]  (V52)  to[out=-60,in=0] (30,4);
  \draw[dashed,thick] (30,4)  to[out=180,in=-120] (V52);
       
\end{scope}
\end{tikzpicture}

    \end{center}
    \caption{Unlabeled rooted maps with up to three 3-valent vertices and the distinguished ``keep to the right'' tree.}
    \label{fig:trevalentwithtree}
\end{figure}
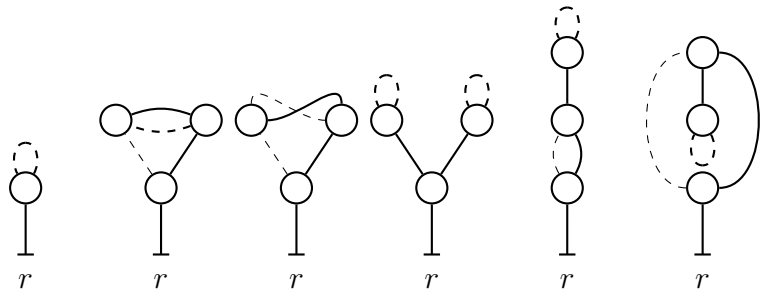

\subsection{Trees and ordinary differential equations}
\label{sec:trees_ODEs}

There are several well-known tree expansions of solutions of ordinary differential equations.
We summarize below some of these expansions which will play a role later, highlighting links to the above combinatorial structures.

\paragraph{Recursive trees.}  

A \emph{recursive tree} (or increasing tree; see \cite{BergeronFlajoletSalvy1992})
with $p$ vertices is a tree over $p$ vertices having $p$ strictly ordered labels $1<2<\dots <p$  such that if the vertex $k$ belongs to the path from the vertex  $l$ to the vertex with minimal label $1$, then $k<l$. We add to a recursive tree an additional univalent root vertex $r$, which also carries the label $0$, hooked to the vertex with label $1$.\footnote{We follow these conventions throughout the paper,
except in Sec.~\ref{sec:pathint} where we allow the root to have degree different from $1$.}

We denote $V^{\rm int}(T)$ the set of non-root vertices of $T$.
We order the edges $(i,j)\in E(T)$ with $i<j$.
Building recursive trees by adding the vertices one by one in increasing order
we observe that a vertex can be hooked to any of the vertices already present hence the number of recursive trees with $p$ non-root vertices is $ (p-1)!$. 
In Fig.~\ref{fig:rectrees} we depicted the recursive trees with up to $4$ non-root vertices.

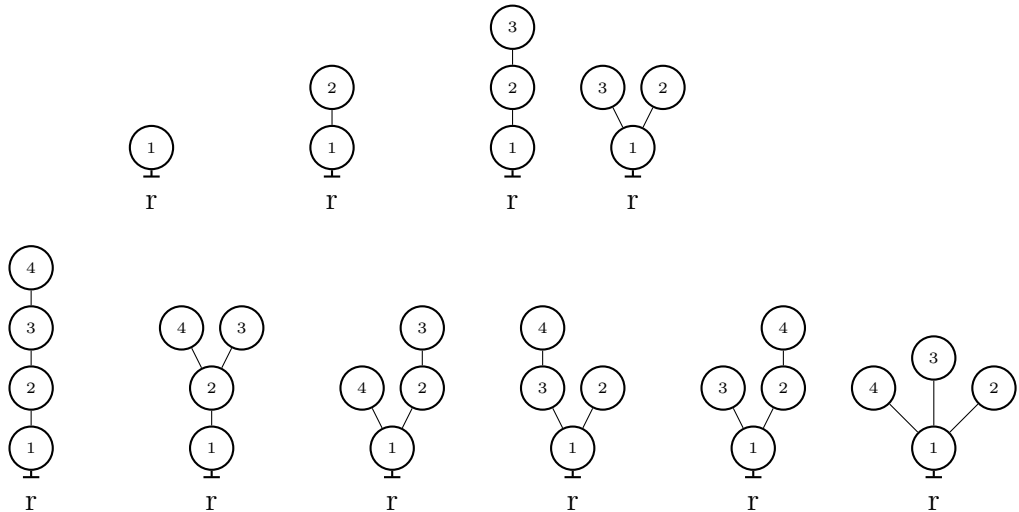
\begin{figure}[ht!]
\begin{center}
\begin{tikzpicture}[scale=0.4]

\begin{scope}[every node/.style={circle, thick,draw,font=\tiny,minimum height=1em}]
      
     \node (R1) at (0,2) {1};
 
     \node (R2) at (6,2) {1};
     \node (V21) at (6,4) {2};

     \node (R3) at (12,2) {1};
     \node (V31) at (12,4) {2};
     \node (V32) at (12,6) {3};

     \node (R4) at (16,2) {1};
     \node (V41) at (17,4) {2};
     \node (V42) at (15,4) {3};

     \node (R5) at (-4,-8) {1};
     \node (V51) at (-4,-6) {2};
     \node (V52) at (-4,-4) {3};
    \node (V53) at (-4,-2) {4} ;

     \node (R6) at (2,-8) {1};
     \node (V61) at (2,-6) {2};
     \node (V62) at (3,-4) {3};
    \node (V63) at (1,-4) {4} ;

     \node (R7) at (8,-8) {1};
     \node (V71) at (9,-6) {2};
     \node (V72) at (9,-4) {3};
    \node (V73) at (7,-6) {4} ;

    \node (R8) at (14,-8) {1};
     \node (V81) at (15,-6) {2};
     \node (V82) at (13,-6) {3};
    \node (V83) at (13,-4) {4} ;

    \node (R9) at (20,-8) {1};
     \node (V91) at (21,-6) {2};
     \node (V92) at (19,-6) {3};
    \node (V93) at (21,-4) {4} ;

     \node (R10) at (26,-8) {1};
     \node (V101) at (24,-6) {4};
     \node (V102) at (26,-5) {3};
    \node (V103) at (28,-6) {2} ;
    
\end{scope}

\begin{scope}[>={Stealth[black]},
              every node/.style={fill=white,circle},
              every edge/.style={draw=black,thick}]

   \draw (R2) -- (V21);

   \draw (R3) -- (V31);    \draw (V31) -- (V32);

   \draw (R4) -- (V41);    \draw (R4) -- (V42);

    \draw (R5) -- (V51); \draw (V51)--(V52); \draw (V52)--(V53);

    \draw (R6) -- (V61); \draw (V61)--(V62); \draw (V61)--(V63);

    \draw (R7) -- (V71); \draw (V71)--(V72); \draw (R7)--(V73);

    \draw (R8) -- (V81); \draw (R8)--(V82); \draw (V82)--(V83);

    \draw (R9) -- (V91); \draw (R9)--(V92); \draw (V91)--(V93);

    \draw (R10) -- (V101); \draw (R10)--(V102); \draw (R10)--(V103);

    \draw[black, thick,-|] (R1) -- (0,1) node[below] {r};
    \draw[black, thick,-|] (R2) -- (6,1) node[below] {r}; 
    \draw[black, thick,-|] (R3) -- (12,1) node[below] {r};
    \draw[black, thick,-|] (R4) -- (16,1) node[below] {r}; 
    \draw[black, thick,-|] (R5) -- (-4,-9) node[below] {r}; 
    \draw[black, thick,-|] (R6) -- (2,-9) node[below] {r}; 
    \draw[black, thick,-|] (R7) -- (8,-9) node[below] {r}; 
    \draw[black, thick,-|] (R8) --  (14,-9)  node[below] {r}; 
    \draw[black, thick,-|] (R9) --  (20,-9)  node[below] {r}; 
    \draw[black, thick,-|] (R10) --  (26,-9)  node[below] {r}; 
       
\end{scope}
\end{tikzpicture}

    \end{center}
    \caption{Recursive trees with up to four non-root vertices, where we have added a root $r$. Notice that these are non-embedded trees, meaning that the order of children at each vertex does not matter.}
    \label{fig:rectrees}
\end{figure}

Recursive trees arise naturally when one solves the differential equation:
\begin{equation}\label{eq:elemeq}
 \frac{d\phi}{dt} = f(\phi) \;, \qquad \phi(0) = 0 \; ,
\end{equation}
where for convenience we choose the initial condition as $0$ without loss of generality.

Let us use the labels $1,\dots p$ for the vertices. We denote $d_T(i)$ the degree of the vertex $i$ in the tree $T$, that is the number of edges of $T$ incident to $i$, and $d_T(1)$ counts also the edge connecting the vertex $1$ to the root vertex $r$. 
We denote the set of recursive trees by $\pmb{T}$ and for $p\ge 1$ one has by induction that:
\[
\frac{d^{p}\phi}{dt^p}= \sum_{\substack{  { T\in \pmb{T} , } \\ { V^{\rm int}(T) = p}  } } \prod_{i \in V^{\rm int}(T)} f^{(d_T(i)-1)}(\phi) \;,
\]
which is the Fa\`a di Bruno formula, see \cite[Sec.~6]{Faris_19}.
We stress that the product runs over the non-root vertices of $T$ with labels $i=1,\dots,p$.
Therefore, the solution of the differential equation~\eqref{eq:elemeq} for $t\ge 0$ can be written in a Taylor series as: 
\begin{align}\label{eq:solrectree}
 \phi(t)  & =  \sum_{p\ge 1} \frac{t^p}{p!} \sum_{\substack{  { T\in \pmb{T} , } \\ { V^{\rm int}(T) = p}  } } \prod_{i \in V^{\rm int}(T)}  f^{(d_T(i)-1)}  \big{|}_0  
\\
&  =    t f \big{|}_{0} + \frac{t^2}{2}( f' f )\big{|}_{0} + \frac{t^3}{3!}( f'' ff + f' f' f )\big{|}_{0}   + \frac{t^4}{4!} ( f''' fff +   f' f''  ff + 3  f'' f' ff + f' f' f' f )\big{|}_{0}  + \dots \; . \nonumber
\end{align}
We also observe that for $p\ge 1$ we have:
\begin{equation}\label{eq:simplex_int}
\frac{t^p}{p!} = \int_{0}^{t} dt_1 \int_{0}^{t_1} \dots \int^{t_{p-1}}_0 dt_{p} \;,
\end{equation}
which is an integral over a time parameter $t_i$ per non-root vertex $i=1,\dots p$ in the tree, constrained to respect the total order of the vertices $t_i > t_j$ if $i<j$ and the last time parameter $t$, that is the upper limit of the integral over $t_1$, is a fixed (not integrated) external time. We associate the external time $t$ to the root vertex.

\paragraph{Rooted combinatorial trees.} Rooted combinatorial trees $(\mathfrak{t},r)$ are trees $\mathfrak{t}$ having an univalent root vertex $r$. We order the edges as $(v,w)\in E(\mathfrak{t})$ where $v$ lies on the path from $w$ to the root (that is $v$ is the parent of $w$), and we denote $V^{\rm int}(\mathfrak{t})= V(\mathfrak{t}) \setminus \{r\}$ the set of internal non-root, vertices of $\mathfrak{t}$.

A morphism between two rooted combinatorial trees $(\mathfrak{t}_1,r_1)$ and $(\mathfrak{t}_2,r_2)$ is a bijection $f:V(\mathfrak{t}_1) \to V(\mathfrak{t}_2)$ such that $f(r_1) =r_2$ and $ (v_1,w_1) \in E(\mathfrak{t}_1) \Leftrightarrow (f(v_1) ,f(w_1)) \in E(\mathfrak{t}_2)$. The rooted combinatorial trees have non trivial automorphism groups and 
unlabeled rooted combinatorial trees, depicted in Fig.~\ref{fig:rootunlabtree},
are the automorphism classes of rooted combinatorial trees.

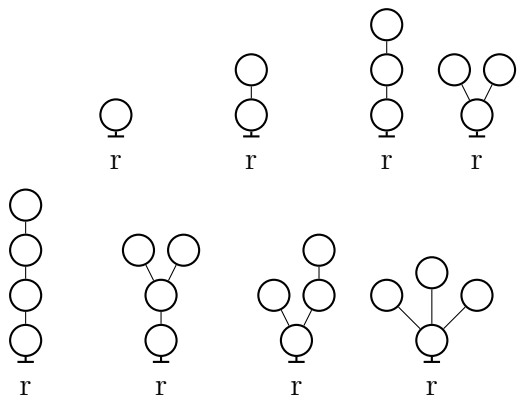
\begin{figure}[ht!]
\begin{center}
\begin{tikzpicture}[scale=0.3]

\begin{scope}[every node/.style={circle, thick,draw,font=\tiny,minimum height=1em}]
     \node (R1) at (0,2) {};
 
     \node (R2) at (6,2) {};
     \node (V21) at (6,4) {};

     \node (R3) at (12,2) {};
     \node (V31) at (12,4) {};
     \node (V32) at (12,6) {};

     \node (R4) at (16,2) {};
     \node (V41) at (17,4) {};
     \node (V42) at (15,4) {};

     \node (R5) at (-4,-8) {};
     \node (V51) at (-4,-6) {};
     \node (V52) at (-4,-4) {};
    \node (V53) at (-4,-2) {} ;

     \node (R6) at (2,-8) {};
     \node (V61) at (2,-6) {};
     \node (V62) at (3,-4) {};
    \node (V63) at (1,-4) {} ;

     \node (R7) at (8,-8) {};
     \node (V71) at (9,-6) {};
     \node (V72) at (9,-4) {};
    \node (V73) at (7,-6) {} ;

     \node (R10) at (14,-8) {};
     \node (V101) at (12,-6) {};
     \node (V102) at (14,-5) {};
    \node (V103) at (16,-6) {} ;
    
\end{scope}

\begin{scope}[>={Stealth[black]},
              every node/.style={fill=white,circle},
              every edge/.style={draw=black,thick}]

    \draw (R2) -- (V21);

   \draw (R3) -- (V31);    \draw (V31) -- (V32);

   \draw (R4) -- (V41);    \draw (R4) -- (V42);

    \draw (R5) -- (V51); \draw (V51)--(V52); \draw (V52)--(V53);

    \draw (R6) -- (V61); \draw (V61)--(V62); \draw (V61)--(V63);

    \draw (R7) -- (V71); \draw (V71)--(V72); \draw (R7)--(V73);

    \draw (R10) -- (V101); \draw (R10)--(V102); \draw (R10)--(V103);

    \draw[black, thick,-|] (R1) -- (0,1) node[below] {r};
    \draw[black, thick,-|] (R2) -- (6,1) node[below] {r}; 
    \draw[black, thick,-|] (R3) -- (12,1) node[below] {r};
    \draw[black, thick,-|] (R4) -- (16,1) node[below] {r}; 
    \draw[black, thick,-|] (R5) -- (-4,-9) node[below] {r}; 
    \draw[black, thick,-|] (R6) -- (2,-9) node[below] {r}; 
    \draw[black, thick,-|] (R7) -- (8,-9) node[below] {r}; 
    \draw[black, thick,-|] (R10) --  (14,-9)  node[below] {r};
    
\end{scope}
\end{tikzpicture}

    \end{center}
    \caption{Unlabeled rooted combinatoiral trees with up to four non-root vertices.}
    \label{fig:rootunlabtree}
\end{figure}

We denote $\mathfrak{T}$ the set of unlabeled rooted trees with at least one non-root vertex.  Taking into account that every recursive tree belongs to one equivalence class of rooted trees, we conclude that there is a many to one mapping from recursive trees to rooted unlabeled trees. We denote $\alpha(\mathfrak{t})$ the number of recursive trees which correspond to the rooted unlabeled tree $\mathfrak{t}$: for the examples in Fig.~\ref{fig:rootunlabtree}
for instance, $\alpha(\mathfrak{t}) $ takes values from left to right and top to bottom $1,1,1,1,1,1,3,1$, as can be seen by comparing with Fig.~\ref{fig:rectrees}. 
The recursive tree expansion in Eq.~\eqref{eq:solrectree} can be reorganized in terms of rooted unlabeled trees:
\begin{equation}\label{eq:solabstree}
\phi(t) = \sum_{\mathfrak{t} \in \mathfrak{T}} \frac{t^{|V^{\rm int}(\mathfrak{t})|} }{
|V^{\rm int}(\mathfrak{t})|!} \; \alpha(\mathfrak{t}) \; \prod_{v\in V^{\rm int}(\mathfrak{t}) } f^{(d_{ \mathfrak{t}}(v)-1)} \Big{|}_{0}\; .
\end{equation}
which can be further written in terms of Butcher elementary differentials
\cite{Butcher_2016}, see also \cite[Sec.~3]{Brouder_2000}.

\paragraph{Rooted plane trees.} We now consider again the same differential equation \eqref{eq:elemeq}:
\[
\frac{d\phi }{dt}(t) = f (\phi(t) ) \;,\qquad \phi(0) = 0 \; , \qquad
f(\phi) = f_0 + \sum_{q\ge 1} f_q \; \phi^q  \; , \qquad 
\]
where $f_q =  \frac{1}{q!} f^{(q)}\big{|}_{0} $, but this time we solve it by integrating:
\[
\phi(t) = \int_{0}^{t} dt_1 \; f(\phi(t_1)) = t f_0 + \sum_{q\ge 1} f_q  
\int_{0}^{t}   dt_1 \; \phi(t_1)^q \;,
\]
and using repeated substitutions.
A \emph{rooted plane tree} is a rooted combinatorial map whose underlying graph is a tree, i.e. a graph that is connected and has no cycle; see \cite[Sec.~I.5]{FlajoletSedgewick2009} for background on plane trees. The solution writes as a sum over such trees, as depicted in Fig.~\ref{fig:rootplantree}.

\begin{figure}[ht!]
\begin{center}
\begin{tikzpicture}[scale=0.3]

\begin{scope}[every node/.style={circle, thick,draw,font=\tiny,minimum height=1em}]
     \node (R1) at (0,2) {};
 
     \node (R2) at (6,2) {};
     \node (V21) at (6,4) {};

     \node (R3) at (12,2) {};
     \node (V31) at (12,4) {};
     \node (V32) at (12,6) {};

     \node (R4) at (16,2) {};
     \node (V41) at (17,4) {};
     \node (V42) at (15,4) {};

     \node (R5) at (-4,-8) {};
     \node (V51) at (-4,-6) {};
     \node (V52) at (-4,-4) {};
    \node (V53) at (-4,-2) {} ;

     \node (R6) at (2,-8) {};
     \node (V61) at (2,-6) {};
     \node (V62) at (3,-4) {};
    \node (V63) at (1,-4) {} ;

     \node (R7) at (8,-8) {};
     \node (V71) at (9,-6) {};
     \node (V72) at (9,-4) {};
    \node (V73) at (7,-6) {} ;

    \node (R8) at (14,-8) {};
     \node (V81) at (15,-6) {};
     \node (V82) at (13,-6) {};
    \node (V83) at (13,-4) {} ;

     \node (R10) at (20,-8) {};
     \node (V101) at (18,-6) {};
     \node (V102) at (20,-5) {};
    \node (V103) at (22,-6) {} ;
    
\end{scope}

\begin{scope}[>={Stealth[black]},
              every node/.style={fill=white,circle},
              every edge/.style={draw=black,thick}]

    \draw (R2) -- (V21);

   \draw (R3) -- (V31);    \draw (V31) -- (V32);

   \draw (R4) -- (V41);    \draw (R4) -- (V42);

    \draw (R5) -- (V51); \draw (V51)--(V52); \draw (V52)--(V53);

    \draw (R6) -- (V61); \draw (V61)--(V62); \draw (V61)--(V63);

    \draw (R7) -- (V71); \draw (V71)--(V72); \draw (R7)--(V73);

    \draw (R8) -- (V81); \draw (R8)--(V82); \draw (V82)--(V83);

    \draw (R10) -- (V101); \draw (R10)--(V102); \draw (R10)--(V103);

    \draw[black, thick,-|] (R1) -- (0,1) node[below] {r};
    \draw[black, thick,-|] (R2) -- (6,1) node[below] {r}; 
    \draw[black, thick,-|] (R3) -- (12,1) node[below] {r};
    \draw[black, thick,-|] (R4) -- (16,1) node[below] {r}; 
    \draw[black, thick,-|] (R5) -- (-4,-9) node[below] {r}; 
    \draw[black, thick,-|] (R6) -- (2,-9) node[below] {r}; 
    \draw[black, thick,-|] (R7) -- (8,-9) node[below] {r}; 
    \draw[black, thick,-|] (R8) --  (14,-9)  node[below] {r}; 
    \draw[black, thick,-|] (R10) --  (20,-9)  node[below] {r};
    
\end{scope}
\end{tikzpicture}

    \end{center}
    \caption{Rooted plane trees with up to four non-root vertices.}
    \label{fig:rootplantree}
\end{figure}
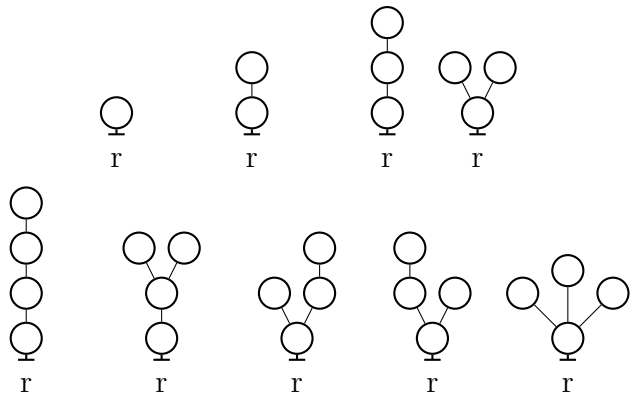

For every vertex $v=(v^1,\dots, v^{d_{\cal T}(v)} )$ in the plane tree, we denote $v^1$ the half-edge pointing towards the root $r$. The children of the vertex are then ordered, say counterclockwise, as the first child hooked to $v^2$, the second child hooked to $v^3$ and so on up to the last child hooked to $ v^{d_{\cal T}(v)} $. We order the edges of the plane tree as $(v^q,w^1)\in E(\cal T)$  where $v$ is the parent of $w$. 
Denoting the set of rooted plane trees\footnote{The number of rooted plane trees with $E$ edges is the Catalan number $C_{E} = \frac{1}{E+1} \binom{2E}{E}$ and can be recursively computed by cutting the first edge one encounters when walking counterclockwise around the tree starting from the root.}  with root $r$ and time parameter $t_r=t$ associated to the root by $\pmb{\cal T}^{(r|t)}$ we have:
\begin{equation}\label{eq:solplantree}
\begin{split}
\phi(t) =  \sum_{ {\cal T} \in \pmb{\cal T}^{(r|t)} } 
  \left( \prod_{v\in V^{\rm int}( {\cal T} ) } f_{d_{\cal T} (v) -1 } \right)
  \int_{0}^{\infty} \prod_{v\in   V^{\rm int}( {\cal T} )  }  dt_v
  \prod_{(v^q,w^1) \in E( {\cal T} ) } \theta(t_v-t_w)  \; ,
\end{split}
\end{equation}
where
\[
\theta(t) = \mathbf{1}_{t\geq 0} \;,
\]
is the Heaviside function 
and $ E( {\cal T} )$ includes the edge of ${\cal T}$ hooked to the root, that is $(r,w^1)$ for some $w^1$, such that the product of Heaviside functions constrains all the time parameters to be lower than the time of the root $t_r=t$.
Here and below we use the shorthand
\[
\int_{a}^{b} \prod_{i\in I} dt_i
= \int_{(a,b)^{|I|}} \prod_{i\in I} dt_i \;,
\]
for a finite set $I$.

 The time integral:
\begin{equation}\label{eq:timeintegral}
  \int_{0}^{\infty} \prod_{v\in   V^{\rm int}( {\cal T} )  }  dt_v
  \prod_{(v^q,w^1) \in E( {\cal T} ) } \theta(t_v-t_w)  \; ,
\end{equation}
can be understood as follows. The tree ${\cal T}$ induces a strict partial order among the vertices $v < w$ if $v$ belongs to the path from $w$ to the root $r$. We integrate one time parameter per non-root vertex, where the time parameters are constraint to respect the tree partial order, $t_v \ge t_w$ if $v < w$, as imposed by the product of Heaviside functions.
The product over the edges includes the edge connected to the root vertex $r$ with time parameter $t_r=t$, which is not integrated and is a maximal element in the partial order $t \ge t_v,\forall v$. We note that the time integral is insensitive to the embedding of ${\cal T}$.

The time integrals associated to the trees displayed in Fig.~\ref{fig:rootunlabtree} are (ordered left to right and top to bottom):
\[
\begin{split}
 \int_0^{t} dt_1  = t \;, & \qquad \int_0^{t} dt_1 \int_{0}^{t_1} dt_2  = \frac{t^2}{2!} \;, \crcr
 \int_0^{t} dt_1 \int_0^{t_1} dt_2 \int_0^{t_2} dt_3 = \frac{t^3}{3!} \;, & \qquad  \int_0^{t} dt_1 \int_0^{t_1} dt_2 \int_0^{t_1} dt_3  = \frac{t^3}{3} \;, \qquad \crcr
  \int_0^{t} dt_1\int_0^{t_1}dt_2 \int_{0}^{t_2}dt_3 \int_0^{t_3} dt_4  = \frac{t^4}{4!} \;, & \qquad
\int_0^{t} dt_1 \int_0^{t_1} dt_2 \int_{0}^{t_2} dt_3 \int_0^{t_2}dt_4  = \frac{t^4}{4 \cdot 3 } \;, \crcr
\int_0^{t} dt_1 \int_0^{t_1} dt_2 \int_0^{t_1} dt_3 \int_0^{t_3}dt_4  = \frac{t^4}{4\cdot 2} \;, & 
\qquad 
\int_0^{t} dt_1 \int_0^{t_1} dt_2  \int_0^{t_1} dt_3  \int_0^{t_1} dt_4  =\frac{t^4}{4} \; ,
\end{split}
\]
where two of the plane trees with four vertices are identical up to re-embedding hence have the same time integral. Eq.~\eqref{eq:solplantree} becomes at first orders (recall that $f_0=f$, $f_1=f'$, $f_2 = \frac{1}{2!}f''$ and $f_3 = \frac{1}{3!}f'''$):
\[
\begin{split}
  \phi(t) 
 = & t \, f + \frac{t^2}{2!} \, f'f + \frac{t^3}{3!} \,  f'f'f + \frac{t^3}{3}\;\frac{1}{2!}f'' f^2 \crcr
& + \frac{t^4}{4!} \, f'f'f'f + \frac{t^4}{4 \cdot 3} \, f' \frac{1}{2!} f'' ff + 
  \frac{t^4}{4 \cdot 2} \,
2 \frac{1}{2!} f'' f' ff + \frac{t^4}{4} \,  \frac{1}{3!}f'''fff  + \dots\;,
\end{split}
\]
reproducing the series in \eqref{eq:solrectree}. 
One can directly obtain \eqref{eq:solplantree}  from \eqref{eq:solrectree} by:
\begin{itemize}
    \item[-] replacing $f^{(q)}$ by $q!f_q$ for every vertex and counting one $f_q$ for each of the $q!$ distinct ways to embed the $q$ children of that vertex, passing thus from a sum over recursive trees to a sum over rooted plane trees with a total order among the time parameters of the vertices;
    \item[-] rewriting the factor $\frac{t^p}{p!}$ in \eqref{eq:solrectree}
    as a time ordered integral in \eqref{eq:simplex_int}
    and adding up together the terms corresponding to all the plane trees with total orders among the time parameters corresponding to the same plane tree with only the tree partial order among the time parameters.
\end{itemize}

\begin{remark}
Comparing \eqref{eq:solabstree} with \eqref{eq:solplantree} we conclude that:
\[
\frac{1}{| V^{\rm int}(\mathfrak{t} ) | !}\alpha(\mathfrak{t}) = 
\frac{ N(\mathfrak{t} ) }{\prod_{v\in V^{\rm int}( \mathfrak{t}) } (d_{\mathfrak{t}}(v)- 1 )! }  \;
\int_{0}^{\infty} \prod_{v\in   V^{\rm int}( \mathfrak{t} )  }  dt_v
  \prod_{(v,w) \in E( \mathfrak{t} ) } \theta(t_v-t_w) \Big{|}_{t_r=1}  \;,
\]
where $ N(\mathfrak{t} )$ denotes the number of rooted plane trees corresponding to the unlabeled rooted combinatorial tree $\mathfrak{t} $.
This provides a way to compute $\alpha(\mathfrak t)$.
We will not use below analogues of \eqref{eq:solabstree} and instead focus on recursive vs. plane tree representations, i.e. \eqref{eq:solrectree} vs. \eqref{eq:solplantree}.
\end{remark}

The example discussed above is a special case of a more general result that we will need later on:
\begin{lemma} \label{lemma:trees}
Let $\mathfrak{t}$ denote an unlabeled rooted combinatorial tree, $\pmb{\cal T}( \mathfrak{t})$ the set of rooted plane trees corresponding to it and $\pmb{T}(\mathfrak{t})$ the set of recursive trees corresponding to it (where we assign to the univalent root the label $0$).
Consider a function $F$ that may depend on $\mathfrak{t}$ and on all its vertex times but is invariant under rooted tree isomorphisms, in particular under permutations of isomorphic branches; thus it is independent of the plane embedding and of the recursive labeling.
In the recursive trees, the labels are totally ordered and the corresponding times form the chain $t_0\geq t_1\geq\cdots$.
We then have (with $t_r = t_0$):
\[
\begin{split}
&  \sum_{ {\cal T} \in \pmb{\cal T}( \mathfrak{t}) }
 \prod_{v\in V^{\rm int}( {\cal T} ) } \frac{1}{ ( d_{\cal T}(v)-1) !}   \int_{-\infty}^{\infty} \prod_{v\in   V^{\rm int}( {\cal T} )  }  dt_v
  \prod_{(v^q,w^1) \in E( {\cal T} ) } \theta(t_v-t_w) \;\; F(\mathfrak{t} ; \{ t_v \}_{v\in V({\cal T})}  )
  \crcr
& = \sum_{ T\in \pmb{T}(\mathfrak{t}) }
 \int_{-\infty}^{t_0} dt_1 \dots  \int_{-\infty}^{t_{|  V^{\rm int}( T ) |-1}} dt_{|  V^{\rm int}( T ) | }  \; \; F(\mathfrak{t}  ;\{t_i\}_{i\in V(T)} ) \; .
\end{split}
\]
\end{lemma}
\begin{proof}
The two sides of the equality are proportional to a time integral corresponding to $\mathfrak{t}$. 
On the one hand, as the integral is independent of the embedding we have: 
\[
\begin{split}
&  \sum_{ {\cal T} \in \pmb{\cal T}( \mathfrak{t}) }
 \prod_{v\in V^{\rm int}( {\cal T} ) } \frac{1}{ ( d_{\cal T}(v)-1) !}   \int_{-\infty}^{\infty} \prod_{v\in   V^{\rm int}( {\cal T} )  }  dt_v
  \prod_{(v^q,w^1) \in E( {\cal T} ) } \theta(t_v-t_w) \;\; F(\mathfrak{t} ;\{ t_v \}_{v\in V({\cal T})} )
  \crcr
& = \frac{ N(\mathfrak{t} ) }{\prod_{v\in V^{\rm int}( \mathfrak{t}) } (d_{\mathfrak{t}}(v)- 1 )! }   \int_{-\infty}^{\infty} \prod_{v\in   V^{\rm int}( \mathfrak{t} )   }  dt_v
  \prod_{(v^q,w^1) \in E( \mathfrak{t} ) } \theta(t_v-t_w) \;\; F(\mathfrak{t} ; \{ t_v \}_{v\in V(\mathfrak{t})} ) \;,
  \end{split}
\]
where $t_r = t_0$ and $ N(\mathfrak{t} ) = |  \pmb{\cal T}( \mathfrak{t})  | $. On the other hand, by adding all the time integrals of the recursive trees corresponding to $\mathfrak{t}$ one obtains an integral respecting the tree partial order but supplemented by a total order of the times of the lowest vertices in isomorphic branches (see some examples below), hence: 
\[
\begin{split}
& \sum_{ T\in \pmb{T}(\mathfrak{t}) }
 \int_{-\infty}^{t_0} dt_1 \dots  \int_{-\infty}^{t_{|  V^{\rm int}( T ) |-1}} dt_{|  V^{\rm int}( T ) | }  \; \; F(\mathfrak{t}  ; \{t_i\}_{i\in V(T)} )  \crcr
& =    \frac{1}{  S(\mathfrak{t} ) }
 \int_{-\infty}^{\infty} \prod_{v\in   V^{\rm int}( \mathfrak{t} )   }  dt_v
  \prod_{(v^q,w^1) \in E( \mathfrak{t} ) } \theta(t_v-t_w) \;\; F(\mathfrak{t} ; \{ t_v \}_{v\in V(\mathfrak{t})} )  
 \; ,
\end{split}
\]
with $S(\mathfrak{t})$ the order of the group of (root-preserving) automorphisms of $\mathfrak{t}$. The latter can be computed recursively as
$S(\mathfrak{t}) = \prod_{\rho} n_\rho ! S(\mathfrak{t}_\rho)^{n_\rho} $, where $n_\rho$ is the number of times the same (up to isomorphisms) branch $\mathfrak{t}_\rho$ is obtained when cutting $\mathfrak{t}$ at the successor of the root. On the other hand, we claim that 
\begin{equation}\label{eq:SN}
 \frac{1}{  S(\mathfrak{t} ) } =    \frac{ N(\mathfrak{t} ) }{\prod_{v\in V^{\rm int}( \mathfrak{t}) } (d_{\mathfrak{t}}(v)- 1 )! }    \; ,
\end{equation}
which concludes the proof.
To see why \eqref{eq:SN} is true, we argue that both $ \prod_{v\in V^{\rm int}( \mathfrak{t}) } (d_{\mathfrak{t}}(v)- 1 )!  $ and $  S(\mathfrak{t} )  N(\mathfrak{t} )$ count the number of plane embeddings of a labeled combinatorial tree corresponding to $\mathfrak{t}$.
This follows from, on the one hand, the fact that a plane
embedding is obtained by independently ordering the
$d_{\mathfrak{t}}(v)-1$ children of every non-root vertex $v$, giving
$\prod_{v\in V^{\rm int}(\mathfrak{t})}(d_{\mathfrak{t}}(v)-1)!$
orderings. On the other hand, the (root-preserving) automorphism group of $\mathfrak{t}$,
of cardinality $S(\mathfrak{t})$, acts freely on these orderings, and
its orbits are precisely the distinct rooted plane trees corresponding
to $\mathfrak{t}$.  Consequently, $ \prod_{v\in V^{\rm int}( \mathfrak{t}) } (d_{\mathfrak{t}}(v)- 1 )! =  S(\mathfrak{t} )  N(\mathfrak{t} )$, proving \eqref{eq:SN}.
\end{proof}

For example, the rooted unlabeled trees with one, two, three and four internal vertices are depicted in Fig.~\ref{fig:rootunlabtree}, and the corresponding recursive and plane trees in  Fig.~\ref{fig:rectrees} and Fig.~\ref{fig:rootplantree}, respectively. Labeling the trees in Fig.~\ref{fig:rootunlabtree} left to right and top to bottom as $(1)$, $(2)$, $(3,1)$, $(3,2)$, $(4,1)$, $(4,2)$, $(4,3)$, $(4,4)$, we have
trivially:
\[
\begin{split}
&   \int_{t_v \le t_0} F^{(1)}  = \int_{t_1 \le t_0} F^{(1)} \;, \qquad \int_{t_{v'} \le t_v \le t_0} F^{(2)}  =  \int_{t_2\le t_1\le t_0} F^{(2)} \;, \crcr
& \int_{t_{v''} \le t_{v'} \le t_{v} \le t_0} F^{(3,1)}  = 
 \int_{t_3\le t_2\le t_1\le t_0} F^{(3,1)} \;, \qquad
 \int_{ t_{v'''} \le  t_{v''} \le  t_{v'} \le t_v\le t_0} F^{(4,1)}  = \int_{ t_{v'''} \le  t_{v''} \le  t_{v'} \le t_v\le t_0} F^{(4,1)} \;,
\end{split}
\]
and slightly less trivially,
\[
\begin{split}
&      \frac{1}{2!} \int_{t_{v'} , t_{v''} \le t_v \le t_0} F^{(3,2)}  =  \int_{t_3\le t_2\le t_1\le t_0} F^{(3,2)} \; , \;\; 
 \frac{1}{2!}
\int_{ t_{v'''}  , t_{v''} \le  t_{v'} \le t_v\le t_0}  
 F^{(4,2)}  = \int_{t_4\le t_3\le t_2 \le t_1\le t_0} F^{(4,2)} \;, \crcr
& \frac{1}{3!} 
 \int_{t_{v'},t_{v''},t_{v'''} \le t_v \le t_0}
 F^{(4,4)} = \int_{t_4\le t_3\le t_2 \le t_1\le t_0} F^{(4,4)} \;.
\end{split}
\]
Note that these examples have a non trivial order of the automorphism group (2! and 3! respectively), which is reproduced by the ordered time integrals of the corresponding recursive trees. For the tree $(4,3)$, reinstating the edge time differences, we get after careful relabeling of dummy variables:
\[
\begin{split}
  \frac{1}{2!} 2
 & \int_{ t_{v''} \le t_{v'} \le t_v ,\, t_{v'''} \le t_v  ,\, t_v \le t_0 }
 F^{(4,3) } (t_0-t_v,t_v-t_{v'},t_{v'}-t_{v''},t_v-t_{v'''}) \crcr
 & =   \int_{t_4\le t_3\le t_2 \le t_1\le t_0} F^{(4,3)}(t_0-t_1,t_1-t_2,t_2-t_3,t_1-t_4)  \crcr
 & \;\; + \int_{t_4\le t_3\le t_2 \le t_1\le t_0} F^{(4,3)}(t_0-t_1,t_1-t_3,t_3-t_4,t_1-t_2)  \crcr
 & \;\; +  \int_{t_4\le t_3\le t_2 \le t_1\le t_0} F^{(4,3)}(t_0-t_1,t_1-t_2,t_2-t_4,t_1-t_3)  
 \; .
\end{split}
\]

\paragraph{$q$-ary trees.} A special case of rooted plane trees are the $q$-ary trees, that is the rooted plane trees such that every non-root vertex either has no descendants (in which case it is called a leaf) or it has exactly $q$ descendants. The root vertex has always one descendant. Such trees are obtained when solving the differential equation for $f(\phi) = f_0 + f_q \phi^q$ with $q$ fixed. We denote the set of $q$-ary trees rooted at $r$ by $\pmb{\cal T}^{(r)}_{q-{\rm ary}}$ and this set includes the tree consisting in the root decorated by a leaf. The generating function of $q$-ary trees $F(g) = \sum_{k\ge 0} g^k N_k$, with $N_k$ the number of $q$-ary trees with $k$ vertices with $q$ descendants respects the Fuss--Catalan equation:
\[
F(g) = 1 + gF(g)^q  \; , \qquad N_k =\frac{1}{qk+1} \binom{qk+1}{k } = \frac{ ( qk) !}{k!  (qk-k+1 )! }\;,
\]
where $1$ counts the tree consisting in the root decorated by a leaf, hence with no internal $(q+1)$--valent vertex. At low orders we have
$N_0 = N_1 = 1$, $N_2 =  q$, $N_3  = \frac{q(3q-1)}{2}$
which for ternary trees $q=3$ become $N_0= N_1=1$, $N_2  = 3$ and $N_3 = 12$, as displayed in the Fig.~\ref{fig:tertree}.
The Fuss--Catalan numbers will only appear again in App.~\ref{app:loworderstoc}.

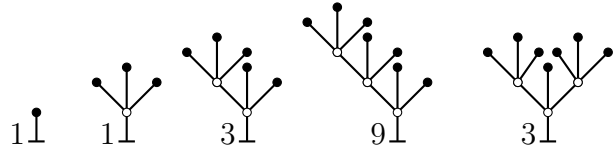
\begin{figure}[ht!]
\begin{center}

\begin{tikzpicture}[scale=0.4]
\draw[black, thick,|-] (0,0) -- (0,1)
node[below left] {1};

\filldraw[black] (0,1) circle (4pt);

\draw[black, thick,|-] (3,0) -- (3,1)
node[below left] {1};

\draw[black, thick] (3,1) -- (2,2);
\filldraw[black] (2,2) circle (4pt);
\draw[black, thick] (3,1) -- (3,2.5);
\filldraw[black] (3,2.5) circle (4pt);
\draw[black, thick] (3,1) -- (4,2);
\filldraw[black] (4,2) circle (4pt);

\filldraw[fill=white] (3,1) circle (4pt);

\draw[black, thick,|-] (7,0) -- (7,1)
node[below left] {3};

\draw[black, thick] (7,1) -- (6,2);
\filldraw[black] (6,2) circle (4pt);
\draw[black, thick] (7,1) -- (7,2.5);
\filldraw[black] (7,2.5) circle (4pt);
\draw[black, thick] (7,1) -- (8,2);
\filldraw[black] (8,2) circle (4pt);

\draw[black, thick] (6,2) -- (5,3);
\filldraw[black] (5,3) circle (4pt);
\draw[black, thick] (6,2) -- (6,3.5);
\filldraw[black] (6,3.5) circle (4pt);
\draw[black, thick] (6,2) -- (7,3);
\filldraw[black] (7,3) circle (4pt);

\filldraw[fill=white] (7,1) circle (4pt);
\filldraw[fill=white] (6,2) circle (4pt);

\draw[black, thick,|-] (12,0) -- (12,1)
node[below left] {9};

\draw[black, thick] (12,1) -- (11,2);
\filldraw[black] (11,2) circle (4pt);
\draw[black, thick] (12,1) -- (12,2.5);
\filldraw[black] (12,2.5) circle (4pt);
\draw[black, thick] (12,1) -- (13,2);
\filldraw[black] (13,2) circle (4pt);

\draw[black, thick] (11,2) -- (10,3);
\filldraw[black] (10,3) circle (4pt);
\draw[black, thick] (11,2) -- (11,3.5);
\filldraw[black] (11,3.5) circle (4pt);
\draw[black, thick] (11,2) -- (12,3);
\filldraw[black] (12,3) circle (4pt);

\draw[black, thick] (10,3) -- (9,4);
\filldraw[black] (9,4) circle (4pt);
\draw[black, thick] (10,3) -- (10,4.5);
\filldraw[black] (10,4.5) circle (4pt);
\draw[black, thick] (10,3) -- (11,4);
\filldraw[black] (11,4) circle (4pt);

\filldraw[fill=white] (12,1) circle (4pt);
\filldraw[fill=white] (11,2) circle (4pt);
\filldraw[fill=white] (10,3) circle (4pt);

\draw[black, thick,|-] (17,0) -- (17,1)
node[below left] {3};

\draw[black, thick] (17,1) -- (16,2);
\filldraw[black] (16,2) circle (4pt);
\draw[black, thick] (17,1) -- (17,2.5);
\filldraw[black] (17,2.5) circle (4pt);
\draw[black, thick] (17,1) -- (18,2);
\filldraw[black] (18,2) circle (4pt);

\draw[black, thick] (16,2) -- (15,3);
\filldraw[black] (15,3) circle (4pt);
\draw[black, thick] (16,2) -- (16,3.5);
\filldraw[black] (16,3.5) circle (4pt);
\draw[black, thick] (16,2) -- (16.7,3);
\filldraw[black] (16.7,3) circle (4pt);

\draw[black, thick] (18,2) -- (17.3,3);
\filldraw[black] (17.3,3) circle (4pt);
\draw[black, thick] (18,2) -- (18,3.5);
\filldraw[black] (18,3.5) circle (4pt);
\draw[black, thick] (18,2) -- (19,3);
\filldraw[black] (19,3) circle (4pt);

\filldraw[fill=white] (17,1) circle (4pt);
\filldraw[fill=white] (16,2) circle (4pt);
\filldraw[fill=white] (18,2) circle (4pt);

\end{tikzpicture}
\end{center}
\caption{Ternary trees. We have grouped together the ternary trees that differ only by the embedding and we have chosen a representative embedded tree per class.}\label{fig:tertree}
\end{figure}

\subsection{Gaussian measures}\label{sec:Gauss}

A normalized centered Gaussian measure $\diff \mu_{C} (\phi)$ with covariance $C$ can be defined by its moments, also known in QFT as $n$-point functions:
\begin{equation}\label{eq:Wickthm}
\begin{split}
  \Braket{\phi_{x^1} \dots \phi_{x^n}  } & =
  \int  \diff \mu_{C} (\phi) \;  \phi_{x^1} \dots \phi_{x^n} 
= \begin{cases}
      \sum_{P \in P_n } \prod_{\{i,j\}\in P} C_{x^i x^j} \quad  &(n \text{ even}),\\
      0 & (n \text{ odd}),   
      \end{cases} \; ,
\end{split}
\end{equation}
where, for $n$ even, $P_n$ is the set of
pairings of $n$ elements, that is partitions of $\{1,\ldots,n\}$ into 2-element subsets.
Formula \eqref{eq:Wickthm} is known as the Wick theorem.
Starting from the moments one can compute the generating function:
\begin{equation}\label{eq:ZJ}
    Z(J) \equiv \int \diff \mu_C(\phi) \; e^{J\phi} = e^{\frac{1}{2} J C J} \; .
\end{equation}
Conversely, if the bilinear form $Q(J_1,J_2)=J_1 C J_2$ is continuous and positive semi-definite, then $Z(i J)$ satisfies the hypothesis of the Bochner-Minlos theorem, hence it defines a unique cylindrical measure on the space of tempered distributions (e.g.\ \cite{GlimmJaffe,Chandra_phdthesis,Lodhia:2016fractional}).

The common way to write the Gaussian measure in the physics literature is 
\[
\diff \mu_C(\phi) =  [\diff \phi] \; \frac{1}{\mathcal{N} }  \; e^{ - \frac{1}{2} \phi C^{-1} \phi } \;.
\]
However, this expression is rather formal, because $[\diff \phi]$ is not a well defined measure, the normalization factor $\mathcal{N}$ usually involves the determinant of an unbounded operator and $C$ might not be invertible. The last point is especially salient.

A standard technique in constructive field theory \cite{BKAR,salmhofer2007renormalization} consists in representing a normalized Gaussian integral as a differential operator:
\[
\int \diff \mu_C(\phi) \; F(\phi) =  \left[ e^{\frac{1}{2} \delta_\phi C \delta_\phi}  F(\phi)\right]_{\phi=0} \; ,
\]
where we used the shorthand notation $ \delta_\phi C \delta_\phi \equiv \int_{xy} \delta_{\phi_x} C_{xy} \delta_{\phi_y} $. Indeed, denoting $J\phi = \int_x J_x \phi_x$ we have:
\[
\left[ e^{\frac{1}{2} \delta_\phi C \delta_\phi}  e^{J\phi}\right]_{\phi=0} 
 = \sum_{n\ge 0}\frac{1}{n!} \left( \frac{1}{2} \delta_\phi C \delta_\phi \right)^{2n}
  \frac{1}{(2n)!} (J\phi)^{2n} = \sum_{n\ge 0}\frac{1}{n!}  \;\frac{1}{2^n}
     (JCJ)^n = e^{\frac{1}{2}JCJ} \;,
\]
reproducing the generating function of the moments of the Gaussian. For arbitrary observables we have:
\[
\begin{split}
\left[ e^{\frac{1}{2} \delta_\phi C \delta_\phi} F (\phi)\right]_{\phi=0}  
& = \left[  F(\frac{\delta}{\delta J}) \left[ e^{\frac{1}{2} \delta_\phi C \delta_\phi} 
e^{J\phi } \right]_{\phi=0}  \right]_{J=0} = \left[ F(\frac{\delta}{\delta J}) e^{\frac{1}{2} JCJ} \right]_{J=0}
\crcr
& = \left[ F(\frac{\delta}{\delta J}) \int \diff\mu_C(\phi) \; e^{ J\phi}   \right]_{J=0} = \int \diff\mu_C(\phi) \;    F(\phi) \;,
\end{split}
\]
where we interpret $F(\frac{\delta}{\delta J})$ using a power series expansion of $F$.

\paragraph{Copies.} A normalized Gaussian integral can be computed using copies of the field \cite{BKAR}:
 \be\label{eq:copies}
 \left[ e^{\frac{1}{2} \delta_\phi C \delta_\phi}  F^{(1)}(\phi) F^{(2)}(\phi)\right]_{\phi=0}  = 
  \left[ e^{\frac{1}{2} (\delta_{\phi^1} +  \delta_{\phi^2} ) C (\delta_{\phi^1} +  \delta_{\phi^2} ) }  F^{(1)}(\phi^1)F^{(2)}(\phi^2)\right]_{\phi^1 = \phi^2=0} 
 \; ,
 \ee
where the field $\phi$ is replaced by two copies $\phi^1$ and $\phi^2$ which can be distributed $\phi\to \phi^1$ or $\phi\to \phi^2$ freely in the  insertions while, at the same time, in the differential operator we substitute $\delta_\phi \to \delta_{\phi^1} +  \delta_{\phi^2}$. 
In fact, for any differential operator $D(\delta_{\phi}) = \sum_p \int_{x_i} D_{x_1\dots x_p} \prod_{i=1}^p  \frac{\delta}{\delta \phi_{x_i}}$
we have:
\[
D(\delta_{\phi}) F^{(1)}(\phi) F^{(2)}(\phi) =  D(\delta_{\phi^1} + \delta_{\phi^2}) F^{(1)}(\phi^1) F^{(2)}(\phi^2) \Big{|}_{\phi^1 = \phi^2 = \phi}  \;,
\]
as
\[
\left( \prod_{i\in \{1,\dots q \} } \frac{\delta}{\delta \phi_{x_i} } \right) F^{(1)}(\phi) F^{(2)}(\phi)
= \sum_{I\subset \{1,\dots q \}} \left[ \left( \prod_{i\in I } \frac{\delta}{\delta \phi_{x_i} } \right)  F^{(1)}(\phi) \right] \left[ \left( \prod_{i\notin I } \frac{\delta}{\delta \phi_{x_i} } \right)  F^{(2)}(\phi) \right]  \;,
\]
while
\[
\begin{aligned}
\left( \prod_{i\in \{1,\dots q \} } \left( \frac{\delta}{\delta \phi^1_{x_i} }  + \frac{\delta}{\delta \phi^2_{x_i} }   \right) \right) F^{(1)}(\phi^1) F^{(2)}(\phi^2)
&= \left( \sum_{I\subset \{1,\dots q \}} \left( \prod_{i\in I } \frac{\delta}{\delta \phi^1_{x_i} } \right)   \left( \prod_{i\notin I } \frac{\delta}{\delta \phi^2_{x_i} } \right)  \right)  F^{(1)}(\phi^1) F^{(2)}(\phi^2)
\\
&= \sum_{I\subset \{1,\dots q \}} \left[ \left( \prod_{i\in I } \frac{\delta}{\delta \phi^1_{x_i} } \right)  F^{(1)}(\phi^1) \right] \left[ \left( \prod_{i\notin I } \frac{\delta}{\delta \phi^2_{x_i} } \right)  F^{(2)}(\phi^2) \right] \;.
\end{aligned}
\]
Notice that in the $(\phi^1,\phi^2)$ space, the Gaussian measure in \eqref{eq:copies} has a noninvertible covariance $C^{ij}=C$, for $i,j=1,2$, thus showing one advantage of working with a definition of Gaussian measure that does not rely on the invertibility of the covariance.

\subsection{The diffusion equation} \label{sec:green}

The operator $A$ in the quadratic part of the action \eqref{eq:action} will be kept generic. In practical applications we have in mind, it can be minus the Laplacian, $A=-\Delta$, it can be the Euclidean Klein-Gordon operator with nonvanishing mass, $A=-\Delta+m^2$, it can contain counterterms or it can be something else. 
We will assume that the covariance of the theory $C=A^{-1}$ has a convergent integral representation in terms of an appropriate ``heat kernel'' associated to the operator $A$, on an appropriate space (e.g.\ $L^2(\mathbb{R}^d)$ if $A=-\Delta+m^2$):
\[
  C_{xy} =  \left( \frac{1}{A}  \right)_{xy}  = \int_0^{\infty} dt  \left( e^{-tA} \right)_{xy} \;,
\]
where the heat kernel $e^{-tA }$ is defined by
\[
(\partial_t + A) \; e^{-tA } = 0 \;,
  \qquad
 \left( e^{-tA} \right)_{xy}\Big{|}_{t=0} = \delta_{xy}\; .
\]

\begin{remark}\label{rem:heat_kernel}
The heat kernel representation exists under fairly weak assumptions, see for example \cite{ABHN_2011, Davies_89, Engel_Nagel_2000}.
Suppose $A$ is a  closed, densely defined, self-adjoint, non-negative operator on a Hilbert space $\cal H$
(e.g.\ ${\cal H}=L^2(\R^d)$).
Then, by the spectral theorem (or more generally the Hille--Yosida theorem), $A$ generates a strongly continuous contraction semigroup
$\{e^{-tA}\}_{t\ge 0}$, and one can define the Green's function $A^{-1}$ via functional calculus.
If $A$ is \emph{strictly} positive, i.e.\ $A\geq c >0$, then $A^{-1}$ is bounded by $c^{-1}$ and the integral of the heat kernel representation converges absolutely in operator norm on $\cal H$.
In general however, $0$ may lie in the spectrum of $A$, in which case
$A^{-1}$ is an unbounded operator and may only exist as an inverse on a subspace of $\cal H$.
The heat-kernel formula still makes sense, but only on a suitable domain.

On a compact manifold for instance (e.g. a torus) with $A=-\Delta$ the Laplace--Beltrami operator, $\ker(A)$ consists of constants and one
recovers the usual Green's function by restricting to mean-zero functions.
On $\R^d$ with $A=-\Delta$, one has $\ker(A)=\{0\}$ in $L^2(\R^d)$ but $0$ belongs to the continuous
spectrum, so the heat kernel integral defines an unbounded operator with domain
\[
\mathrm{Dom}(A^{-1})
=
\Big\{ f \in L^2(\R^d)\, :\, \int d\xi \; |\xi|^{-4} |\hat f(\xi)|^2 < \infty\Big\}
= \dot H^{-2}(\R^d) \cap L^2(\R^d)
\]
where $\hat f$ is the Fourier transform of $f$ and $\dot H^{-2}(\R^d)$ is the homogeneous Sobolev space.

For $f\in\mathrm{Dom}(A^{-1})$ the heat-kernel representation
$
A^{-1}f=\int_0^\infty e^{-tA}f\,dt
$
holds as a Bochner integral.
As $\ker(A)$ is orthogonal to $\mathrm{Dom}(A^{-1})$, if for instance $A$ has a discrete spectrum, then $A^{-1}$ defined by the heat-kernel representation and extended as $0$ to $\ker(A)$
is the Moore--Penrose pseudo-inverse of $A$.
On $\ker(A)^\perp$, the spectral theorem also implies $e^{-uA}\to 0$ strongly as $u\to\infty$.
\end{remark}

In the case $ A =  - \Delta $ we have the usual heat kernel:
\[
 \left( e^{- u A} \right)_{xy}
 = \frac{1}{ ( 4 \pi u )^{d/2}} e^{-\frac{|x-y|^2}{4 u}} \;, 
\]
while for the Laplacian on the torus we get additional contributions from the winding modes. If instead $A = -\Delta  + \Omega^2 x_\mu x^\mu$ then
\[
 \left( e^{- u A} \right)_{xy}= \frac{\Omega^{d/2} }{(4\pi \sinh (\Omega u) )^{d/2}}
  e^{- \frac{\Omega \cosh ( \Omega u) }{4 \sinh(\Omega u)} (x^2+y^2) -\frac{\Omega}{2 \sinh (\Omega u) } xy }\;,
\]
is the Mehler kernel.

Our treatment below does not depend on the details of the operator $A$, beyond the existence of a heat kernel representation of its inverse.

\paragraph{Heat diffusion and Green's function.} Associated to the operator $A$ we have a generalized inhomogeneous heat diffusion equation:
\[
 \partial_t H_{(t,x)}  = - \int_y A_{xy} H_{(t,y)}  + f(t,x) \;, \qquad 
 H_{(t_{\rm in},x)} = H_{\rm in}(x) \; ,
\]
which is solved in terms of the Green function of the diffusion operator:
\begin{equation}\label{eq:freeC}
\left(  \frac{1}{\partial_t + A } \right)_{(t,x) (t',x')} = \theta(t-t') \left( e^{- (t-t') A} \right)_{xx'} \equiv \pmb{C}_{(t,x)(t',x')}  \; .
\end{equation}
We note that, due to the presence of the Heaviside function, the Green function propagates forward in fictitious time. The solution of the heat diffusion equation writes:
\[
 H_{(t,x)} = \int_{t_{\rm in}}^t dt' \int_{x'} (e^{- (t-t') A} )_{xx'} \; f_{(t',x')}  +\int_{x'} (e^{- (t-t_{\rm in}) A} )_{xx'} H_{\rm in}(x')  \; ,
\]
and the second term drops out if we impose the initial condition at past infinity $t_{\rm in}\to -\infty$. 

\section{Perturbative expansions and their equivalence}
\label{sec:pert}

In both approaches, path integral and stochastic, the $n$-point functions can be computed as formal power series in the vertex kernels. These series are usually divergent and one needs to use non trivial resummation techniques in order to make sense of them rigorously. Here we will implicitly assume that a proper regularization is in place, and we will prove 
the order by order equality of the two formal perturbative expansions.

\subsection{The Feynman expansion of the path integral}

In the path integral approach the $n$-point functions are computed in the Feynman formal perturbative series in vertex kernels.\footnote{Usually one introduces coupling constants in front of the vertex kernels and speaks about an expansion in the coupling constants.}
We briefly review this here in the spirit of \cite{Rivasseau:1991ub,Gurau:2014vwa}.

\paragraph{The free models.} The $n$-point functions of a free model $S(\phi) =\frac{1}{2} \phi A \phi$ are just the moments of the Gaussian measure (see Sec.~\ref{sec:Gauss}) with covariance $C=A^{-1}$, and using the differential representation we can write:
\begin{equation}
\begin{split}
  \Braket{\phi_{x^1} \dots \phi_{x^n}  } & =
  \int  \diff \mu_{C} (\phi) \;  \phi_{x^1} \dots \phi_{x^n}   = \left[ e^{\frac{1}{2} \delta_\phi C \delta_\phi}  \prod_{i=1}^{n} \phi_{x^i} \right]_{\phi=0}  \; ,
\end{split}
\end{equation}
reproducing the Wick theorem \eqref{eq:Wickthm}.

\paragraph{The interacting models.} In the interacting case, the $n$-point functions are modified by the presence in the action \eqref{eq:action} of a generic interaction that can be written as in Eq.~\eqref{eq:geint}:
\begin{equation}
\begin{split}
  \Braket{\phi_{x^1} \dots \phi_{x^n}  } & =
  \frac{1}{Z(0)}\int  \diff \mu_{C} (\phi) \; e^{-V(\phi)} \;  \phi_{x^1} \dots \phi_{x^n}  
  = \frac{1}{Z(0)}\left[ e^{\frac{1}{2} \delta_\phi C \delta_\phi}  \; e^{-V(\phi)} \; \prod_{i=1}^{n} \phi_{x^i} \right]_{\phi=0}  \; .
\end{split}
\end{equation}
In order to deal with the exponential, one Taylor expands it:
\[
\Braket{\phi_{x^1} \dots \phi_{x^n}  } = \frac{ \int  \diff \mu_{C} (\phi) \; \sum_{p\ge 0}\frac{1}{p!} \Big(-V(\phi)  \Big)^p \;\phi_{x^1} \dots \phi_{x^n} }{\int  \diff \mu_{C} (\phi) \; \sum_{p\ge 0}\frac{1}{p!} \Big(-V(\phi)  \Big)^p}\;,
\]
and then one commutes (possibly illegally) the Taylor sum with the Gaussian integral. Computing the latter term by term using the Wick theorem in Eq.~\eqref{eq:Wickthm} we obtain the perturbative Feynman series.  

In this paper we treat this expansion somewhat formally, and make statements that only hold order by order in the power series, tacitly assuming that any ultraviolet divergence is properly regularized.\footnote{See Remark \ref{rem:renorm} in App.~\ref{app:loworderFeynm} for some comments on the troubles one typically encounters in this perturbative expansion.} 
We will now outline how the calculations proceed and how the structures reviewed in Sec.~\ref{sec:prelim} emerge in this context.

\paragraph{Rooted combinatorial maps and the Feynman expansion.}
For simplicity, we detail the contributions to the one point function in a cubic model:
\[
\Braket{\phi_{x^1}} = \frac{1}{Z(0)}\sum_{p\ge 0} \frac{1}{p! 3^p}\int d\mu_C(\phi) \; \phi_{x^1} \; \prod_{v=1}^p \int_{x_v^1 \, x_v^2 \, x_v^3} \; - V(x_v^1,x_v^2,x_v^3) \; \phi_{x^1_v} \phi_{x^2_v} \phi_{x^3_v}  \;.
\]
The Gaussian integral is evaluated as a sum over pairings (or  Wick contractions). The factor $Z(0)^{-1}$ 
has as the effect of canceling all ``vacuum" contributions not connected to the external point $x^1$.

Writing
\[
{\cal D} = \{ r\} \bigcup_{v=1}^p \{ v^1,v^2,v^3\} \;,
\]
and letting $P_{\cal D}$ denote the set of pairings of $\cal D$,
it follows that each pairing $P\in P_{\cal D}$
leads to a labeled combinatorial map, as introduced in Sec.~\ref{subsec:combi},
with half-edge set $\cal D$ and permutation and involution given respectively by
\[
  \sigma= (r) \prod_{v=1}^p (v^1v^2v^3) \;, \qquad \alpha = \prod_{\{v^h,w^f\}\in P}(v^h,w^f) \; .
\]
In particular, the root is $r$ and there are $p$ three valent vertices.

We group together all the labeled combinatorial maps which correspond to the same unlabeled combinatorial map, which at first orders are represented in Fig.~\ref{fig:trevalentwithtree}. In order to count how many times the same unlabeled map is obtained, we consider the unlabeled map and pick a tree by ``keeping to the right'' as described in Sec.~\ref{sec:tree_picking}. Proceeding from the root, we see that the first edge in the tree connects the root with a half edge on a first vertex and there are $p$ possible choices for the label $v$ of the vertex and $3$ possibilities for the choice of the half edge $v^1,v^2$ or $v^3$, hence  $3p$ admissible pairs $(r,v^h)$. Continuing along the tree, the starting half-edge of the next tree edges, $v^{h'}$, is fixed as it is the successor of $v^h$ and there are $3(p-1)$ distinct choices for the end half-edge $v_1^f$ of the second tree edge $(v^{h'},v_1^f)$. Iterating, we conclude that each unlabeled map is obtained exactly $p!3^p$ times hence:
\[
 \braket{\phi_{x^1}} = \sum_{ {\cal G} \in \pmb{\cal G}^{(r|x^1)}_{3} }  \int \left( \prod_{v\in V^{\rm int}({\cal G} )}
   \prod_{h=1}^{3} dx^h_v  \right)
   \prod_{ \{ v^h,w^{h'} \} \in E({\cal G})}C_{x^h_vx^{h'}_w} \;  \prod_{v\in V^{\rm int}({\cal G} )} - V(x_v^1,  x_v^2,x_v^3)  \; ,
\]
where $\pmb{\cal G}^{(r|x^1)}_{3}$ denotes the set of unlabeled rooted maps with three valent internal vertices and, in order to write the integrand, we have chosen a particular labeled representative of the unlabeled combinatorial map. The point of the normalization by $1/(q+1)$ of the interaction kernels \eqref{eq:geint}
is that each unlabeled map comes with coefficient exactly 1 in this sum.

The case above is easily generalized to arbitrary interactions and general $n$-point functions.
The perturbative series is indexed by possibly disconnected (unlabeled) combinatorial maps ${\cal G}$, each counted with coefficient 1.
Every half-edge $v^h$ on an internal vertex $v$ has a position $x^h_{v}$,
and the half-edges of the external vertices inherit the external positions $x_{r^{\nu}} = x^\nu $;
every internal vertex contributes a vertex kernel;
every edge connects two half-edges and contributes a covariance. 
Denoting $\pmb{\cal G}^{(r^1 | x^1) \dots (r^n | x^n)}$ the set of possibly disconnected combinatorial maps with $n$ external vertices $r^1$ with position $x^1$, $r^2$ with position $x^2$ and so on, we have:
\begin{equation}\label{eq:Feynmanampli}
\begin{split}
 \Braket{\phi_{x^1} \dots \phi_{x^n} } & = \sum_{ {\cal G} \in \pmb{\cal G}^{(r^1|x^1) \dots (r^n | x^n)} } {\cal A}({\cal G}) \; , \crcr
 {\cal A}({\cal G}) & = \int \left( \prod_{v\in V^{\rm int}({\cal G} )}
   \prod_{h=1}^{d_{\cal G}(v)} dx^h_v  \right)
   \prod_{ \{ v^h,w^{h'} \} \in E({\cal G})}C_{x^h_vx^{h'}_w} \;  \prod_{v\in V^{\rm int}({\cal G} )} - V(x_v^1,  \dots , x_v^{d_{\cal G}(v ) } )  \;.
\end{split}
\end{equation}
The quantity ${\cal A}({\cal G})$ is called the \emph{amplitude} of ${\cal G}$. 

The lowest orders of such expansion are made explicit in App.~\ref{app:loworderFeynm}, for the case of a local $\phi^4$ interaction.

\subsection{The tree expansion in stochastic quantization}
\label{sec:stoch-exp}

We now turn our attention to the perturbative expansion in terms of vertex kernels in the stochastic quantization.

\paragraph{The free model.} We first consider the free model with action $S(\phi) =  \frac{1}{2} \phi A \phi $, whose $n$-point functions in Eq.~\ref{eq:Wickthm} are the moments of the Gaussian distribution with covariance $C = A^{-1}$. The stochastic gradient descent equation for the free model,
\[
 \partial_t \Phi_{(t,x)} = -\frac{\delta S}{ \delta \phi_x} \Big{|}_{ \phi_x =  \Phi_{(t,x)} } + \xi_{(t,x)}  =   -( A \Phi)_{(t,x)} + \xi_{(t,x)}  \;,\qquad \braket{\xi_{(t,x)} \; \xi_{(t',x')}}_{\xi} =   2  \; \delta_{tt'} \delta_{xx'} \;,
\]
 is solved in terms of the Green function discussed in Sec.~\ref{sec:green}:
\[
\Phi_{(t,x)} = \int_{t_1,x_1} \pmb{C}_{(t,x)(t_1,x_1)} \xi_{(t_1,x_1)} =\int_{-\infty}^t dt_1 \int_{x_1} (e^{- (t-t_1) A} )_{xx_1}  \; \xi_{(t_1,x_1)}\; ,
\]
where we imposed the initial condition at past infinity.

The stochastic correlation functions are computed by taking expectations over the Gaussian noise $\xi$. For instance the stochastic $2$-point function is:\footnote{This computation justifies the normalization of the white noise covariance to $2$, as this offsets the fact that the intermediate time appears as $-2t_1$ after integration over the intermediate position.}
\begin{equation}\label{eq:Phi_free_SQ}
\begin{split}
 \Braket{ \Phi_{(t,x)} \Phi_{(s,y)} }_{\xi}
&  = \int_{-\infty}^t dt_1 \int_{x_1} (e^{- (t-t_1) A} )_{xx_1}
 \int_{-\infty}^s dt_2 \int_{x_2} (e^{- (s-t_2) A } )_{yx_2} \; 2 \; \delta_{t_1t_2} \delta_{x_1x_2} \crcr
& = \int_{-\infty}^{\min(t,s)}  2  \; dt_1 \;
  \left( e^{-(t+s-2t_1) A} \right)_{xy}
\xlongequal{u=t+s-2t_1}
\int_{|t-s|}^{\infty} du  \left( e^{-u A} \right)_{xy}
  = \left( \frac{e^{- |t-s| A }}{ A}  \right)_{xy}
\; ,
\end{split}
\end{equation}
which at equal fictitious times becomes $\Braket{ \Phi_{(t,x)} \Phi_{(t,y)} }_{\xi}
  = \left( 1/A  \right)_{xy} $
reproducing the free $2$-point function in the path integral quantization $\Braket{\phi_x \phi_y}  =  C_{xy}$.
Once the two point function is proven to coincide in the two approaches, the equality of the equal times stochastic expectations:
\[
\Braket{\Phi_{(t, x^1)} \dots \Phi_{(t, x^n)} }_{\xi}  =
 \begin{cases}
      \sum_{P \in P_n } \prod_{(i,j)\in P}
       \Braket{\Phi_{(t, x^i )}\Phi_{ (t, x^j) }}_{\xi} = 
       \sum_{P \in P_n } \prod_{(i,j)\in P}
      \left( A^{-1} \right)_{x^i x^j} \quad & (n \text{ even}),\\
      0  & (n \text{ odd}),
      \end{cases} \; ,
\]
and path integral expectations in Eq.~\eqref{eq:Wickthm} follows.

We observe that the stochastic 2-point function at equal times is independent of the time $t$. This is true for all the stochastic correlation functions and it is 
due to the fact that we have chosen to solve the gradient descent equation with initial condition at past infinity $t_{\rm in}=-\infty$, hence we are already in the infinite-time limit. If instead we solve the equation with initial condition $\Phi_{(0,x)}=0$ at $t_{\rm in}=0$, we find:
\[
\Braket{ \Phi_{(t,x)} \Phi_{(t,y)} }_{\xi} 
= \int_0^{2t} du  \left( e^{-u A} \right)_{xy} \;,
\]
and in order to recover the quantum field theory propagator we need to take the late time limit $t\to+\infty$.

\paragraph{The interacting model.}
We now switch to the interacting case $S(\phi) = \frac{1}{2}  \phi A \phi + V(\phi)$ which leads to a non-linear stochastic gradient descent equation:
\be \label{eq:SPDE}
  ( \partial_t \Phi + A \Phi)_{(t,x)} = - \frac{\delta V} {\delta \phi_x } \Big{|}_{\phi_x = \Phi_{(t,x)} } +\xi_{(t,x)} \;,\qquad \braket{\xi_{(t,x)} \; \xi_{(t',x')}}_{\xi} =  2  \; \delta_{tt'} \delta_{xx'} \; ,
\ee
whose solution can be written as a sum over trees. In fact, following Sec.~\ref{sec:trees_ODEs}, there are several such formulas.
Using \eqref{eq:geintderiv} we write \eqref{eq:SPDE} as:
\[
  ( \partial_t\Phi + A \Phi)_{(t,x)} = - \sum_{q\ge 1}  \int_{ x^1 \dots x^q} \;  V (x, x^1,\dots , x^q) \; \Phi_{(t,x^1)} \dots \Phi_{(t,x^q)}  +\xi_{(t,x)} \;,\;\;  \braket{\xi_{(t,x)} \; \xi_{(t',x')}}_{\xi} =  2  \; \delta_{tt'} \delta_{xx'} \; ,
\]
and we solve it by repeated substitutions in terms of rooted plane trees,
as described in Sec.~\ref{sec:trees_ODEs}, see in particular \eqref{eq:solplantree}.
We denote $\pmb{\cal T}^{(r | t,x)}$ the set of rooted plane trees with root vertex (and half-edge) $r$ with position $x_r=x$ and time $t_r=t$, 
and for every other vertex we label the half-edge pointing towards the root as $1$ to get:
\begin{equation}
\begin{split}\label{eq:SPDTsoltree}
\Phi_{(t,x)} = & \sum_{{\cal T}\in \pmb{\cal T}^{(r | t,x)} }
  \int_{-\infty}^\infty \prod_{v\in   V^{\rm int}( {\cal T} )  }  dt_v
   \int  \left( \prod_{v\in   V^{\rm int}( {\cal T} )  } \prod_{h=1}^{d_{\cal T} (v) } dx^h_v \right)  \;
   \prod_{(v^h,w^1) \in E( {\cal T} ) } \theta(t_v-t_w)  (e^{-(t_v-t_w)A})_{x^h_vx^1_w}  \crcr
 & \qquad \qquad \qquad \times \prod_{ \substack{ { v\in V^{\rm int}( {\cal T} ) } \\ { d_{\cal T} (v ) \ge 2} } }  
    - V(x_v^1,\dots x_v^{d_{\cal T}(v) } )
\prod_{ \substack{  { v\in V^{\rm int}( {\cal T} ) } \\ { d_{\cal T}(v ) = 1 } } } \xi_{( t_v, x_v^1)}    \; ,
\end{split}
\end{equation}
where $E({\cal T})$ includes the edge hooked to the root $(r,w^1)$ for some $w$, hence the root time $t$ is the maximal time.
To verify \eqref{eq:SPDTsoltree}, apply $(\partial_t+A)$ to the heat kernel on the edge incident to the root and use $(\partial_t+A)\pmb{C}_{(t,x)(t',x')}=\delta_{tt'}\delta_{xx'}$ from~\eqref{eq:freeC}. If the vertex attached to the root is a leaf, this gives the noise term $\xi$; otherwise, cutting that vertex produces its ordered branches and reconstructs the corresponding term in $-\delta V/\delta\phi$ in \eqref{eq:SPDE}.
This proves that the right-hand side of \eqref{eq:SPDTsoltree} is indeed a (formal) solution to \eqref{eq:SPDE}.
Alternatively, rewriting \eqref{eq:SPDE} as $\Phi=\pmb{C} (-\delta V/\delta \Phi + \xi)$ and repeatedly substituting it back into  its right-hand side produces every rooted plane tree exactly once, because the plane order records the ordered factors at each interaction vertex, again showing \eqref{eq:SPDTsoltree}.

We remark that we have passed the term $A\Phi$ on the left-hand side of the stochastic gradient descent equation. We can very well decide to leave it on the right-hand side, in which case it would lead to time-ordered integrals of arbitrary long chains of $(-A)$ bivalent vertices, which sum up to:
\[
\sum_{l\ge 0} \int^{t}_{t'} dt_1 \int_{t'}^{t_1}dt_2 \dots \int^{t_{l-1}}_{t'} dt_l  \; (-A)^l
 = e^{-(t-t')A} \;,
\]
reconstructing the heat kernel.

In the case of a local $\phi^{q+1}$ interaction $V(x^1,\dots x^{q+1})=g\prod_{i=2}^{q+1} \delta_{x^1x^i}$, all but one of the half edge positions $x_v^h$ can be integrated out on every vertex and one gets: 
\begin{equation}\label{eq:SPDTq-1arytrees}
\begin{split}
\Phi(t,x) =   & \sum_{{\cal T}  \in \pmb{\cal T}^{(r|t,x)}_{q-{\rm ary}} } 
\int_{-\infty}^{+\infty} \prod_{v \in V({\cal T})} dt_v \int \prod_{v \in V^{\rm int}({\cal T})}  d x_v \;  \prod_{(v,w) \in E( {\cal T} ) } \theta(t_v-t_w)  (e^{-(t_v-t_w)A})_{x_vx_w} 
\crcr
 &  \qquad  \qquad \times
  \prod_{ \substack{ { v\in V^{\rm int}( {\cal T} ) }\\ { d_{\cal T} (v) = q+1 } } }  
    (-g)
\prod_{ \substack{ {  v\in V^{\rm int}( {\cal T} ) } \\ { d_{\cal T} (v) = 1 } } } \xi_{( t_v, x_v)}    \;.
\end{split}
\end{equation}
In this sum we have one term per embedded tree, but the value of the integral is not sensitive to the embedding. In App.~\ref{app:loworderstoc} we compute explicitly some contributions to the stochastic expectations.

\paragraph{Noise edges and the stochastic rules.}\label{sec:stochRules}

In order to compute stochastic expectations:
\[ \Braket{\Phi_{(t^1, x^1)} \dots \Phi_{(t^n, x^n)} }_{\xi}  \;, \] one inserts the tree expansion for the stochastic field $ \Phi_{(t^\nu, x^\nu)}$
and takes the expectation with respect to $\xi$. This amounts to mating leaves by $\xi$ contractions. Each $\xi$ contraction acts by gluing two leaves (either within the same tree or between two distinct trees) and forms a new edge which we call a \emph{noise edge},
that contributes: 
\[
\begin{split}
& \int_{t',t'',z_1,z_2}\pmb{C}_{(t,x) (t',z_1) }
   \pmb{C}_{ (s,y) (t'' , z_2) } \;\; \braket{\xi_{ (t',z_1) } \; \xi_{(t'', z_2) } }_{\xi}
= \left( \frac{e^{- |t-s| A  }}{A}  \right)_{xy} \;,
\end{split}
\]
to the integrand, where in order to evaluate the integral we observe that the left-hand side is exactly the 2-point function of the free theory from \eqref{eq:Phi_free_SQ}.
After performing the $\xi$ average one obtains the stochastic correlation functions as sums over plane trees decorated by embedded noise edges, which we call stochastic diagrams. 

Let us denote $ \pmb{\cal G}^{(r^1| t^1 ,x^1) \dots (r^n| t^n , x^n) }$ the set of possibly disconnected combinatorial maps with $n$ external vertices: $r^1$ with time $t^1$ and position $x^1$, $r^2$ with time $t^2$ and position $x^2$ and so on. The stochastic $n$-point function is a sum over stochastic diagrams with $n$ external vertices, which consist of a possibly
disconnected combinatorial map ${\cal G} \in  \pmb{\cal G}^{(r^1| t^1 ,x^1) \dots (r^n| t^n ,x^n) }$, together with a partition of the edges of ${\cal G}$ into:
\begin{itemize}
    \item[-] \emph{tree edges} which form a forest $ {\cal F} = ({\cal T}^1,\ldots,{\cal T}^n)$ consisting of $n$ disjoint rooted plane trees ${\cal T}^\nu \in \pmb{\cal T}^{(r^\nu|t^\nu x^\nu)}$ for $\nu = 1,\dots n$, that is the tree ${\cal T}^{\nu}$ is rooted at the external vertex $r^\nu$ which has time $t^\nu$ and position $x^\nu$, such that ${\cal F}$ is a spanning forest in ${\cal G}$ (that is every vertex of ${\cal G}$ belongs to a tree ${\cal T}^\nu$), which we denote ${\cal F} \prec {\cal G}$;
    we denote by $\pmb{\cal F}^{(r^1| t^1 ,x^1) \dots (r^n| t^n , x^n) }$ the set of all such forests;

    \item[-] \emph{noise edges} $E({\cal G}) \setminus E( {\cal F} )$. 
\end{itemize}

We order the edges in the forest as $(v^h,w^1) \in {\cal F}$ , where the vertex $v$ is parent of $w$ and we always connect the tree edge to the half edge $w^1$ on $w$. Each stochastic diagram brings a contribution obtained by applying the following diagrammatic rules: 
\begin{itemize}
    \item[-] every stochastic diagram contributes a \emph{combinatorial factor} $1$,
  \item[-] every \emph{external vertex} $r^\nu$ has a fixed time coordinate $t_{r^\nu} = t^\nu$
    and a fixed position $x_{r^\nu}=x^{\nu}$,
    \item[-] every \emph{internal vertex} $v$ has a time coordinate $t_v$ and $ d_{\cal G}(v)$
    positions variables $x_v^{h}$ for $h=1, \dots d_{\cal G}(v)$, one for each half-edge, which are integrated,
    \item[-] every internal vertex $v$ contributes a vertex kernel
         $-V(x_v^1,\dots x_v^{d_{\cal G}(v)})$,
    \item[-] every tree edge $(v^h,w^1)\in E({\cal F})$ contributes a \emph{tree propagator}
    $  \theta(t_v - t_w) (e^{ -(t_v -t_w) A })_{x_v^h x_w^1}$, 
    \item[-] every noise edge $(v^h, w^{h'}) \in E({\cal G}) \setminus E( {\cal F} )$ contributes a \emph{noise-edge propagator}
         $ \left( \frac{e^{- |t_v-t_w| A}}{A} \right)_{x_v^h  x_{w}^{h'} }$.
\end{itemize}
Gathering everything we obtain:
\begin{align}\label{eq:Stochampli}
&  \Braket{\Phi_{(t^1, x^1)} \dots \Phi_{(t^n, x^n)} }_{\xi} =
   \sum_{ {\cal F} \in \pmb{\cal F}^{(r^1| t^1 ,x^1) \dots (r^n| t^n , x^n) } }  \;
   \sum_{ \substack{ { {\cal G} \in  \pmb{\cal G}^{(r^1| t^1 ,x^1) \dots (r^n| t^n , x^n) } } 
   \\ { {\cal F} \prec {\cal G}  } } }  A({\cal G}, {\cal F} ) \; , \\
&  A({\cal G}, {\cal F} )  =   \int_{-\infty}^{+\infty} \prod_{v\in   V^{\rm int}( {\cal G} )  }  dt_v
   \int  \left( \prod_{v\in   V^{\rm int}( {\cal G} )  } \prod_{h=1}^{d_{\cal G} (v) } dx^h_v \right)  \;
   \prod_{(v^h,w^1) \in E(  {\cal F} ) } \theta(t_v-t_w)  (e^{-(t_v-t_w)A})_{x^h_vx^1_w}  \nonumber \\
 & \qquad \qquad \qquad \qquad
  \prod_{ \{ v^h, w^{h'} \}  \in  E({\cal G}) \setminus E( {\cal F} ) }           \left( \frac{e^{- |t_v-t_w| A}}{A} \right)_{x_v^h  x_{w}^{h'} }   \; 
 \prod_{v\in V^{\rm int}( {\cal G} ) }  
    - V(x_v^1,\dots x_v^{d_ {\cal G} (v) } ) \; . \nonumber
\end{align}
Observe that, while the tree edges are oriented $ (v^h,w^1) \in E( {\cal F}) $, with $v$ the parent of $w$, the noise edges $\{ v^h, w^{h'} \}  \in  E({\cal G}) \setminus E( {\cal F} ) $ are not oriented.

The lowest orders of such expansion are made explicit in App.~\ref{app:loworderstoc}, for the case of a local $\phi^4$ interaction.

\subsection{Main result}

We now give the precise reformulation of our main Theorem \ref{thm:main} in the sense perturbation series.

\begin{theorem}\label{thm:main_precise}
Consider ${\cal G} \in \pmb{\cal G}^{(r^1|x^1) \dots (r^n | x^n)} $.
Then
\begin{equation}\label{eq:sumstoch}
 {\cal A}({\cal G}) =  
   \sum_{   \substack{     { {\cal F} \in \pmb{\cal F}^{(r^1| t ,x^1) \dots (r^n| t , x^n) } } \\    { {\cal F} \prec {\cal G}   } } }
    A({\cal G}, {\cal F} )  \;. 
\end{equation}
\end{theorem}

In other words, every term in the Feynman expansion in Eq.~\eqref{eq:Feynmanampli} is the sum of the amplitudes of all the stochastic diagrams  in \eqref{eq:Stochampli} based on the same combinatorial map ${\cal G}$ and equal external times $t^\nu  = t $.

By summing over ${\cal G} \in \pmb{\cal G}^{(r^1|x^1) \dots (r^n | x^n)} $ we recover the statement in Theorem \ref{thm:main} in the sense of formal power series in the vertex kernels.
 
\section{Proof of Theorem~\ref{thm:main_precise} at the level of individual graphs}
\label{sec:first_proof}

The strategy of our proof consists in starting from a Feynman amplitude ${\cal A}({\cal G})$ and rewrite it as a sum over forests ${\cal F} \prec {\cal G}$ using an appropriate Taylor interpolation. Starting from a combinatorial map ${\cal G}$ with $n$ external vertices $r^{\nu},\nu=1,\dots n$, one can list all the forests 
${\cal F} \prec {\cal G}$ proceeding recursively as follows:
\begin{itemize}
    \item[-] we declare at step $0$ the set of vertices $ V_0 = \{ r^{\nu} | \nu = 1,\dots n \} $ consisting in the external vertices of ${\cal G}$, the set of tree edges ${\cal F}_0 = \emptyset$, and the set of noise edges ${\cal L}_0 = \{ \{ r^\nu,r^{\nu'} \}\in E({\cal G})\} $ consisting in the edges connecting directly external vertices of ${\cal G}$ (if they exist);
    \item[-] for every step $p = 1,\dots , |V^{\rm int}({\cal G})|$ we pick any of the edges $e_p = ( v^h ,v_p^{h'} ) $ connecting one of the vertices $v\in V_{p-1}$ to a vertex 
    $v_p\in V^{\rm int}({\cal G}) \setminus V_{p-1}$; we orient $e_p$ from $v$ to $v_p$ and relabel $v_p^{h'} \equiv v_p^1$;
    \item[-] we update:
      \begin{itemize}
      \item the set of vertices to $V_p = V_{p-1} \cup \{v_p\}$,
      \item the set of tree edges to ${\cal F}_p = {\cal F}_{p-1} \cup \{e_p\}$,
      \item the set of noise edges to:
       \[
       {\cal L}_p = {\cal L}_{p-1}\cup \bigg\{ \{  u^{h}, v_p^{h'} \} \in E({\cal G}) \, \Big{|} \, u \in V_{p} \, , \, \,\{  u^{h}, v_p^{h'} \}  \neq e_p \bigg\} \; ,\]
       that is, we declare as noise edges all the non-tree edges that connect the vertex $v_p$ to vertices in $V_{p-1}$, as well as the tadpole edges on $v_p$;
      \end{itemize}
    \item[-] the set of tree edges after $ |V^{\rm int}({\cal G})|$ iteration steps is a spanning forest in ${\cal G}$, and listing all the allowed choices of tree edges at all the steps, exhausts the list of all such forest.
\end{itemize}

This iterative construction tells us how to combinatorially partition ${\cal G}$ into spanning forests glued by noise edges, but it does not tell us how to relate the amplitude ${\cal A}({\cal G})$ in Eq.~\eqref{eq:Feynmanampli} to the amplitudes of the stochastic diagrams  in \eqref{eq:Stochampli}, and in particular how the fictitious time and heat kernel arise.
This part of the proof is achieved
by a judicious use of the fundamental theorem of calculus, or lowest-order Taylor expansion with integral rest:
\[
f(t) = f(a) + \int_a^t dt' \;\frac{d}{dt'} f(t') \;,
\]
which, remembering that $\lim_{u\to \infty} (e^{-uA} )_{x_1y_1} =0$ on the relevant subspace as recalled in Remark~\ref{rem:heat_kernel}, we use to write the following fundamental formula:
\[
 \prod_i \left(  \frac{e^{ -(s_i - t_{p-1}) A}}{A}\right)_{x_iy_i}
 =\int_{-\infty}^{t_{p-1}} dt_p \;\frac{d}{dt_p}
  \left[  \prod_i \left(  \frac{e^{ -(s_i - t_p) A}}{A}\right)_{x_iy_i} \right] \;,
\]
where $s_i\geq t_{p-1}$ will below turn out to be the time already assigned to the parent endpoint of the $i$-th candidate edge.
This formula holds also in the case that $s_i=t_{p-1}$ for some (or all) $i$, in which case the corresponding heat kernels on the left-hand side are replaced by a delta function, because $\lim_{u\to 0} (e^{-uA} )_{x_1y_1} = \delta_{x_1y_1}$.
We have used subscripts $p-1$ and $p$ in anticipation of the iterative procedure to be described below.

At each step one needs to select candidate tree edges and apply the interpolation formula. At the first step $p=1$ the candidate tree edges are all the edges that connect vertices in $V_0= V^{\rm ext}({\cal G})$ with vertices in $V({\cal G}) \setminus V_0 = V^{\rm int}({\cal G})$, that is the edges of ${\cal G}$ which connect external with internal vertices. 
We assign to all the external vertices $r^{\nu}$ the same fictitious time $t_{r^\nu} = t$, and to all the internal vertices connected to the external ones the same fictitious time $t_1 \le t$. We interpolate\footnote{In case two external vertices are directly connected by an edge, this edge is declared a noise edge and it is not interpolated.} as:
\[
\begin{split}
& \prod_{\substack{ e = (r^\nu, w^{h} ) \in E({\cal G}) \\ r^{\nu}\in V_0 ,  w \in V({\cal G}) \setminus V_0  } } \left( \frac{1}{A}\right)_{x^\nu x_w^h }
 =  \int_{-\infty}^{t} dt_1
\frac{d}{dt_1} \left[ \prod_{\substack{ e = (r^\nu, w^h ) \in E({\cal G}) \\ r^{\nu}\in V_0 ,  w \in V({\cal G}) \setminus V_0  } } \left( \frac{e^{-(t-t_1)A}}{A}\right)_{x^\nu x_w^h } \right] \crcr
& \qquad  =\sum_{\substack{e_1= (r^{\nu} ,v_1^1 ) \in E({\cal G}) \\ r^{\nu} \in V_0 ,  v_1^1\in V( { \cal G } ) \setminus V_0 } } \int_{-\infty}^t
    dt_1 \; (e^{-(t-t_1) A})_{x^{\nu} x^1_{v}  }  
 \prod_{\substack{e = (r^{\mu},w^{h} ) \in E({\cal G})  \setminus \{e_1\} \\ r^{\mu} \in V_0 , w\in V( { \cal G } ) \setminus V_0 } }  
    \left( \frac{e^{-(t-t_1) A}}{A} \right)_{x^{\mu} x^h_w } \;.
\end{split}
\]
It follows that at this first step we obtain a sum over all the possible ways to chose a first tree edge $e_1=(r^{\nu}, v_1^{h})$ connecting one vertex in $V_0$ with an internal vertex,
and we relabel the half edge $v_1^{h}$ as $v_1^1$. We set
$V_1 = V_0 \cup\{ v_1\}$, ${\cal F}_1 = \{e_1\}$ and, once $e_1$ is chosen,  if there are any other edges connecting vertices in $V_0$ (that is external vertices) with $v_1$, we declare them noise edges and we also declare as noise edges any tadpoles on $v_1$.

The candidate tree edges at step $p=2$ depend on the first tree edge and its end vertex.
For each term in the sum, we declare as candidate tree edges all the edges $(v^h,v_2^{h'})$ which connect a vertex $v \in V_1$ with a vertex $v_2\in V({\cal G})\setminus V_1 $. 
For $v\in V_0 = \{ r^1,\ldots, r^n \}$, these edges have an interpolated propagator $ e^{-(t-t_1) A }/A  $, 
while for $v =v_1 \in V_1\setminus V_0$ they have a propagator $1/A =  e^{-(t_1-t_1)A}/A $.
We take a second interpolation on all these candidate tree edges $(v^h,v_2^{h'} )$ by assigning a fictitious time $t_2 \le t_1$ to the vertices $v_2\in V({\cal G}) \setminus V_1$. The interpolation at step $p$ writes: 
\[
\begin{split}
& \prod_{\substack{ e = (v^h, w^{h'} ) \in E({\cal G}) \\ v\in V_{p-1} ,  w \in V({\cal G}) \setminus  V_{p-1} } } \left( \frac{e^{-(t_v-t_{p-1} ) A } }{A}\right)_{x_v^h x_w^{h'} }
 =  \int_{-\infty}^{t_{p-1}} dt_p
\frac{d}{dt_p} \left[\prod_{\substack{ e = (v^h, w^{h'} ) \in E({\cal G}) \\ v\in V^{p-1} ,  w \in V({\cal G}) \setminus  V_{p-1} } } \left( \frac{e^{-(t_v-t_{p} ) A } }{A}\right)_{x_v^h x_w^{h'}  } \right] \crcr
&  =\sum_{\substack{e_p = ( v^h ,v_p^1 ) \in E({\cal G}) \\  v \in V_{p-1} ,  v_p \in V( { \cal G } ) \setminus V_{p-1} } } \int_{-\infty}^{t_{p-1}}
    dt_p \; (e^{-(t_v - t_p ) A})_{x^h_{v} x^1_{v_p}  }  
 \prod_{\substack{e = \{ u^\ell,w^{h'}  \}  \in E({\cal G})  \setminus \{e_p\} \\ u \in V_{p-1} , w\in V( { \cal G } ) \setminus V_{p-1} } }  
    \left( \frac{e^{-(t_u-t_p ) A}}{A} \right)_{x^\ell_u x^{h'}_w } \;,
\end{split}
\]
where we note that the vertex $v$ belongs to $V_{p-1}$, hence has a time which we denote $t_v$ in this formula. This yields a sum over all the possible choices for the tree edge $e_p$ at step $p$.
Note that $t_v=t_{p-1}$ for the vertex $v=v_{p-1}$.

We iterate the interpolations on candidate tree edges until we exhaust the vertices. By the end of the iteration we have chosen a forest ${\cal F} $ of edges in ${\cal G}$ consisting in $n$ trees $ {\cal T}^\nu\in \pmb{\cal T}^{(r^\nu|t, x^\nu)}$, one for each root $r^\nu \in V_0$, such that every vertex in $V({\cal G})$ is either an external vertex (hence has an associated time $t$) or is the vertex $v_p$ added at step $p$ for some $p$, hence has time $t_p$. The times are totally ordered $t \ge t_1 \ge \dots \ge t_{|V({\cal G})|}$ and each term contributes (we order all the edges $(v_q^h,v_p^{h'})$ with $t_q>t_p$):
\[
\begin{split}
\int_{-\infty}^{t} dt_1 \dots \int_{-\infty}^{t_{|V({\cal G})| -1}} dt_{ |V({\cal G})| }
   \int  &\Big\{
   \prod_{v\in   V^{\rm int}( {\cal G} )  } \prod_{h=1}^{d_{\cal G} (v) } dx^h_v
   \Big\} 
   \Big\{
   \prod_{(v_q^h,v_p^1) \in E( {\cal F}    ) }  (e^{-(t_q-t_p)A})_{x^h_{v_q} x^1_{v_p}} 
   \Big\}
   \crcr
   & \Big\{
   \prod_{ ( v_q^h, v_p^{h'} ) \in  E({\cal G}) \setminus E( {\cal F}   )}   \left( \frac{e^{- ( t_q-t_p) A}}{A} \right)_{x^h_{v_q} x^{h'}_{v_p}}
   \Big\}
   \; 
 \Big\{
 \prod_{v_p \in V^{\rm int}( {\cal G} ) }  
    - V(x_{v_p}^1,\dots x_{v_p}^{d_{\cal G}(v_p) } )
    \Big\}\;.
\end{split}
\]

Finally, for every forest ${\cal F} \prec {\cal G}$ we sum over all total orders of the times compatible with its tree partial order. This is a simplified analogue of Lemma~\ref{lemma:trees} for labeled-forest: the corresponding ordered simplexes partition the partial-order domain, and since the vertices and half-edges retain their labels from ${\cal G}$, which is a combinatorial map, no symmetry factors arise. We obtain:
\[
\begin{split}
&  {\cal A}({\cal G}) =  
   \sum_{\substack{     { {\cal F} \in \pmb{\cal F}^{(r^1| t ,x^1) \dots (r^n| t , x^n) } }  \\    { {\cal F} \prec {\cal G}   }  } }
    A({\cal G}, {\cal F} )  \; , \crcr
&  A({\cal G}, {\cal F} )  =   \int_{-\infty}^{+\infty}
\Big\{
\prod_{v\in   V^{\rm int}( {\cal G} )  }  dt_v
\Big\}
   \int  \Big\{ \prod_{v\in   V^{\rm int}( {\cal G} )  } \prod_{h=1}^{d_{\cal G} (v) } dx^h_v \Big\}  \;
   \Big\{
   \prod_{(v^h,w^1) \in E( {\cal F}) } \theta(t_v-t_w)  (e^{-(t_v-t_w)A})_{x^h_vx^1_w}
   \Big\}
   \nonumber \\
 & \qquad\qquad\qquad\qquad
  \Big\{
  \prod_{ \{ v^h, w^{h'} \}  \in  E({\cal G}) \setminus E( {\cal F} )}           \left( \frac{e^{- |t_v-t_w| A}}{A} \right)_{x_v^h  x_{w}^{h'} }
  \Big\}
  \;
  \Big\{
 \prod_{v\in V^{\rm int}( {\cal G} ) }  
    - V(x_v^1,\dots x_v^{d_ {\cal G} (v) } )
    \Big\}
    \; ,
\end{split}
\]
proving Eq.~\eqref{eq:sumstoch}. Some examples are worked out in App.~\ref{app:lowordex}

\section{Proof of Theorem~\ref{thm:main}  via a path integral Taylor interpolation}
\label{sec:pathint}

The Taylor interpolation we performed at the level of individual graphs in the previous section can be promoted to a Taylor interpolation at the level of the path integral.
We thus prove directly, without performing the full perturbative expansion, that the correlation functions of the model can be written as the expectation over a white noise $\xi_{(t,x)}$ of a product of stochastic fields $\Phi_{(t,x)}$ which are given by tree sums as in Eq.~\eqref{eq:SPDTsoltree};
recall that this is the rooted plane tree expansion obtained by iterating the Langevin equation.
In fact we will show the more general result that the expectation with respect to the field $\phi$ of an arbitrary observable $F(\phi)$ is the expectation of $F(\Phi)$ with respect to the white noise.
This proves the second part of Theorem~\ref{thm:main} avoiding the full Feynman perturbative expansion. 

 In order to simplify the notation we suppress whenever possible the position variables, that is, below:
 \[
\Phi_t \equiv \Phi_{(t,x)}\; , \;\; \xi_t \equiv \xi_{(t,x)} \;, \;\; \pmb{C}_{t,s} = \theta(t-s) e^{-(t-s)A} \equiv \theta(t-s) (e^{-(t-s)A})_{xy} = \pmb{C}_{(t,x)(s,y)} \;,
 \]
so that $\int_s \pmb{C}_{t,s} \xi_s \equiv \int_s \int_y \theta(t-s) (e^{-(t-s)A})_{xy} \,\xi_{(s,y)} $ and so on.

In the following, we will use the notion of recursive trees (resp. rooted planar trees) from Sec.~\ref{sec:trees_ODEs}
with the difference that we allow the root vertex $0$ (resp. $r$) to have degree different from $1$ (including $0$).
Lemma~\ref{lemma:trees} is used below only for non-trivial trees that have a univalent root.

\paragraph{Taylor interpolation.}
In order to write integrals of functionals $F(\phi)$ against the interacting measure $\nu$ in \eqref{eq:nu}, we use the properties of the Gaussian measure discussed in Sec.~\ref{sec:Gauss} and introduce two copies of  the field, one copy $\phi^0$ for the functional $F$, and a second copy $\phi^1$ 
for the interaction. 
That is, we write
\[
\nu(F) = \int \diff \nu(\phi) F(\phi)  =  \braket{F(\phi)}  =
\frac{1}{Z(0)}\left[ e^{\frac{1}{2}
  \delta_{\phi^0} C \delta_{\phi^0} +   \delta_{\phi^0} C \delta_{\phi^1}
  + \frac{1}{2}   \delta_{\phi^1} C \delta_{\phi^1} } 
 \;   F(\phi^0) \;  e^{ - V(\phi^1)}  
\right]_{\phi^0 = \phi^1=0} \; ,
\]
which is a special case of \eqref{eq:copies} and
where we used $\delta_{\phi^0} C \delta_{\phi^1}=\delta_{\phi^1} C \delta_{\phi^0}$ by symmetry of $C$.

This is the starting point of our interpolation. Taking into account that the covariance of the theory can be represented as:
\[ 
\frac{1}{A} = \frac{ e^{ - (t_i - t_j ) A } }{A} \Big{|}_{t_i=t_j}=  \int_{0}^{\infty} du \; e^{- (t_i-t_j+u)A} \Big{|}_{t_i=t_j} \;,
\]
we have:
\[
\nu(F) =  
\frac{1}{Z(0)}\left[ e^{\frac{1}{2} \delta_{\phi^0} C \delta_{\phi^0} + 
   \delta_{\phi^0} \frac{e^{ - (t_0 - t_1 ) A } }{A} \delta_{\phi^1}
+\frac{1}{2} \delta_{\phi^1} C \delta_{\phi^1} } 
 \; F(\phi^0) \; e^{ - V(\phi^1) }    \right]^{t_1 = t_0}_{\phi=0} \; ,
\]
where $\phi=0$ means all $\phi^i=0$ and we introduced two equal fictitious time variables $t_0=t_1$.
Using a Taylor formula with integral rest at first order this becomes:
\[
\begin{split}
\nu(F)  = & 
\frac{1}{Z(0)}\left[ e^{\frac{1}{2} \delta_{\phi^0} C \delta_{\phi^0} + 
  \delta_{\phi^0} \frac{e^{ - (t_0 - t_1 ) A } }{A} \delta_{\phi^1}  +\frac{1}{2} \delta_{\phi^1} C \delta_{\phi^1} } 
 F(\phi^0) \; e^{ - V(\phi^1) }   \right]^{t_1 = -\infty }_{\phi=0}
\crcr
& \qquad \qquad
+ \frac{1}{Z(0)}
\int^{t_0}_{-\infty} dt_1 \frac{d}{dt_1} \left[ e^{\frac{1}{2} \delta_{\phi^0} C \delta_{\phi^0} + 
  \delta_{\phi^0} \frac{e^{ - (t_0 - t_1 ) A } }{A} \delta_{\phi^1}
+\frac{1}{2} \delta_{\phi^1} C \delta_{\phi^1}} 
 \; F(\phi^0) \;  e^{ - V(\phi^1)}   \right]_{\phi=0} \; ,
\end{split}
\]
and in the boundary term at $t_1 = -\infty$ we observe that since the mixing term vanishes, the integral over $\phi^1$ decouples and reconstitutes $Z(0)$, that is:
\[
\begin{split}
 \nu(F)   & = \left[ e^{\frac{1}{2} \delta_{\phi^0} C \delta_{\phi^0} } 
 \; F(\phi^0)  \right]_{\phi=0} \crcr
 & \qquad+ 
 \frac{1}{Z(0)}
 \int_{-\infty}^{t_0} dt_1 \; 
 \left[ e^{\frac{1}{2} \delta_{\phi^0} C \delta_{\phi^0} + 
  \delta_{\phi^0} \frac{e^{ - (t_0 - t_1 ) A } }{A} \delta_{\phi^1}
+\frac{1}{2} \delta_{\phi^1} C \delta_{\phi^1}} \left( \delta_{\phi^0} e^{ - (t_0 - t_1 ) A } \delta_{\phi^1} \right)
 \; F(\phi^0) \; e^{ - V(\phi^1) }   \right]_{\phi=0} \; .
\end{split}
\] 
Commuting $e^{ - V(\phi^1)}$ with the non-exponentiated derivatives,
we get:
\[
\begin{split}
\nu(F) =  & \left[ e^{\frac{1}{2} \delta_{\phi^0} C \delta_{\phi^0} } 
 \; F(\phi^0)  \right]_{\phi=0} \crcr
 & + 
 \frac{1}{Z(0)}
 \int_{-\infty}^{t_0} dt_1 \; 
 \left[ e^{\frac{1}{2} \delta_{\phi^0} C \delta_{\phi^0} + 
  \delta_{\phi^0} \frac{e^{ - (t_0 - t_1 ) A } }{A} \delta_{\phi^1}
+\frac{1}{2} \delta_{\phi^1} C \delta_{\phi^1}} \; e^{ - V(\phi^1) }  \left( \delta_{\phi^0} e^{ - (t_0 - t_1 ) A } \delta_{\phi^1} \right) F(\phi^0) (- V|_{\phi^1} )  
 \right]_{\phi=0} 
 \; ,
\end{split}
\]
which is the result of our first interpolation step.

We continue with a second interpolation in the integral rest term. We begin by introducing two copies of field $\phi^1$:
\begin{itemize}
    \item[-] one copy we keep for the explicit vertex $-\delta_{\phi^1} V|_{\phi^1} $, which we have brought down from the exponential and which conserves the label $\phi^1$;
    \item[-] a second copy we use in the exponential of the interaction, which we label $\phi^2$, leading to an $e^{-V(\phi^2)}$ insertion in the path integral.
\end{itemize}
This leads to the following expression: 
\[
\begin{split}
\nu(F)  =  \left[ e^{\frac{1}{2} \delta_{\phi^0} C \delta_{\phi^0} }  \; F(\phi^0)  \right]_{\phi=0} +  & 
\frac{1}{Z(0)}
\int_{-\infty}^{t_0} dt_1 \; 
  \left[ e^{\frac{1}{2} \delta_{\phi^0} C \delta_{\phi^0} + 
  \delta_{\phi^0} \frac{e^{ - (t_0 - t_1 ) A } }{A}  ( \delta_{\phi^1}+ \delta_{\phi^2} ) 
+\frac{1}{2} ( \delta_{\phi^1} + \delta_{\phi^2} ) C ( \delta_{\phi^1} + \delta_{\phi^2} ) }   \right. \crcr
&\qquad  \left.  e^{ - V(\phi^2) }  \; ( \delta_{\phi^0} e^{ - (t_0 - t_1 ) A } \delta_{\phi^1} )  (- V|_{\phi^1} )
 \; F(\phi^0)  \right]_{\phi=0}\; .
\end{split} 
\]

Proceeding with the second interpolation step, we introduce a fictitious time $t_2=t_1$ in the Gaussian measure which we associated to every operator $\delta_{\phi^2}$ in the mixed terms, that is we rewrite the second bracket as:
\[
\begin{split}
&  \Bigg[ \exp\Bigg\{\frac{1}{2} \delta_{\phi^0} C \delta_{\phi^0} + 
 \delta_{\phi^0} \frac{e^{ - (t_0 - t_1 ) A } }{A} \delta_{\phi^1} 
 + \frac{1}{2} \delta_{\phi^1} C \delta_{\phi^1} +  
 \delta_{\phi^0} \frac{e^{ - (t_0 - t_2 ) A } }{A} \delta_{\phi^2}
 + \delta_{\phi^1} \frac{e^{ - (t_1 - t_2 ) A } }{A} \delta_{\phi^2} 
+\frac{1}{2} \delta_{\phi^2} C \delta_{\phi^2} \Bigg\}  
\\
&  \qquad \qquad  e^{ - V(\phi^2) } \; ( \delta_{\phi^0} e^{ - (t_0 - t_1 ) A } \delta_{\phi^1} )   (- V|_{\phi^1} )
 \; F(\phi^0)   \Bigg]^{t_2 = t_1}_{\phi=0}   \; .
\end{split} 
\] 
Applying the Taylor interpolation formula on $t_2$, we obtain two terms: one corresponding to the boundary value $t_2 = -\infty$ and one with an integral $\int^{t_1}_{-\infty} dt_2 \frac{d}{dt_2} [\ldots]$. In the boundary term the integral over $\phi^2$ decouples and reconstitutes $Z(0)$:
\[
\begin{split}
&
\frac{1}{Z(0)}\left[ e^{\frac{1}{2} \delta_{\phi^0} C \delta_{\phi^0} + 
 \delta_{\phi^0} \frac{e^{ - (t_0 - t_1 ) A } }{A} \delta_{\phi^1} 
 + \frac{1}{2} \delta_{\phi^1} C \delta_{\phi^1} 
+\frac{1}{2} \delta_{\phi^2} C \delta_{\phi^2}}   \; e^{ - V(\phi^2) }  \;
\left( \delta_{\phi^0} e^{ - (t_0 - t_1 ) A } \delta_{\phi^1} \right)   (- V|_{\phi^1} ) 
 \; F(\phi^0)   \right]_{\phi=0} \\
& = 
\left[ e^{\frac{1}{2} \delta_{\phi^0} C \delta_{\phi^0} + 
 \delta_{\phi^0} \frac{e^{ - (t_0 - t_1 ) A } }{A} \delta_{\phi^1} 
 + \frac{1}{2} \delta_{\phi^1} C \delta_{\phi^1} }    
\left( \delta_{\phi^0} e^{ - (t_0 - t_1 ) A } \delta_{\phi^1} \right)   (- V|_{\phi^1} ) \;
 F(\phi^0)   \right]_{\phi=0} \; ,
 \end{split}
\]
while the integral rest term writes:
\[
\begin{split}
&  \int_{-\infty}^{t_1} dt_2
  \frac{d}{dt_2}   \Bigg[ e^{\frac{1}{2} \delta_{\phi^0} C \delta_{\phi^0} + 
 \delta_{\phi^0} \frac{e^{ - (t_0 - t_1 ) A } }{A} \delta_{\phi^1} 
 + \frac{1}{2} \delta_{\phi^1} C \delta_{\phi^1} +  
 \delta_{\phi^0} \frac{e^{ - (t_0 - t_2 ) A } }{A} \delta_{\phi^2}
 + \delta_{\phi^1} \frac{e^{ - (t_1 - t_2 ) A } }{A} \delta_{\phi^2} 
+\frac{1}{2} \delta_{\phi^2} C \delta_{\phi^2} }  
\\
&  \qquad \qquad \qquad \qquad  \qquad \qquad  e^{ - V(\phi^2) } \; \left( \delta_{\phi^0} e^{ - (t_0 - t_1 ) A } \delta_{\phi^1} \right) (- V|_{\phi^1} ) 
 \; F(\phi^0)  \Bigg]_{\phi=0}   \; .
\end{split}
\]
We now evaluate the action of the derivative $d/dt_2$, which  brings down a second tree edge (as the edge must end on the $e^{ -V(\phi^2)} $ term). Putting everything together, we obtain:
\begin{equation}\label{eq:marelung}
\begin{split}
 & \nu(F) =  \left[ e^{\frac{1}{2} \delta_{\phi^0} C \delta_{\phi^0} }  \; F(\phi^0)  \right]_{\phi=0} \crcr
& +    \int_{-\infty}^{t_0} dt_1 \;  
\left[ e^{\frac{1}{2} \delta_{\phi^0} C \delta_{\phi^0} + 
 \delta_{\phi^0} \frac{e^{ - (t_0 - t_1 ) A } }{A} \delta_{\phi^1} 
 + \frac{1}{2} \delta_{\phi^1} C \delta_{\phi^1} }    
\left( \delta_{\phi^0} e^{ - (t_0 - t_1 ) A } \delta_{\phi^1} \right)  (- V|_{\phi^1} ) \;
 F(\phi^0)   \right]_{\phi=0}  + R_2(F)  \;,
\end{split} 
\end{equation}
with the Taylor integral rest term at order $2$:
\[
\begin{split}
R_2(F)  & =     \frac{1}{Z(0)} \int_{-\infty}^{t_0} dt_1  \int_{-\infty}^{t_1} dt_2
  \Bigg[ e^{\frac{1}{2} \delta_{\phi^0} C \delta_{\phi^0} + 
 \delta_{\phi^0} \frac{e^{ - (t_0 - t_1 ) A } }{A} \delta_{\phi^1} 
 + \frac{1}{2} \delta_{\phi^1} C \delta_{\phi^1} +  
 \delta_{\phi^0} \frac{e^{ - (t_0 - t_2 ) A } }{A} \delta_{\phi^2}
 + \delta_{\phi^1} \frac{e^{ - (t_1 - t_2 ) A } }{A} \delta_{\phi^2} 
+\frac{1}{2} \delta_{\phi^2} C \delta_{\phi^2} }  
\\
&  \qquad \qquad \qquad \qquad 
  e^{ - V(\phi^2) }\; \left( \delta_{\phi^0} e^{ - (t_0 - t_2 ) A } \delta_{\phi^2}\right)   (- V|_{\phi^2} ) 
\left( \delta_{\phi^0} e^{ - (t_0 - t_1 ) A } \delta_{\phi^1} \right)  (- V|_{\phi^1} ) 
 \; F(\phi^0)   \Bigg]_{\phi=0} \crcr
& +
\frac{1}{Z(0)}
\int_{-\infty}^{t_0} dt_1  \int_{-\infty}^{t_1} dt_2
    \Bigg[ e^{\frac{1}{2} \delta_{\phi^0} C \delta_{\phi^0} + 
 \delta_{\phi^0} \frac{e^{ - (t_0 - t_1 ) A } }{A} \delta_{\phi^1} 
 + \frac{1}{2} \delta_{\phi^1} C \delta_{\phi^1} +  
 \delta_{\phi^0} \frac{e^{ - (t_0 - t_2 ) A } }{A} \delta_{\phi^2}
 + \delta_{\phi^1} \frac{e^{ - (t_1 - t_2 ) A } }{A} \delta_{\phi^2} 
+\frac{1}{2} \delta_{\phi^2} C \delta_{\phi^2} }  
\\
&  \qquad \qquad \qquad \qquad 
  e^{ - V(\phi^2) } \left( \delta_{\phi^1} e^{ - (t_1 - t_2 ) A } \delta_{\phi^2}  \right)  (- V|_{\phi^2} ) 
\left( \delta_{\phi^0} e^{ - (t_0 - t_1 ) A } \delta_{\phi^1} \right)  (- V|_{\phi^1} )
 \; F(\phi^0)   \Bigg]_{\phi=0}
 \; .
\end{split}
\] 

We associate the terms in this expression to recursive trees with vertices labeled $0,1$ and $2$ as follows. The vertices of the trees correspond to the 
non-expotentiated insertions $F(\phi^0) $, $- V|_{\phi^1} $ and $- V|_{\phi^2} $ and are labeled $0,1,2$ by the label of the field copy. The tree edges correspond to the non-exponentiated derivatives: the term $ \delta_{\phi^i} e^{ - (t_i - t_j ) A } \delta_{\phi^j}$ represents the edge $\{i,j\}$. Thus, in order, the four terms in Eq.~\eqref{eq:marelung} corresponds to:
\begin{itemize}
    \item[-] the tree with one vertex labeled $0$ and no edge,
    \item[-] the tree with vertices $0$ and $1$ and an edge $(0,1)$,
    \item[-] the two rest terms correspond to the two trees with vertices $0$, $1$ and $2$ and edges $ \{ ( 0,1 ), ( 0,2 ) \}$ and $\{ ( 0,1 ),( 1,2 ) \}$ respectively. Such terms are divided by $Z(0)$ and have an additional exponentiated interaction $e^{-V(\phi^2)}$ placed at the left of the non-exponentiated tree differential operator, such that the latter acts only on the explicit insertion $- V|_{\phi^2} $ and not on $e^{-V(\phi^2)}$.
\end{itemize}

At next step, in the two integral rest terms we introduce a new copy $\phi^3$ for the field by replacing $e^{-V(\phi^2) } \to e^{-V(\phi^3)} $ and $\delta_{\phi^2} \to \delta_{\phi^2} + \delta_{\phi^3}$ and we interpolate all the mixed terms in the Gaussian measure involving one $\delta_{\phi_3}$ with a third time parameter $t_3\le t_2$. All the boundary terms at $t_3=-\infty$
contain a decoupled integral on the field $\phi^3$ which reconstitutes a $Z(0)$ factor and cancels the denominator.

After the interpolation step $p$, the boundary terms correspond to all the recursive trees with up to $p$ vertices corresponding to $F(\phi^0)$, 
$- V|_{\phi^1} $, $- V|_{\phi^2} $ up to $-V|_{\phi^{p-1}}$ and edges corresponding to the non-exponentiated tree derivative operator $ \prod_{\{i,j\}}\delta_{\phi^i} e^{ - (t_i - t_j ) A } \delta_{\phi^j}$. The rest terms have a prefactor $Z(0)^{-1}$ and are a sum over all the recursive trees with $p+1$ vertices  $F(\phi^0)$, 
$- V|_{\phi^1} $, $- V|_{\phi^2} $ up to $-V|_{\phi^{p}}$, edges the tree differential operator $ \prod_{\{i,j\}} \delta_{\phi^i} e^{ - (t_i - t_j ) A } \delta_{\phi^j}$ and having an additional exponentiated interaction $e^{-V(\phi^p)}$ at the left of the tree operator.
We introduce copies of the field $\phi^p$, that is $\delta_{\phi^p} \to \delta_{\phi^p} + \delta_{\phi^{p+1}}$, reserving the new copy $\phi^{p+1}$ for the exponentiated interaction $e^{-V(\phi^p)} \to e^{-V(\phi^{p+1})}$, we interpolate the Gaussian measure as:
\begin{equation*}\begin{aligned}
&\delta_{\phi_q} \frac{e^{ -(t_q - t_{p}) A }}{A}\delta_{\phi_{p}} &\to &  \quad
\delta_{\phi_q} \frac{e^{ -(t_q - t_{p}) A }}{A}\delta_{\phi_{p}}  + \delta_{\phi_q} \frac{e^{ -(t_q - t_{p+1}) A }}{A}\delta_{\phi_{p+1}} \Big{|}_{t_{p+1} = t_p}\;,\qquad \forall q < p  \;,\\
&\frac{1}{2} \delta_{\phi_p} \frac{1}{A}\delta_{\phi_{p}} & \to & \quad
\frac{1}{2} \delta_{\phi_p} \frac{1}{A}\delta_{\phi_{p}} 
+\frac{1}{2} \delta_{\phi_{p+1}} \frac{1}{A}\delta_{\phi_{p+1}}
+\delta_{\phi_p} \frac{e^{ -(t_p - t_{p+1}) A }}{A}\delta_{\phi_{p+1}} \Big{|}_{t_{p+1} = t_p} \; ,
\end{aligned}\end{equation*}
and we use a Taylor formula with integral rest on $t_{p+1}$. This generates both the new boundary terms and the new rest terms corresponding to recursive trees with one more vertex.
Effectively, at step $p+1$ the rest term from step $p$ gives  new boundary terms by the rule $e^{-V(\phi^p)} \to Z(0)$, and a new rest term with one more tree edge.

\paragraph{Tree expansion.}
Generalizing the above considerations, we obtain the following result.

\begin{proposition}
After $p$ interpolation steps the integral of $F$ against $\nu$ writes as a sum over recursive trees with distinct labels, where each vertex $i$ has an associated time $t_i$, and an integral rest term:
\[
\begin{split}
\nu(F) = \Braket{ F(\phi)  } =& \sum_{0\le q\le p-1}
\sum_{ \substack{ {T \in \pmb{T}  } \\ { |V (T)|= 1 + q  } } }
\int_{-\infty}^{t_0} dt_1 \dots \int_{-\infty}^{t_{q-1}} dt_q 
 \Bigg[ \exp\Bigg\{ \frac{1}{2} \sum_{i,j=0}^q \delta_{\phi^i} \frac{e^{- |t_i -t_j| A}}{A} \delta_{\phi^j} \Bigg\} 
\crcr
&
 \qquad \qquad  \prod_{(i,j) \in E(T) }  \left( \delta_{\phi^i} e^{- (t_i -t_j) A}  \delta_{\phi^j } \right) F(\phi^0) \prod_{i = 1}^q (- V |_{\phi^i} )
 \Bigg]_{\phi=0}  + R_{p}(F)  \;,
 \end{split}
\]
where, in the $q$-th summand, $V(T) = \{0,1,\dots,q\}$ is the set of vertices of $T$. We recall that $\pmb{T}$ denotes the set of recursive trees, and the integral rest term is a sum over trees with one more vertex:
\[
\begin{split}
R_{p}(F)     = & \frac{1}{Z(0)}\sum_{ \substack{ {T \in \pmb{T}  } \\ { |V(T)|=p +1 } } }
\int_{-\infty}^{t_0} dt_1 \dots \int_{-\infty}^{t_{p-1}} dt_p 
 \Bigg[ \exp\Bigg\{ \frac{1}{2} \sum_{i,j=0}^p \delta_{\phi^i }  \frac{e^{- |t_i -t_j| A}}{A} \delta_{\phi^j } \bigg\}  \; e^{ - V(\phi^p)}
\crcr
& \qquad \qquad
\qquad    \prod_{(i,j) \in E(T) }  \left( \delta_{\phi^i} e^{- ( t_i -t_j ) A} \delta_{\phi^j } \right) 
F(\phi^0)  \prod_{i = 1}^{p} ( -V |_{\phi^i} )
 \Bigg]_{\phi=0} \; .
\end{split}
\]
Recall that the root of the recursive trees is labeled $0$ and the degree is not restricted to be $1$.
The time $t_0$ of the vertex $0$ corresponding to the observable $F(\phi^0)$ is not integrated. The first term in the sum consists in the tree with only the root vertex $0$, and contributes
$ \Big[ \exp\big\{ \frac{1}{2} \delta_{\phi^0} \frac{1}{A} \delta_{\phi^0} \big\} F(\phi^0) \Big]_{\phi^0=0}$.
\end{proposition}

\begin{proof}
The cases $p=1,2$ were shown above. 
Assuming the result at step $p$, split the field copy carrying $e^{-V(\phi^p)}$ in $R_p(F)$ and apply the first-order Taylor formula in $t_{p+1}$ as described above. The boundary term gives precisely the $q=p$ summand, while differentiating a mixed covariance chooses the unique parent of the new vertex and gives $R_{p+1}(F)$. This proves the formulas by induction. Since the recursive labels are all distinct and the parent of the new largest label is chosen exactly once, every recursive tree occurs once; hence no symmetry factors appear.
\end{proof}

\paragraph{Emergence of the stochastic field from the Hubbard--Stratonovich transformation.} 

The next step consists in reworking the Gaussian integral. First, we note that the covariance for the multiple copies of the field is positively defined:
\[
\frac{1}{2} \sum_i f_i \frac{e^{- | t_i-t_i | A}}{A} f_i
+  \sum_{i < j} f_i \frac{e^{- | t_i-t_j | A}}{A} f_j
 = \int_{u}  
   \left[ \sum_{i}  \theta(t_i - u ) (e^{ - ( t_i -u ) A}) f_i \right]^2 \ge 0 \;, 
\]
where we used:
\[
\int_{u,x}  du \;  \theta(t_i-u) (e^{-(t_i-u)A})_{x_ix} \; \theta(t_j- u) (e^{-(t_j-u)A})_{xx_j}
= \frac{1}{2} \left( \frac{e^{- |t_i-t_j| A  }}{A}  \right)_{x_ix_j} \; .
\]

We can thus apply a Hubbard--Stratonovich transformation to the Gaussian integral to get (reinstating the positions for a moment):
\[
 e^{ \frac{1}{2}\sum_i   \delta_{\phi^i} \frac{e^{-|t_i-t_i|A}}{A} \delta_{\phi^i}  +  \sum_{i<j} 
  \delta_{\phi^i} \frac{e^{-|t_i-t_j|A}}{A} \delta_{\phi^j} }= \left[ e^{ \int_{t,x} \frac{\delta}{ \delta \xi_{(t,x)} } \frac{\delta}{ \delta \xi_{(t,x)} } } e^{ \sum_{i}\int_{u;x,y}  
    \frac{\delta}{ \delta \phi^i_x }  \theta( t_i - u ) (e^{ - ( t_i -u ) A})_{xy} \xi_{(u,y)} }  \right]_{\xi=0} \;, 
\]
where $\xi_{(t,x)}$ is a white noise with covariance $2\delta_{tt'} \delta_{xx'}$. Indeed, denoting $ \int_{t,x} \frac{\delta}{ \delta \xi_{(t,x)} } \frac{\delta}{ \delta \xi_{(t,x)} } \equiv \int_t 
\frac{\delta}{ \delta \xi_t } \frac{\delta}{ \delta \xi_t }   $ and recognizing the Green function of the diffusion operator $\pmb{C}_{t,s}  = \theta(t-s) e^{-(t-s)A} $ we have:
\[
\begin{split}
& \left[ e^{\int_t 
\frac{\delta}{ \delta \xi_t } \frac{\delta}{ \delta \xi_t } } e^{\sum_i \int_u \delta_{\phi^i} \pmb{C}_{t_i,u} \xi_{u} } \right]_{\xi=0}  = 
 \sum_{n\ge 0} \frac{1}{n!} \left(\int_t 
\frac{\delta}{ \delta \xi_t } \frac{\delta}{ \delta \xi_t }  \right)^n \frac{1}{(2n)!} \left(\sum_i \int_u \delta_{\phi^i} \pmb{C}_{t_i,u} \xi_{u}  \right)^{2n} \crcr
& =\sum_n\frac{1}{ n!} \left(\sum_{i,j} \int_u ( \delta_{\phi^i} \pmb{C}_{t_i,u} ) ( \delta_{\phi^j} \pmb{C}_{t_j,u} )  \right)^n
 = \exp \left\{\frac{1}{2} \sum_{i,j} \delta_{\phi^i} \frac{e^{-|t_i-t_j|A}}{A} \delta_{\phi^j}    \right\} \; .
\end{split}
\]
Commuting the white noise integral with the sum over trees and sending $p$ to infinity, we obtain:
\[
\begin{split}
\nu(F) =  \Braket{ F(\phi) } = & \Bigg[ e^{\int_t \frac{\delta}{ \delta \xi_t } \frac{\delta}{ \delta \xi_t } }
 \sum_{q\ge 0 }
\sum_{ \substack{ {T \in \pmb{T}  } \\ { |V(T)|=q +1 } } }
\int_{-\infty}^{t_0} dt_1 \dots \int_{-\infty}^{t_{q-1}} dt_q \;
 \exp\Bigg\{
  \sum_{i=0}^q \int _u \; \frac{\delta}{\delta \phi^i} \pmb{C}_{t_i,u} \xi_u \Bigg\}
\crcr
&
 \qquad \qquad  \prod_{(i,j) \in E(T) }  \left( \delta_{\phi^i} e^{- (t_i -t_j) A}  \delta_{\phi^j } \right) F(\phi^0) \prod_{i = 1}^q \big(- V (\phi^i) \big)
 \Bigg]_{\phi=0,\xi=0} \;.
\end{split}
\]
The $\phi$ integrals can now be computed as one recognizes that they yield field translations:
\[
\Bigg[e^{ \int_{u} 
    \frac{\delta}{ \delta \phi } \pmb{C}_{t,u} \xi_u
  } H(\phi) \Bigg]_{\phi=0}
=  H\left(\phi + \int_{u}   \pmb{C}_{t,u} \xi_u  \right)\Big{|}_{\phi=0} = H(\phi)\Big{|}_{\phi= \int_{u}   \pmb{C}_{t,u} \xi_u   } \; ,
\]
and using this transformation term by term in the tree expansion, we obtain:
\begin{equation}\label{eq:recHubStrat}
 \begin{split}
 \nu(F) = \Braket{ F(\phi) } = & \Bigg[ e^{\int_t \frac{\delta}{ \delta \xi_t } \frac{\delta}{ \delta \xi_t } }
 \sum_{q\ge 0}
\sum_{ \substack{ {T \in \pmb{T}  } \\ { |V(T)|=q+1} } }
\int_{-\infty}^{t_0} dt_1 \dots \int_{-\infty}^{t_{q-1}} dt_q
\crcr
& \qquad
\left\{ \prod_{(i,j) \in E(T) }  \left( \delta_{\phi^i} e^{- (t_i -t_j) A}  \delta_{\phi^j } \right) F(\phi^0) \prod_{i = 1}^q \big(- V (\phi^i) \big) \right\}_{\phi^j = \int_u \pmb{C}_{t_j,u} \xi_u }
 \Bigg]_{\xi=0} \;.
\end{split}
\end{equation}
Note that, when sending $p$ to infinity, we have discarded the rest term.  A complete treatment would require keeping it and establishing appropriate bounds, which entails specific choices for the quadratic part $A$, the interaction $V$, the underlying space and so on.

\paragraph{The Langevin equation revisited.}
In Sec.~\ref{sec:stoch-exp} we wrote the solution of the 
non linear Langevin equation in Eq.~\eqref{eq:Langevin} and Eq.~\eqref{eq:SPDE} as a sum over trees obtained by Taylor expanding the interaction $V(\phi)$ around the zero field configuration. We present here a different solution of this equation as a tree expansion obtained by Taylor expanding around the free theory solution of the Langevin equation.

Our starting point is the observation that, as $(\partial_t +A) \pmb{C}_{t,s}  = \delta_{t,s} $, the solution of the non linear Langevin equation \eqref{eq:Langevin} and \eqref{eq:SPDE} can also be written as:
\[
\Phi_t  = \int_s \; \pmb{C}_{t,s} \xi_s +  \int_s
   \pmb{C}_{t,s} \;  (- \delta_\phi V)(\Phi_s) \;.
\]

We separate this into the the free theory part  $\phi_t =  \int_s \; \pmb{C}_{t,s} \xi_s $ and the displacement due to the presence of the interaction $\chi_t =  \int_s  \pmb{C}_{t,s} (- \delta_\phi V)(\Phi_s)
= \int_s  \pmb{C}_{t,s} (- \delta_\phi V)(\phi_s + \chi_s)$. Taylor expanding in $\chi_s$ we obtain (reinstating for a moment the position variables):
\begin{equation}\label{eq:chi}
\begin{split}
 \chi_t = \Phi_t - \phi_t & =
  \int_s \pmb{C}_{t,s} (- \delta_\phi V)(\phi_s + \chi_s) \crcr
  & = \sum_{\ell\ge 0} \left(-\frac{1}{\ell!} \right) 
   \int_{y}
  \int_s
   \pmb{C}_{(t,x) (s,y) }    \int_{ y_1,\dots y_\ell} 
   \left( \frac{ \delta^{\ell+1} V }{ \delta\phi_{y_1} \dots \delta \phi_{y_\ell} \delta \phi_y } \right) (\phi_s ) \; \chi_{(s,y_1)} \dots \chi_{(s,y_\ell)} \; .
\end{split}
\end{equation}

We solve this equation by iterated substitutions and find $\Phi_t$ as a sum over rooted plane trees with root $r$ of degree at most $1$ and root time $t_r = t$:
\begin{equation}\label{eq:Phi_chi}
\begin{split}
 \Phi_t = &
  \sum_{ {\cal T} \in \pmb{\cal T}^{(r|t)} }
 \prod_{v\in V^{\rm int}( {\cal T} ) } \frac{1}{ ( d_{\cal T}(v)-1) !} \crcr
& \qquad    \int_{-\infty}^{\infty} \prod_{v\in   V^{\rm int}( {\cal T} )  }  dt_v
\left[   \prod_{(v^q,w^1) \in E( {\cal T} ) }
 (  \delta_{\phi_{t_v}} \pmb{C}_{t_v,t_w} \delta_{\phi_{t_w}} )
   \; \; \phi_{t_r}
 \prod_{v\in   V^{\rm int}( {\cal T} )  }
 ( -  V (\phi_{t_v} ) ) \right]_{\phi_{t_v} = \int_u \pmb{C}_{t_v,u} \xi_u}\; ,
\end{split}
\end{equation}
where the first term in this sum, which corresponds to the tree with only the root vertex and no edge, reproduces $\phi_t$.\footnote{We note that in the tree expansion in Sec.~\ref{sec:stoch-exp} this term is zero and the lowest order tree which contributes there connects the root with a univalent $\xi$ vertex.}
Eq.~\eqref{eq:Phi_chi} follows from \eqref{eq:chi} in a similar fashion to how we obtained \eqref{eq:SPDTsoltree}, but
the advantage of this new expansion with respect to the one of Sec.~\ref{sec:stoch-exp} is that the white noise field $\xi$ arises now only via the evaluation of the differential operators on the free solution $\phi_{t_v} = \int_u \pmb{C}_{t_v,u} \xi_u$.

Using Lemma~\ref{lemma:trees} from Sec.~\ref{sec:trees_ODEs}, 
the sum over rooted plane trees with univalent root can be traded for a sum over rooted recursive trees. 
We denote the root vertex of the recursive trees by $0$ and note that it has degree at most $1$; the internal vertices of the recursive tree are denoted $V^{\rm int}(T) = \{ 1,2, \dots\}$ and we obtain: 
\begin{equation}\label{eq:Phi_recur}
 \Phi_{t_0} = \sum_{q\ge 0}
\sum_{ \substack{ {T \in \pmb{T}  } \\ { |V(T)|=q + 1}  \\ {d_T(0) \le 1} } }
 \int_{-\infty}^{t_0} dt_1 \dots  \int_{-\infty}^{t_{q-1}} dt_q  
 \left[ \prod_{(i,j) \in E(T), }  \left( \delta_{\phi^i} \pmb{C}_{t_i,t_j} \delta_{\phi^j } \right) \phi^0 \prod_{i = 1}^q \big(- V (\phi^i) \big) \right]_{\phi^j =
  \int du \; \pmb{C}_{t_j,u} \xi_u }\; ,
\end{equation}
where the first term $q=0$ reproduces $\phi_{t_0}$ and the sum starting at $q=1$ is $\chi_{t_0}$. A monomial writes as the product of $n$ such recursive tree expansions (where we reinstate the external positions):
\begin{align*}
&\Phi_{(t_0,x^1)} \dots \Phi_{(t_0,x^n)}  =  \sum_{q^1, \dots q^n \ge 0}
\sum_{ \substack{ {T^1, \dots T^n \in \pmb{T}  } \\ { |V(T^\nu)|=q^\nu + 1}  \\ {d_{T^\nu}(0) \le 1} } } \\
&\qquad \prod_{\nu=1}^n 
\Big\{  \int_{-\infty}^{t_{0^\nu}=t_0} dt_{1^\nu} \dots  \int_{-\infty}^{t_{q^\nu-1}} dt_{q^\nu}  \left[ \prod_{(i^\nu,j^\nu) \in E(T^\nu) }  \left( \delta_{\phi^{i^\nu}} \pmb{C}_{t_{i^\nu} ,t_{j^\nu}} \delta_{\phi^{j^\nu} } \right) \phi^{0^\nu}_{x^\nu} \prod_{i = 1}^{q^\nu} \big(- V (\phi^{i^\nu}) \big) \right]_{\phi^{i^\nu} =
  \int du \; \pmb{C}_{t_{i^\nu} ,u} \xi_u } \Big\} \; . \nonumber
\end{align*}
Joining together the trees $T^\nu$ at the root we obtain a tree $T$ with $1+q^1+\dots +q^n$ vertices and with a branch for each $T^\nu$, hence the degree of the root in $T$ is lower than or equal to $n$. While each branch $T^\nu$ is recursively ordered, $T$ itself is not, but summing over all the possible total orders among the times $\{t_{1^\nu},\dots t_{q^\nu} \}_{\nu=1,\dots n}$ we obtain a sum over all the recursively labeled $T$s (see below for a detailed proof):
\begin{align}\label{eq:onerectree}
&\Phi_{(t_0,x^1)} \dots \Phi_{(t_0,x^n)} =\\
&\qquad  \sum_{q \ge 0}
\sum_{ \substack{ {T \in \pmb{T}  } \\ { |V(T)|=q +1 } \\{d_T(0) \le n } } }  \int_{-\infty}^{t_0} dt_1 \dots  \int_{-\infty}^{t_{q-1}} dt_q 
  \left[ \prod_{(i,j) \in E(T) }  \left( \delta_{\phi^i} \pmb{C}_{t_i,t_j}  \delta_{\phi^j } \right) \prod_{\nu=1}^n\phi^0_{x^\nu} \prod_{i = 1}^q \big(- V (\phi^i) \big) \right]_{\phi^i =
  \int du \; C_{t_i,u} \xi_u
 } \;. \nonumber
\end{align}
We then have for a generic observable:
\begin{equation}\label{eq:recobsfin}
\begin{split}
 F(\Phi_{t_0}) =
  \sum_{q \ge 0}
\sum_{ \substack{ {T \in \pmb{T}  } \\ { |V(T)|=q +1 } } } & \int_{-\infty}^{t_0} dt_1 \dots  \int_{-\infty}^{t_{q-1}} dt_q \crcr
&  \left[ \prod_{(i,j) \in T }  \left( \delta_{\phi^i} \pmb{C}_{t_i,t_j}  \delta_{\phi^j } \right) F(\phi^0) \prod_{i = 1}^q \big(- V (\phi^i) \big) \right]_{\phi^i =
  \int du \; C_{t_i,u} \xi_u
 } \;,
\end{split}
\end{equation}
and comparing Eq.~\eqref{eq:recobsfin} with Eq.~\eqref{eq:recHubStrat} we conclude that
$\nu(F) =\braket{F(\phi)} = \left[ e^{\int_t \partial_{\xi_t} \partial_{\xi_t} }  F( \Phi_{t_0} )\right]$, thus 
completing the proof of Theorem~\ref{thm:main}.


\paragraph{Proof of Eq.~\eqref{eq:onerectree}.}
We first prove a slightly more precise identity, in which the $n$ root fields
$\phi^{0^1},\dots,\phi^{0^n}$ are kept distinguished.
When acting on $\phi^{0^\nu}$ with differential operators, these fields
are all evaluated at $\int du \; \pmb{C}_{t_0,u} \xi_u $ as in \eqref{eq:onerectree},
which we suppress from our notation.

We need the following minor generalization of recursive trees:
a \emph{generalized recursive tree} is a tree over $p$ vertices having $p$ strictly ordered labels $k_0<k_1<\dots <k_{p-1}$, with $k_i\geq 0$ an integer, such that if the vertex $k_i$ belongs to the path from the vertex  $k_l$ to the vertex with minimal label $k_0$, then $k_i<k_l$.
To every generalized recursive tree there is a unique recursive tree, called its \emph{reduction}, where we relabel the vertices by $k_i \mapsto i$ for $i=0,\ldots, p-1$.

For a recursive tree $T$, let $\mathrm{Br}(T)$ be the set of generalized recursive trees
given by the connected components of $T$ after removing
the root vertex $0$ and edges attached to it
and retaining the labels assigned to the vertices inside $T$.
We call elements of $\mathrm{Br}(T)$ the
\emph{branches} of $T$.
The cardinality of $\mathrm{Br}(T)$ is $d_T(0)$.
See Fig.~\ref{fig:branches_of_T_three} for an example.

\begin{figure}[ht!]
\begin{center}
\begin{tikzpicture}[scale=0.4]

\begin{scope}[every node/.style={circle, thick,draw,font=\tiny,minimum height=1em}]
    \node (R)  at (0,0)    {0};
    \node (A1) at (-3,2)   {1};
    \node (A2) at (-3,4)   {3};
    \node (B1) at (0,2)    {4};
    \node (C1) at (3,2)    {2};
    \node (C2) at (3,4)    {5};

    \node (D1) at (12,2)   {1};
    \node (D2) at (12,4)   {3};

    \node (E1) at (17,2)   {4};

    \node (F1) at (22,2)   {2};
    \node (F2) at (22,4)   {5};
\end{scope}

\begin{scope}[every edge/.style={draw=black,thick}]
    \draw (R)--(A1);
    \draw (A1)--(A2);
    \draw (R)--(B1);
    \draw (R)--(C1);
    \draw (C1)--(C2);

    \draw (D1)--(D2);
    \draw (F1)--(F2);
\end{scope}

\begin{scope}[every node/.style={draw=none,font=\small}]
    \node at (0,-1.8) {$T=\{(0,1),(1,2),(0,3),(0,4),(4,5)\}$};
    \node at (7.2,2.2) {$\Longrightarrow$};

    \node at (12,5.4) {$B_{13}$};
    \node at (17,3.4) {$B_4$};
    \node at (22,5.4) {$B_{25}$};

    \node at (17,-1.8) {$\mathrm{Br}(T)=\{B_{13},B_4,B_{25}\}$};
\end{scope}

\end{tikzpicture}
\end{center}
\caption{A recursive tree $T$ and its set of branches. The branches $B_{13}$ and $B_{25}$ are isomorphic as rooted trees, but they are distinct as generalized recursive trees and as elements of $\mathrm{Br}(T)$.}
\label{fig:branches_of_T_three}
\end{figure}

Now consider $T\in\pmb T$ a recursive tree with $d_T(0)\le n$.
Let $\mathcal I_n(T)$ be the set of injections
\[
\alpha\colon \mathrm{Br}(T) \hookrightarrow \{1,\dots,n\}
\;,
\]
see Fig.~\ref{fig:alpha_examples_three} for examples.
The cardinality of $\mathcal{I}_n(T)$ is $n!/(n-d_T(0))!$.

\begin{figure}[ht!]
\begin{center}
\begin{tikzpicture}[scale=0.38]

\begin{scope}[xshift=0cm]
    \begin{scope}[every node/.style={circle, thick,draw,font=\tiny,minimum height=1em}]
        \node (A1) at (0,2)   {1};
        \node (A2) at (0,4)   {3};
        \node (B1) at (3,2)   {4};
        \node (C1) at (6,2)   {2};
        \node (C2) at (6,4)   {5};

        \node (S1) at (0,-2)  {1};
        \node (S2) at (2,-2)  {2};
        \node (S3) at (4,-2)  {3};
        \node (S4) at (6,-2)  {4};
    \end{scope}

    \begin{scope}[every edge/.style={draw=black,thick}]
        \draw (A1)--(A2);
        \draw (C1)--(C2);

        \draw[dashed] (A1.south) -- (S1.north);
        \draw[dashed] (B1.south) -- (S2.north);
        \draw[dashed] (C1.south) -- (S4.north);
    \end{scope}

    \begin{scope}[every node/.style={draw=none,font=\small}]
        \node at (3,6.2) {$\alpha_1$};
        \node at (0,5.4) {$B_{13}$};
        \node at (3,3.4) {$B_4$};
        \node at (6,5.4) {$B_{25}$};
        \node at (3,-3.8) {$\alpha_1(B_{13})=1,\ \alpha_1(B_4)=2,\ \alpha_1(B_{25})=4$};
    \end{scope}
\end{scope}

\begin{scope}[xshift=20cm]
    \begin{scope}[every node/.style={circle, thick,draw,font=\tiny,minimum height=1em}]
        \node (A1r) at (0,2)   {1};
        \node (A2r) at (0,4)   {3};
        \node (B1r) at (3,2)   {4};
        \node (C1r) at (6,2)   {2};
        \node (C2r) at (6,4)   {5};

        \node (S1r) at (0,-2)  {1};
        \node (S2r) at (2,-2)  {2};
        \node (S3r) at (4,-2)  {3};
        \node (S4r) at (6,-2)  {4};
    \end{scope}

    \begin{scope}[every edge/.style={draw=black,thick}]
        \draw (A1r)--(A2r);
        \draw (C1r)--(C2r);

        \draw[dashed] (A1r.south) -- (S4r.north);
        \draw[dashed] (B1r.south) -- (S3r.north);
        \draw[dashed] (C1r.south) -- (S1r.north);
    \end{scope}

    \begin{scope}[every node/.style={draw=none,font=\small}]
        \node at (3,6.2) {$\alpha_2$};
        \node at (0,5.4) {$B_{13}$};
        \node at (3,3.4) {$B_4$};
        \node at (6,5.4) {$B_{25}$};
        \node at (3,-3.8) {$\alpha_2(B_{13})=4,\ \alpha_2(B_4)=3,\ \alpha_2(B_{25})=1$};
    \end{scope}
\end{scope}

\end{tikzpicture}
\end{center}
\caption{Two example injections $\alpha_1,\alpha_2 \in \mathcal I_{4}(T)$.
There are in total $4\cdot3\cdot2=\frac{4!}{1!}$ such injections.}
\label{fig:alpha_examples_three}
\end{figure}

For $\alpha\in\mathcal I_n(T)$, we denote by $\mathcal A_n(T,\alpha)$ the
differential operator obtained from the integrand in
Eq.~\eqref{eq:onerectree} by letting the root
derivative corresponding to the branch $B \in \mathrm{Br}(T)$ act on the distinguished root
field $\phi^{0^{\alpha(B)}}_{x^{\alpha(B)}}$.
If $\mathrm{Br}(T)$ is empty, we use the convention $\mathcal A_n(T,\alpha) = \mathrm{id}$.

We claim that
\begin{align}\label{eq:ordered_root_version}
\Phi_{(t_0,x^1)}\cdots \Phi_{(t_0,x^n)}
=
\sum_{q\ge 0}
\sum_{\substack{T\in\pmb T\\ |V(T)|=q+1\\ d_T(0)\le n}}
\sum_{\alpha\in\mathcal I_n(T)}
\mathcal A_n(T,\alpha)
\Bigl(
\prod_{\nu=1}^n \phi^{0^\nu}_{x^\nu}
\Bigr)
\;.
\end{align}
Before proving \eqref{eq:ordered_root_version},
we note that it implies 
Eq.~\eqref{eq:onerectree}.
Indeed, for a fixed tree $T$ with $d_T(0)=m$, summing over
$\alpha\in\mathcal I_n(T)$ corresponds to choosing which of the $m$ root derivatives acts on which factor among
$\phi^{0^1}_{x^1},\dots,\phi^{0^n}_{x^n}$.
After identifying all root fields $\phi^{0^\nu}=\phi^0$, this is exactly the
result of applying the root part of the operator associated with $T$ to $\prod_{\nu=1}^n \phi^0_{x^\nu}$.
Hence~\eqref{eq:ordered_root_version} implies~\eqref{eq:onerectree}.

It remains to prove Eq.~\eqref{eq:ordered_root_version}, which we do by induction on $n$.
For $n=1$, the possible choices for $T$ are the tree with a single vertex, in which case $\mathcal A_1(T,\alpha)\phi^{0^1}_{x^1}=\phi^{0^1}_{x^1}$,
or a tree with one branch, in which case $\mathcal I_1(T)$ contains only the unique map $\alpha\colon \{1\}\to\{1\}$.
Hence Eq.~\eqref{eq:ordered_root_version} for $n=1$ is exactly the recursive tree
expansion \eqref{eq:Phi_recur} of $\Phi_{(t_0,x^1)}$.
In particular, we can write
\begin{equation}\label{eq:alpha_S}
\Phi_{(t_0,x)} = \sum_{S}\mathcal{A}(S)\phi^{0}_{x} \;,
\end{equation}
where $\mathcal A(S)$ is the contribution from the recursive tree $S$ with $d_S(0)\leq 1$ to the sum \eqref{eq:Phi_recur}, understood as a differential operator acting on $\phi^{0}_{x}$.

Assume now that Eq.~\eqref{eq:ordered_root_version} holds for some $n\ge 1$.
Multiplying by the recursive tree expansion \eqref{eq:alpha_S} of $\Phi_{(t_0,x^{n+1})}$, 
we obtain that
\begin{equation}\label{eq:n_plus_1}
\Phi_{(t_0,x^1)}\cdots 
\Phi_{(t_0,x^{n})}
\Phi_{(t_0,x^{n+1})}
=
\sum_{(T,\alpha,S)}
\mathcal A_{n}(T,\alpha)
\Bigl(
\prod_{\nu=1}^{n} \phi^{0^\nu}_{x^\nu}
\Bigr)
\mathcal A(S)
\phi^{0^{n+1}}_{x^{n+1}} \;,
\end{equation}
where 
the sum is over triples $(T,\alpha,S)$, where $(T,\alpha)$ comes from the induction hypothesis
and $S$ comes from the last factor.

We now let $\mathcal{R}(T,\alpha,S)$ denote the set of pairs $(T',\alpha')$ where $T'$ is a recursive tree and
$\alpha' \in \mathcal{I}_{n+1}(T')$, such that
the reduction of $B = (\alpha')^{-1}(n+1)$
is equal to the reduction of the branch of $S$ (which is possibly empty)
and the reduction of the tree $T'\setminus B$ (understood as the generalized recursive tree obtained by removing the branch $B$ and retaining the rest of its vertex labelings)
is equal to $T$,
and that $\alpha'$ restricted to $\mathrm{Br}(T')\setminus \{B\}$ is equal to $\alpha$.
See Fig.~\ref{fig:R_example_three} for an example.
Because $S$ is a recursive tree, the cardinality of $\mathcal{R}(T,\alpha,S)$  is $\binom{|V^{\rm int}(T)|+|V^{\rm int}(S)|}{|V^{\rm int}(S)|}$.

\begin{figure}[ht!]
\begin{center}
\begin{tikzpicture}[scale=0.34]

\begin{scope}[every node/.style={circle, thick,draw,font=\tiny,minimum height=1em}]
    \node (R)  at (0,0)    {0};
    \node (A1) at (-3,2)   {1};
    \node (A2) at (-3,4)   {3};
    \node (B1) at (0,2)    {4};
    \node (C1) at (3,2)    {2};
    \node (C2) at (3,4)    {5};

    \node (RS) at (8,0) {0};
    \node (S1) at (8,2) {1};
    \node (S2) at (8,4) {2};
\end{scope}

\begin{scope}[every edge/.style={draw=black,thick}]
    \draw (R)--(A1);
    \draw (A1)--(A2);
    \draw (R)--(B1);
    \draw (R)--(C1);
    \draw (C1)--(C2);

    \draw (RS)--(S1);
    \draw (S1)--(S2);
\end{scope}

\begin{scope}[every node/.style={draw=none,font=\small}]
    \node at (0,-1.8) {$T$};
    \node at (0,-3.5) {$\alpha(B_{13})=1,\ \alpha(B_4)=2,\ \alpha(B_{25})=4$};
    \node at (8,-1.8) {$S$};
\end{scope}
\end{tikzpicture}

\vspace{0.5cm}

\begin{tikzpicture}[scale=0.34]
\begin{scope}[xshift=0cm]
\begin{scope}[every node/.style={circle, thick,draw,font=\tiny,minimum height=1em}]
    \node (R1)  at (0,0)    {0};
    \node (X11) at (-3,2)   {3};
    \node (X12) at (-3,4)   {5};
    \node (X13) at (0,2)    {6};
    \node (X14) at (3,2)    {4};
    \node (X15) at (3,4)    {7};
    \node (X16) at (6,2)    {1};
    \node (X17) at (6,4)    {2};
\end{scope}
\begin{scope}[every edge/.style={draw=black,thick}]
    \draw (R1)--(X11);
    \draw (X11)--(X12);
    \draw (R1)--(X13);
    \draw (R1)--(X14);
    \draw (X14)--(X15);
    \draw[dashed] (R1)--(X16.south);
    \draw (X16)--(X17);
\end{scope}
\end{scope}
\begin{scope}[xshift=15cm]
\begin{scope}[every node/.style={circle, thick,draw,font=\tiny,minimum height=1em}]
    \node (R2)  at (0,0)    {0};
    \node (Y21) at (-3,2)   {1};
    \node (Y22) at (-3,4)   {4};
    \node (Y23) at (0,2)    {6};
    \node (Y24) at (3,2)    {2};
    \node (Y25) at (3,4)    {7};
    \node (Y26) at (6,2)    {3};
    \node (Y27) at (6,4)    {5};
\end{scope}
\begin{scope}[every edge/.style={draw=black,thick}]
    \draw (R2)--(Y21);
    \draw (Y21)--(Y22);
    \draw (R2)--(Y23);
    \draw (R2)--(Y24);
    \draw (Y24)--(Y25);
    \draw[dashed] (R2)--(Y26.south);
    \draw (Y26)--(Y27);
\end{scope}
\end{scope}

\begin{scope}[xshift=29cm]
\begin{scope}[every node/.style={circle, thick,draw,font=\tiny,minimum height=1em}]
    \node (R3)  at (0,0)    {0};
    \node (Z31) at (-3,2)   {1};
    \node (Z32) at (-3,4)   {3};
    \node (Z33) at (0,2)    {5};
    \node (Z34) at (3,2)    {2};
    \node (Z35) at (3,4)    {6};
    \node (Z36) at (6,2)    {4};
    \node (Z37) at (6,4)    {7};
\end{scope}
\begin{scope}[every edge/.style={draw=black,thick}]
    \draw (R3)--(Z31);
    \draw (Z31)--(Z32);
    \draw (R3)--(Z33);
    \draw (R3)--(Z34);
    \draw (Z34)--(Z35);
    \draw[dashed] (R3)--(Z36.south);
    \draw (Z36)--(Z37);
\end{scope}
\end{scope}

\end{tikzpicture}

\end{center}
\caption{The top diagram shows an example of $(T,\alpha,S)$ for $n=4$ with $(T,\alpha)$ taken from the left diagram in Fig~\ref{fig:alpha_examples_three}.
The bottom diagram shows 3 examples of pairs $(T',\alpha')\in\mathcal R(T,\alpha,S)$.
There are in total $\binom{7}{2}$ such pairs in $\mathcal R(T,\alpha,S)$.
The corresponding $\alpha'$ coincides on the left three branches with $\alpha$
and maps the `new' rightmost branch, connected to the root with a dotted line,
to $n+1=5$.}
\label{fig:R_example_three}
\end{figure}

Remark that the triple $(T,\alpha, S)$ is uniquely determined by any element of $\mathcal{R}(T,\alpha,S)$.
Conversely, every pair $(T',\alpha')$, where $T'$ is a recursive trees and $\alpha'\in\mathcal I_{n+1}(T')$,
belongs to $\mathcal{R}(T,\alpha,S)$ for a unique triple $(T,\alpha,S)$.
Consequently $\mathcal{R}(T,\alpha,S)$ with $(T,\alpha,S)$ running over all triples as in \eqref{eq:n_plus_1}
partitions the set of pairs $(T',\alpha')$.

We claim that, for fixed $(T,\alpha,S)$ as above,
\[
\mathcal A_{n}(T,\alpha)
\Bigl(
\prod_{\nu=1}^{n} \phi^{0^\nu}_{x^\nu}
\Bigr)
\mathcal A(S)
\phi^{0^{n+1}}_{x^{n+1}}
=
\sum_{(T',\alpha')\in\mathcal{R}(T,\alpha,S)}
\mathcal A_{n+1}(T',\alpha')
\Bigl(
\prod_{\nu=1}^{n+1} \phi^{0^\nu}_{x^\nu}
\Bigr)\;.
\]
Indeed, this follows from interlacing the times of $S$ into those of $T$ and remarking that every possible order of times corresponding to a unique recursive tree in $\mathcal{R}(T,\alpha,S)$.

Moreover, because $\mathcal{R}(T,\alpha,S)$ forms a partition of the set of pairs $(T',\alpha')$,
we have
\[
\sum_{(T,\alpha,S)}
\sum_{(T',\alpha')\in\mathcal{R}(T,\alpha,S)}
\mathcal A_{n+1}(T',\alpha')
\Bigl(
\prod_{\nu=1}^{n+1} \phi^{0^\nu}_{x^\nu}
\Bigr)
=
\sum_{(T',\alpha')}
\mathcal A_{n+1}(T',\alpha')
\Bigl(
\prod_{\nu=1}^{n+1} \phi^{0^\nu}_{x^\nu}
\Bigr)
\;.
\]
By \eqref{eq:n_plus_1}, this completes the induction step and proves \eqref{eq:ordered_root_version}.

\appendix

\section{Low orders of the Feynman expansion}
\label{app:loworderFeynm}

In order to get a concrete handle on the relation between the stochastic and path integral expansions, we provide in this and the next appendix some low-order example calculations, starting here from the Feynman expansion.

We consider the expansion \eqref{eq:Feynmanampli} for the particular case of a local $\phi^4$ model, for which the vertex kernels are:
\[ V(x^1,x^2,x^3,x^4) =  g \, \delta_{x^1x^2}\delta_{x^1x^3}\delta_{x^1x^4} \;, \] 
After using the delta functions, only one integral over a half-edge position $x^h_v$ survives for each vertex in  \eqref{eq:Feynmanampli}.
Denoting the remaining position $x_v$, we obtain:
\[
\Braket{\phi_{x^1} \dots \phi_{x^n} }  = \sum_{ {\cal G}   \in \pmb{\cal G}^{(r^1| x^1) \dots (r^n | x^n)}  } {\cal A}({\cal G}) \; , \qquad {\cal A}({\cal G})  = (-g)^{ V^{\rm int}({\cal G})}  \int \prod_{v\in V^{\rm int}({\cal G} )}d x_v \prod_{\{ v^h,w^{h'}\} \in E({\cal G})}C_{x_vx_w} \; ,
\]
where all the internal vertices have degree $4$. 

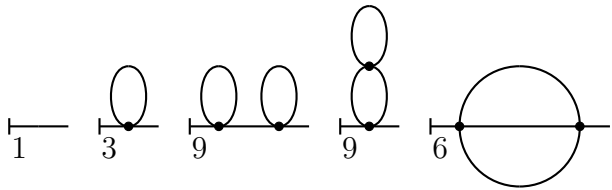
\begin{figure}[ht!]
\begin{center}
\begin{tikzpicture}[scale=0.4]

\draw[black, thick,|-] (-3,0) -- (-2,0) node[below left] {1};
\draw[black, thick] (-2,0) -- (-1,0);

\draw[black, thick,|-] (0,0) -- (1,0) node[below left] {3};
\filldraw[black] (1,0) circle (4pt);
\draw[black, thick] (1,0) -- (2,0);
\draw[black, thick] (1,0) to[out=0,in=0] (1,2);
\draw[black, thick] (1,0) to[out=180,in=180] (1,2);

\draw[black, thick,|-] (3,0) -- (4,0) node[below left] {9};
\filldraw[black] (4,0) circle (4pt);
\draw[black, thick] (4,0) to[out=0,in=0] (4,2);
\draw[black, thick] (4,0) to[out=180,in=180] (4,2);
\draw[black, thick] (4,0) -- (6,0);
\filldraw[black] (6,0) circle (4pt);
\draw[black, thick] (6,0) to[out=0,in=0] (6,2);
\draw[black, thick] (6,0) to[out=180,in=180] (6,2);
\draw[black, thick] (6,0) -- (7,0);

\draw[black, thick,|-] (8,0) -- (9,0) node[below left] {9};
\filldraw[black] (9,0) circle (4pt);
\draw[black, thick] (9,0) to[out=0,in=0] (9,2);
\draw[black, thick] (9,0) to[out=180,in=180] (9,2);
\draw[black, thick] (9,0) -- (10,0);
\filldraw[black] (9,2) circle (4pt);
\draw[black, thick] (9,2) to[out=0,in=0] (9,4);
\draw[black, thick] (9,2) to[out=180,in=180] (9,4);

\draw[black, thick,|-] (11,0) -- (12,0) node[below left] {6};
\filldraw[black] (12,0) circle (4pt);
\draw[black, thick] (12,0) to[out=90,in=180] (14,2);
\draw[black, thick] (14,2) to[out=0,in=90] (16,0);
\draw[black, thick] (12,0) to[out=-90,in=180] (14,-2);
\draw[black, thick] (14,-2) to[out=0,in=-90] (16,0);
\draw[black, thick] (12,0) -- (16,0);
\filldraw[black] (16,0) circle (4pt);
\draw[black, thick] (16,0) -- (17,0);

\end{tikzpicture}
\end{center}
    \caption{Feynman graphs at zero, first and second order contributing to the two point function in the $\phi^4$ model with their multiplicity counting the number of distinct embeddings in two-dimensional surfaces. From left to right, we have the one edge graph, the tadpole graph, two double tadpole graphs (obtained by inserting a tadpole into a tadpole) and the sunset graph.
    These are precisely the same graphs and coefficients as in Fig.~\ref{fig:examplegraphs}.}
    \label{fig:firstord}
\end{figure}

Let us focus on the two point function $\braket{\phi_{x^1} \phi_{x^2} }  $.
Due to invariance under the $\mathbb{Z}_2$ transformation $\phi\to-\phi$, only connected maps with two external vertices contribute. We observe that the amplitude 
is insensitive to the embedding\footnote{This is not the case for theories with non-local interactions or theories in which the field is a matrix, in which case the amplitudes detect the genus of the embedding surface \cite{'tHooft:1973jz,DiFrancesco:1993nw}.} and it is customary to add together all the embedded graphs that correspond to the same abstract graph, leading to the so called symmetry factors in the Feynman expansion (note that we have normalized the interaction by $1/4$ instead of the $1/4!$ more common in the physics literature). 
The terms up to second order in the coupling constant are depicted in Fig.~\ref{fig:firstord}, where the root vertex $r^1$ corresponds to the external $\phi_{x^1}$ and they contribute (all variables except $x^1$ and $x^2$ are integrated):
\[
\begin{split}
& C_{x^1x^2} + (-g)  \Big[ 3 \; C_{x^1z } C_{zz} C_{zx^2}  \Big]  \crcr
& + (-g)^2
 \Big[ 9 \;  C_{x^1z_1} C_{z_1z_1}C_{z_1z_2} C_{z_2z_2}C_{z_2x^2}  + 9 \; C_{x^1z_1} C_{z_1z_2}C_{z_2z_2}C_{z_2z_1}C_{z_1x^2}  + 6 \; C_{x^1z_1} C_{z_1z_2}^3 C_{z_2 x^2} \Big] \; .
 \end{split}
\]
From left to right the terms correspond to the one edge graph, the tadpole graph, two double tadpoles and the sunset graph.

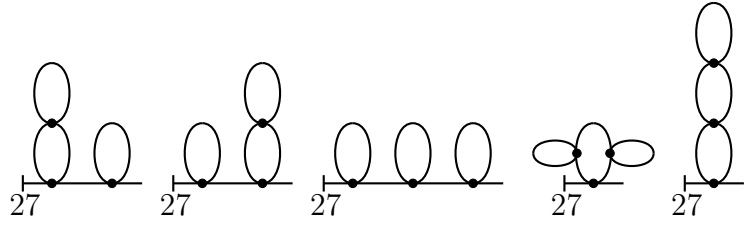
\begin{figure}[ht!]
\begin{center}
    
\begin{tikzpicture}[scale=0.4]

\draw[black, thick,|-] (0,0) -- (1,0) node[below left] {27};
\filldraw[black] (1,0) circle (4pt);
\draw[black, thick] (1,0) to[out=0,in=0] (1,2);
\draw[black, thick] (1,0) to[out=180,in=180] (1,2);

\filldraw[black] (1,2) circle (4pt);
\draw[black, thick] (1,2) to[out=0,in=0] (1,4);
\draw[black, thick] (1,2) to[out=180,in=180] (1,4);

\draw[black, thick] (1,0) -- (3,0);
\filldraw[black] (3,0) circle (4pt);
\draw[black, thick] (3,0) to[out=0,in=0] (3,2);
\draw[black, thick] (3,0) to[out=180,in=180] (3,2);

\draw[black, thick] (3,0) -- (4,0);

\draw[black, thick,|-] (5,0) -- (6,0) node[below left] {27};
\filldraw[black] (6,0) circle (4pt);
\draw[black, thick] (6,0) to[out=0,in=0] (6,2);
\draw[black, thick] (6,0) to[out=180,in=180] (6,2);

\draw[black, thick] (6,0) -- (8,0);
\filldraw[black] (8,0) circle (4pt);
\draw[black, thick] (8,0) to[out=0,in=0] (8,2);
\draw[black, thick] (8,0) to[out=180,in=180] (8,2);

\draw[black, thick] (8,2) to[out=0,in=0] (8,4);
\draw[black, thick] (8,2) to[out=180,in=180] (8,4);
\filldraw[black] (8,2) circle (4pt);

\draw[black, thick] (8,0) -- (9,0);

\draw[black, thick,|-] (10,0) -- (11,0) node[below left] {27};
\filldraw[black] (11,0) circle (4pt);
\draw[black, thick] (11,0) to[out=0,in=0] (11,2);
\draw[black, thick] (11,0) to[out=180,in=180] (11,2);

\draw[black, thick] (11,0) -- (13,0);
\filldraw[black] (13,0) circle (4pt);
\draw[black, thick] (13,0) to[out=0,in=0] (13,2);
\draw[black, thick] (13,0) to[out=180,in=180] (13,2);

\draw[black, thick] (13,0) -- (15,0);

\filldraw[black] (15,0) circle (4pt);
\draw[black, thick] (15,0) to[out=0,in=0] (15,2);
\draw[black, thick] (15,0) to[out=180,in=180] (15,2);

\draw[black, thick] (15,0) -- (16,0);

\draw[black, thick,|-] (18,0) -- (19,0) node[below left] {27};
\filldraw[black] (19,0) circle (4pt);
\draw[black, thick] (19,0) to[out=0,in=0] (19,2);
\draw[black, thick] (19,0) to[out=180,in=180] (19,2);
\filldraw[black] (18.45,1) circle (4pt);
\draw[black, thick] (18.45,1) to[out=90,in=90] (17,1);
\draw[black, thick] (18.45,1) to[out=-90,in=-90] (17,1);

\filldraw[black] (19.55,1) circle (4pt);
\draw[black, thick] (19.55,1) to[out=90,in=90] (21,1);
\draw[black, thick] (19.55,1) to[out=-90,in=-90] (21,1);

\draw[black, thick] (19,0) -- (20,0);

\draw[black, thick,|-] (22,0) -- (23,0) node[below left] {27};
\filldraw[black] (23,0) circle (4pt);
\draw[black, thick] (23,0) to[out=0,in=0] (23,2);
\draw[black, thick] (23,0) to[out=180,in=180] (23,2);

\filldraw[black] (23,2) circle (4pt);
\draw[black, thick] (23,2) to[out=0,in=0] (23,4);
\draw[black, thick] (23,2) to[out=180,in=180] (23,4);

\filldraw[black] (23,4) circle (4pt);
\draw[black, thick] (23,4) to[out=0,in=0] (23,6);
\draw[black, thick] (23,4) to[out=180,in=180] (23,6);

\draw[black, thick] (23,0) -- (24,0);

\end{tikzpicture}

\end{center}

   \caption{Triple tadpoles.}
    \label{fig:tripletadpoles}
\end{figure}

At next order we have triple tadpoles (represented in Fig.~\ref{fig:tripletadpoles}) which contribute:
\[
\begin{split}
 (-g)^3 \Big[ & 27 \;  C_{x^1z_1} C_{z_1z_3}C_{z_3z_3}C_{z_3z_1}C_{z_1z_2} C_{z_2z_2}C_{z_2x^2} + 27 \; C_{x^1z_1} C_{z_1z_1}C_{z_1z_2} C_{z_2z_3}C_{z_3z_3}C_{z_3z_2}C_{z_2x^2} \crcr
 & + 27  \; C_{x^1z_1} C_{z_1z_1}C_{z_1z_3}C_{z_3z_3}C_{z_3z_2} C_{z_2z_2}C_{z_2x^2}   + 27  \; C_{x^1z_1} C_{z_1z_2}C_{z_2z_2} C_{z_2 z_3 } C_{z_3z_3} C_{z_3z_1}C_{z_1x^2} \crcr
 & + 27 \; C_{x^1z_1} C_{z_1z_2}C_{z_2z_3}C_{z_3z_3} C_{z_3z_2} C_{z_2z_1}C_{z_1x^2}  \Big] \;,
\end{split}
\]
as well as tadpole-sunset graphs (obtained by inserting tadpole graphs into the sunset or the sunset graph into the tadpole) and a new primitive graph which can not be obtained by inserting previous graphs into each other, see Fig.~\ref{fig:remaining}, yielding:
\[
\begin{split}
  (-g)^3 \Big[ & 54 \; C_{x^1z_1} C_{z_1z_3} C_{z_3z_3}C_{z_3z_2}C_{z_1z_2} C_{z_2x^2}
  + 18  \;  C_{x^1z_3} C_{z_3z_3} C_{z_3z_1}C_{z_1z_2}^3 C_{z_2x^2}
  + 18 \; C_{x^1z_1} C_{z_1z_2}^3 C_{z_2z_3}C_{z_3z_3}C_{z_3x^2} \crcr
 & + 54 \; 
 C_{x^1z_1}  C_{z_1z_3}^2 C_{z_2z_3}^2 C_{z_1z_2} C_{z_2x^2}  +
 18 \; C_{x^1z_1} C_{z_1z_2} C_{z_2z_3}^3 C_{z_3z_1} C_{z_1 x^2} \Big] \; .
\end{split}
\]

\begin{figure}[ht!]
\begin{center}
\begin{tikzpicture}[scale=0.4]

\draw[black, thick,|-] (0,0) -- (1,0) node[below left] {54};

\filldraw[black] (1,0) circle (4pt);
\draw[black, thick] (1,0) to[out=90,in=180] (3,2);

\filldraw[black] (3,2) circle (4pt);
\draw[black, thick] (3,2) to[out=0,in=0] (3,4);
 \draw[black, thick] (3,2) to[out=180,in=180] (3,4);

\draw[black, thick] (3,2) to[out=0,in=90] (5,0);
\draw[black, thick] (1,0) to[out=-90,in=180] (3,-2);
\draw[black, thick] (3,-2) to[out=0,in=-90] (5,0);
\draw[black, thick] (1,0) -- (5,0);
\filldraw[black] (5,0) circle (4pt);
\draw[black, thick] (5,0) -- (6,0);

\draw[black, thick,|-] (7,0) -- (8,0) node[below left] {18};
\filldraw[black] (8,0) circle (4pt);
 \draw[black, thick] (8,0) to[out=0,in=0] (8,2);
 \draw[black, thick] (8,0) to[out=180,in=180] (8,2);

 \draw[black, thick] (8,0) -- (10,0);

\filldraw[black] (10,0) circle (4pt);
\draw[black, thick] (10,0) to[out=90,in=180] (12,2);
\draw[black, thick] (12,2) to[out=0,in=90] (14,0);
\draw[black, thick] (10,0) to[out=-90,in=180] (12,-2);
\draw[black, thick] (12,-2) to[out=0,in=-90] (14,0);
\draw[black, thick] (10,0) -- (14,0);
\filldraw[black] (14,0) circle (4pt);
\draw[black, thick] (14,0) -- (15,0);

 \draw[black, thick,|-] (16,0) -- (17,0) node[below left] {18};

\filldraw[black] (17,0) circle (4pt);
\draw[black, thick] (17,0) to[out=90,in=180] (19,2);
\draw[black, thick] (19,2) to[out=0,in=90] (21,0);
\draw[black, thick] (17,0) to[out=-90,in=180] (19,-2);
\draw[black, thick] (19,-2) to[out=0,in=-90] (21,0);
\draw[black, thick] (17,0) -- (21,0);
\filldraw[black] (21,0) circle (4pt);
\draw[black, thick] (21,0) -- (23,0);

\filldraw[black] (23,0) circle (4pt);

 \draw[black, thick] (23,0) to[out=0,in=0] (23,2);
 \draw[black, thick] (23,0) to[out=180,in=180] (23,2);

 \draw[black, thick] (23,0) -- (24,0);

\end{tikzpicture}

\begin{tikzpicture}[scale=0.4]

\draw[black, thick,|-] (0,0) -- (1,0) node[below left] {54};

\filldraw[black] (1,0) circle (4pt);
\draw[black, thick] (1,0) to[out=90,in=180] (3,2);

\filldraw[black] (3,2) circle (4pt);

\draw[black, thick] (3,2) to[out=0,in=90] (5,0);

\draw[black, thick] (1,0) to[out=0,in=-90] (3,2);

\draw[black, thick] (3,2) to[out=-90,in=180] (5,0);

\draw[black, thick] (1,0) to[out=-90,in=180] (3,-2);
\draw[black, thick] (3,-2) to[out=0,in=-90] (5,0);

\filldraw[black] (5,0) circle (4pt);
\draw[black, thick] (5,0) -- (6,0);

\draw[black, thick,|-] (8,0) -- (9,0) node[below left] {18};

\filldraw[black] (9,0) circle (4pt);

\draw[black, thick] (9,0) -- (10,0);

\draw[black, thick] (9,0) to[out=150,in=-90] (7,2);

\filldraw[black] (7,2) circle (4pt);

\draw[black, thick] (9,0) to[out=30,in=-90] (11,2);

\filldraw[black] (11,2) circle (4pt);

\draw[black, thick] (7,2) to[out=90,in=180] (9,3);
\draw[black, thick] (9,3) to[out=0,in=90] (11,2);

\draw[black, thick] (7,2) to[out=-90,in=180] (9,1);
\draw[black, thick] (9,1) to[out=0,in=-90] (11,2);

\draw[black, thick] (7,2) -- (11,2);

\end{tikzpicture}

\end{center}
    \caption{Tadpole-sunset graphs and a new primitive graph.}
    \label{fig:remaining}
\end{figure}
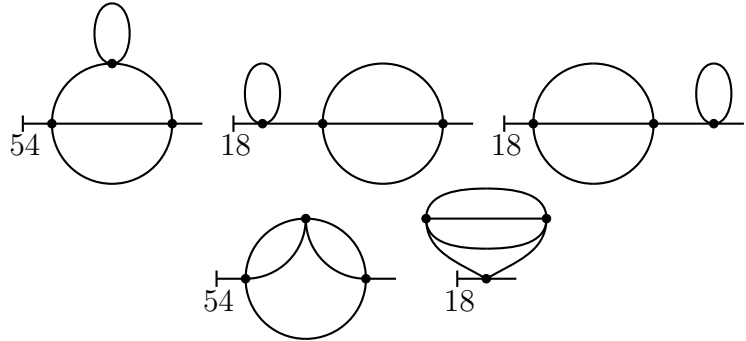

\begin{remark}
\textbf{Divergences and renormalization.}
\label{rem:renorm}
The perturbative procedure faces several issues, as a function of the particular model and the dimension $d$. 
One first problem, present even in zero dimensions with $A=1$, is that the series is divergent because the number of terms (i.e.\ the number of combinatorial maps) grows super exponentially and in order to rigorously make sense of the correlation functions one needs to use constructive field theory techniques which replace the perturbative expansion by convergent constructive expansions \cite{GlimmJaffe,Rivasseau:1991ub,salmhofer2007renormalization}.\footnote{This problem is also related to \emph{resurgence} \cite{Aniceto:2018bis}. For a link between resurgence and constructive methods in zero dimensions, see \cite{Benedetti:2022twd}.}

A second problem is that the integrals over the internal positions might be divergent, in which case one needs to renormalize the perturbative expansion. This depends not only on the interaction, here chosen to be a local $\phi^4$, but also on the covariance and the dimension $d$.
In the familiar case $A=C^{-1} = -\Delta +  m^2$, all the integrals over the internal positions are convergent for $d<2$. For $d=2$ the tadpoles diverge, but the divergences are subtracted to zero by Wick ordering, that is using as bare covariance:
\[ C^{-1} = -\Delta + m^2 +
 3 (-g) C^0_{xx} \;, \qquad C^0 = (-\Delta + m^2)^{-1} \; .\] 
 This renders all the diagrams finite for 
 $d<3$. In $d=3$ one needs to furthermore renormalize the sunset insertions which, in the  BPHZ prescription \cite{bogoliubow1957multiplikation,Hepp:1966eg,Zimmermann:1969jj}, 
 is achieved by using as bare covariance:
 \[ 
 C^{-1} = (C^0)^{-1} +
 3 (- g) C^0_{xx} + 6 (-g)^2 \int_y (C^0)^3_{xy} \;. \]
 The terms proportional to $g$ and $g^2$ are particular instances of counterterms, and they are treated perturbatively.
 In detail, taking the inverse in the sense of formal power series up to order $g^3$ we get:
 \[
 \begin{split}
  C_{xy} = & C^0_{xy} + C^0_{xz} \Big( - 3 (- g) C^0_{zz} - 6(-g)^2 \int_{z_1} (C^0)^3_{zz_1} \Big) C^0_{zy} \crcr
  & + C^0_{xz_1} (3 (-g) C^0_{z_1z_1}) C^0_{z_1z_2} (3(-g) C^0_{z_2z_2}) C^0_{z_2 y}
+ 2 C^0_{xz_1} (3 (-g) C^0_{z_1z_1}) C^0_{z_1z_2} \Big(6 (-g)^2 \int_{z_3} (C^0)^3_{z_2z_3} \Big) C^0_{z_2 y} \\
& - C^0_{xz_1} (3 (-g) C^0_{z_1z_1}) C^0_{z_1z_2} (3 (-g)C^0_{z_2z_2}) C^0_{z_2 z_3} (3 (-g) C^0_{z_3z_3}) C^0_{z_3y} + O(g^4)\;,
 \end{split}
 \]
and substituting this in the Feynman expansion we get up to order $g^3$ only finite terms for $0\le d<4$. The same structure of divergences arises for $C^{-1} = -\Delta +  x^2$ and they are subtracted by the same technique. 

The renormalization is more delicate in dimension $d=4$ where one needs to add an infinity of mass and coupling constant counterterms in order to subtract the divergences. The precise counterterms can be defined recursively and, in the BPHZ prescription, the recursion can be solved in terms of sums over forests of non overlapping divergent sub-graphs of products of Taylor localization operators \cite{Zimmermann:1969jj,Rivasseau:1991ub}. Rewriting the Feynman expansion in terms of renormalized couplings leads to finite $n$-point functions order by order at all orders in perturbation theory.
In $d>4$ instead the model is non-renormalizable, in the sense that at each perturbative order one needs to introduce counterterms also for new higher-order interactions.

Lastly, a more subtle problem is that, if in $d\geq 4$ one attempts to take the ultraviolet cutoff to infinity with finite bare coupling, then the renormalized coupling goes to zero \cite{Frohlich:1982tw,Aizenman:2019yuo}. This is the so-called triviality problem. 

These remarks are made here for the purpose of providing a quick overview, but we do not consider such problems in the rest of the paper.
\end{remark}

\section{Low orders of the stochastic expansion}
\label{app:loworderstoc}

We now turn to low-order examples for the stochastic expansion.
For concreteness, we consider again a single local interaction of order $q+1$.
Starting from Eq.~\eqref{eq:SPDTsoltree} and summing together the terms which are equal up to re-embedding, denoting collectively $\pmb{x} = (t,x)$ and introducing the Green function $\pmb{C}$, at low orders in the $\phi^{q+1}$ model we obtain:
\[
\begin{split}
  \Phi_{\pmb{x}}
 = & \pmb{C}_{\pmb{x} \pmb{x}_1} \xi_{\pmb{x}_1}
 + (- g) \int_{\pmb{x}_1} \pmb{C}_{ \pmb{x} \pmb{x}_1}
 \left[ \int_{\pmb{x}_2}\pmb{C}_{\pmb{x}_1 \pmb{x}_2}
 \xi_{\pmb{x}_2} \right]^{q} \crcr
&  +(-g)^2 q \int_{\pmb{x}_1} \pmb{C}_{\pmb{x} \pmb{x}_1}
   \left[ \int_{\pmb{x}'_2}\pmb{C}_{\pmb{x}_1 \pmb{x}'_2}
   \left[ \int_{\pmb{x}_3}\pmb{C}_{\pmb{x}'_2\pmb{x}_3}  \xi_{\pmb{x}_3} \right]^{q} \right]
  \left[ \int_{\pmb{x}_2} \pmb{C}_{\pmb{x}_1 \pmb{x}_2}
 \xi_{\pmb{x}_2} \right]^{q-1} \crcr
& + (-g)^3 q^2  \int_{\pmb{x}_1} \pmb{C}_{\pmb{x} \pmb{x}_1}
   \left[ \int_{\pmb{x}'_2}\pmb{C}_{\pmb{x}_1 \pmb{x}'_2}
   \left[ \int_{\pmb{x}'_3}\pmb{C}_{\pmb{x}'_2\pmb{x}'_3}
   \left[ \int_{\pmb{x}_4}\pmb{C}_{\pmb{x}'_3\pmb{x}_4} \xi_{\pmb{x}_4} \right]^{q} \right]
   \left[\int_{\pmb{x}_3}\pmb{C}_{\pmb{x}'_2\pmb{x}_3}  \xi_{\pmb{x}_3}\right]^{q-1} \right]
  \left[ \int_{\pmb{x}_2}\pmb{C}_{\pmb{x}_1 \pmb{x}_2}
 \xi_{\pmb{x}_2} \right]^{q-1} \crcr
&  + (-g)^3    \binom{q}{2} \int_{\pmb{x_1}} \pmb{C}_{\pmb{x} \pmb{x}_1}
   \left[ \pmb{C}_{\pmb{x}_1 \pmb{x}'_2}
   \left[ \int_{\pmb{x}_3}\pmb{C}_{\pmb{x}'_2\pmb{x}_3}  \xi_{\pmb{x}_3} \right]^{q} \right]^2
  \left[ \int_{\pmb{x}_2}\pmb{C}_{\pmb{x}_1 \pmb{x}_2}
 \xi_{\pmb{x}_2} \right]^{q-2} + O(g^4)\;,
\end{split}
\]
which is depicted in the case $q=3$ in Fig.~\ref{fig:tertree}.
We now evaluate some of the contributions to the stochastic two-point function in the local $\phi^4$ model at low orders. 

\paragraph{Connected contribution to $\Braket{\Phi_{(t,x)}\Phi_{(s,y)} }_{\xi}$ at order zero.} 
\begin{figure}[ht!]
    \begin{center}
\begin{tikzpicture}[scale=0.4]
\draw[black, thick] (0,0) -- (0,1);
\filldraw[black] (0,1) circle (4pt);

\draw[black, thick] (3,0) -- (3,1);
\filldraw[black] (3,1) circle (4pt);

\draw[red,thick] (0,1) edge [bend left=45](3,1);

\end{tikzpicture}
    \end{center}
    \caption{Zero order contribution from the mating of the two trees with no internal vertex. Each black dot represents a leafs decorated by a field $\xi$. The red line represents a Gaussian contraction of two such fields, $\braket{\xi_{(t,x)} \; \xi_{(t',x')}}_{\xi}=  2  \; \delta_{tt'} \delta_{xx'}$.}
    \label{fig:zeroord}
\end{figure}
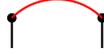

At zero order the connected part of the stochastic 2-point expectation  consists in two trees with no internal vertex mated, as depicted in Fig.~\ref{fig:zeroord}:
\[ \int_{\pmb {z}} \pmb{C}_{\pmb{x} \pmb{z}}
  \; 2 \; \pmb{C}_{\pmb{y} \pmb{z}}=  \left( \frac{e^{- |t-s| A  }}{A}  \right)_{xy} \; ,
\]
which is exactly the stochastic two-point function in the free case and at equal times reproduces the path integral two-point function.

\paragraph{Connected contribution to $\Braket{\Phi_{(t,x)}\Phi_{(t,y)} }_{\xi}$ at order $g$.}
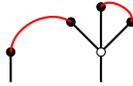
\begin{figure}[ht!]
    \begin{center}

\begin{tikzpicture}[scale=0.4]
\draw[black, thick] (0,0) -- (0,1);
\filldraw[black] (0,1) circle (4pt);

\draw[black, thick] (3,0) -- (3,1);

\draw[black, thick] (3,1) -- (2,2);
\filldraw[black] (2,2) circle (4pt);
\draw[black, thick] (3,1) -- (3,2.5);
\filldraw[black] (3,2.5) circle (4pt);
\draw[black, thick] (3,1) -- (4,2);
\filldraw[black] (4,2) circle (4pt);

\filldraw[fill=white] (3,1) circle (4pt);

\draw[red,thick] (0,1) edge [bend left=60](2,2);
\draw[red,thick] (4,2) edge [bend right=90](3,2.5);

\end{tikzpicture}
    \end{center}
    \caption{The connected first order contributions.}
        \label{fig:SPDEtad}
\end{figure}

At first order in the coupling we get two contributions, as depicted in Fig.~\ref{fig:SPDEtad},
corresponding to the two trees, one with one internal vertex $g$ and one with no internal vertex $g$, mated via all possible contractions of the four $\xi$ fields on the leaves:
 \[
 \begin{split}
  & 3 \int_{\pmb{x}_1}\pmb{C}_{\pmb{x} \pmb{x}_1}
    \left[ \int_{\pmb{x}_2}\pmb{C}_{\pmb{x}_1 \pmb{x}_2}
    2 \pmb{C}_{\pmb{x}_1\pmb{x}_2} \right]
    \left[ \int_{\pmb{x}_3 }\pmb{C}_{\pmb{x}_1\pmb{x}_3 }
     2 \pmb{C}_{\pmb{y} \pmb{x}_3 } \right]
  +   3 \int_{\pmb{x}_2} \left[ \int_{\pmb{x}_1}\pmb{C}_{\pmb{x} \pmb{x}_1}
   2\pmb{C}_{\pmb{x}_2 \pmb{x}_1 } \right]
    \left[ \int_{\pmb{z}} \pmb{C}_{\pmb{x}_2 \pmb{z} }
     2 \pmb{C}_{\pmb{x}_2\pmb{z}} \right]
    \pmb{C}_{\pmb{y}  \pmb{x}_2 } \crcr
& \qquad = 3 \int_{\pmb{x}_1}\left( \frac{1}{A}\right)_{x_1x_1} \pmb{C}_{\pmb{x} \pmb{x}_1}
  \left( \frac{e^{- |t_1-t| A  }}{A}  \right)_{x_1y} +
  3  \int_{\pmb{x}_2} \left( \frac{1}{A}\right)_{x_2x_2}
  \left( \frac{e^{- |t-t_2| A  }}{A}  \right)_{xx_2}  \pmb{C}_{\pmb{y}\pmb{x_2} } \; ,
  \end{split}
 \]
and the two terms combine into ($x_1$ and $x_2$ are integrated):
 \[
 \begin{split}
& 3 \int_{-\infty}^t dt_1 \; \left( \frac{1}{A} \right)_{x_1x_1} \left[  ( e^{-(t-t_1)A})_{xx_1} \left( \frac{e^{- (t-t_1 ) A  }}{A}  \right)_{x_1y} 
 +  \left( \frac{e^{- (t-t_1) A  }}{A}  \right)_{xx_1}
  ( e^{-(t-t_1)A})_{yx_1}   \right] \crcr
& =3 \int_{-\infty}^t dt_1 \; \left( \frac{1}{A} \right)_{x_1x_1} \frac{d}{dt_1} \left[  \left( \frac{e^{- (t-t_1) A  }}{A}  \right)_{xx_1}  \left( \frac{e^{- (t-t_1 ) A  }}{A}  \right)_{x_1y}  \right]  = 3 \left( \frac{1}{A} \right)_{xx_1}\left( \frac{1}{A} \right)_{x_1x_1}\left( \frac{1}{A} \right)_{yx_1} \;.
 \end{split}
 \]
reproducing the tadpole contribution to the path integral 2-point function, including the combinatorial factor 3 (see Fig.~\ref{fig:firstord}).

\paragraph{Sunset contributions to $\Braket{\Phi_{(t,x)}\Phi_{(t,y)} }_{\xi}$ at order $g^2$.}
\begin{figure}[ht!]
    \begin{center}
    
\begin{tikzpicture}[scale=0.4]
\draw[black, thick] (0,0) -- (0,1);
\filldraw[black] (0,1) circle (4pt);

\draw[black, thick] (4,0) -- (4,1);

\draw[black, thick] (4,1) -- (3,2);
\filldraw[black] (3,2) circle (4pt);
\draw[black, thick] (4,1) -- (4,2.5);
\filldraw[black] (4,2.5) circle (4pt);
\draw[black, thick] (4,1) -- (5,2);
\filldraw[black] (5,2) circle (4pt);

\draw[black, thick] (3,2) -- (2,3);
\filldraw[black] (2,3) circle (4pt);
\draw[black, thick] (3,2) -- (3,3.5);
\filldraw[black] (3,3.5) circle (4pt);
\draw[black, thick] (3,2) -- (4,3);
\filldraw[black] (4,3) circle (4pt);

\filldraw[fill=white] (4,1) circle (4pt);
\filldraw[fill=white] (3,2) circle (4pt);

\draw[red,thick] (0,1) edge [bend left=60](2,3);
\draw[red,thick] (4,2.5) arc (-115:115:0.3);
\draw[red,thick] (5,2) arc (-45:150:1.25);

\draw[black, thick] (8,0) -- (8,1);

\draw[black, thick] (8,1) -- (7,2);
\filldraw[black] (7,2) circle (4pt);
\draw[black, thick] (8,1) -- (8,2.5);
\filldraw[black] (8,2.5) circle (4pt);
\draw[black, thick] (8,1) -- (9,2);
\filldraw[black] (9,2) circle (4pt);

\filldraw[fill=white] (8,1) circle (4pt);

\draw[black, thick] (11,0) -- (11,1);

\draw[black, thick] (11,1) -- (10,2);
\filldraw[black] (10,2) circle (4pt);
\draw[black, thick] (11,1) -- (11,2.5);
\filldraw[black] (11,2.5) circle (4pt);
\draw[black, thick] (11,1) -- (12,2);
\filldraw[black] (12,2) circle (4pt);

\filldraw[fill=white] (11,1) circle (4pt);

\draw[red,thick] (10,2) arc (0:180:0.5);
\draw[red,thick] (11,2.5) arc (0:180:1.5);
\draw[red,thick] (12,2) arc (0:180:2.5);

\end{tikzpicture}
    \end{center}
    \caption{Mated trees leading to sunset graphs.}
    \label{fig:SPDEsunsets}
\end{figure}
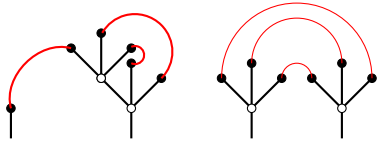

There are multiple contributions to the stochastic 2-point function at second order in the coupling. We focus here on the contributions to the sunset graph depicted in Fig.~\ref{fig:SPDEsunsets}. There are three mated trees contributing a total of: 
\[
\begin{split}
 & 3\cdot 6 \int_{\pmb{x}_1,\pmb{x}_2}\pmb{C}_{\pmb{x}\pmb{x}_1}
  \pmb{C}_{\pmb{x}_1\pmb{x}_2}
\left[\int_{\pmb{x}_3} \pmb{C}_{\pmb{x}_1\pmb{x}_3} 2  \pmb{C}_{\pmb{x}_2\pmb{x}_3} \right]^2
 \left[\int_{\pmb{z}} \pmb{C}_{\pmb{x}_2\pmb{z} } 2  \pmb{C}_{\pmb{y}\pmb{z}} \right] \crcr
& + 3\cdot 6 \int_{\pmb{y}_1 \pmb{y}_2}
\left[\int_{\pmb{z}} \pmb{C}_{\pmb{x}\pmb{z} }  2  \pmb{C}_{\pmb{y}_2\pmb{z}} \right]
\pmb{C}_{\pmb{y}_1\pmb{y}_2}
\left[ \int_{\pmb{z}} \pmb{C}_{\pmb{y}_2\pmb{z} }  2 \pmb{C}_{\pmb{y}_1\pmb{z}}  \right]^2
\pmb{C}_{\pmb{y}\pmb{y}_1} +
6 \int_{\pmb{x_1}\pmb{y}_1} \pmb{C}_{\pmb{x}\pmb{x}_1}
\left[\int_{\pmb{z}} \pmb{C}_{\pmb{x}_1\pmb{z} } 2 \pmb{C}_{\pmb{y}_1 \pmb{z}}  \right]^3
\pmb{C}_{\pmb{y}\pmb{y}_1} \;,
\end{split}
\]
and setting $s=t$ this writes in detail ($x_1$ and $x_2$ are integrated):
\[
\begin{split}
& 18 \int_{-\infty}^t dt_1 ( e^{-(t-t_1)A})_{xx_1} \int_{-\infty}^{t_1} dt_2 ( e^{-(t_1-t_2)A})_{x_1x_2}
 \left( \frac{e^{- (t_1-t_2) A  }}{A}  \right)^2_{x_1x_2}
 \left( \frac{e^{- (t-t_2) A  }}{A}  \right)_{x_2y} +18 \;(x\leftrightarrow y)\crcr
& + 6
 \int_{-\infty}^t dt_1
 ( e^{-(t-t_1)A})_{xx_1}
  \left( \frac{e^{- |t_1-t_2| A  }}{A}  \right)^3_{x_1y_1}
  \int_{-\infty}^{t} dt_2
   ( e^{-(t-t_2)A})_{y_1y} \; .
\end{split}
\]
The first term reduces to:
\[
 6 \int_{-\infty}^t dt_1 ( e^{-(t-t_1)A})_{xx_1} \int_{-\infty}^{t_1} dt_2  \frac{d}{dt_2}\left[  \left( \frac{e^{- (t_1-t_2) A  }}{A}  \right)^3_{x_1x_2} \right]
 \left( \frac{e^{- (t-t_2) A  }}{A}  \right)_{x_2y}  \;.
\]
Writing the second term similarly, splitting the last term according to 
$t_2<t_1$ and $t_1<t_2$ and changing variables $t_1\leftrightarrow t_2$ in this second case we get:
\[
6 \int_{-\infty}^t dt_1 ( e^{-(t-t_1)A})_{xx_1} \int_{-\infty}^{t_1} dt_2   \frac{d}{dt_2}\left[  \left( \frac{e^{- (t_1-t_2) A  }}{A}  \right)^3_{x_1x_2}   \left( \frac{e^{- (t-t_2) A  }}{A}  \right)_{x_2y}  \right]  + (x \leftrightarrow y) \; . 
\]
Observing that the boundary terms at $t_2=-\infty$ are $0$ we obtain:
\[
\begin{split}
& 6 \int_{-\infty}^t dt_1 
\left(  \frac{1}{A}  \right)^3_{x_1x_2} 
\Bigg[ 
 ( e^{-(t-t_1)A})_{xx_1} \left( \frac{e^{- (t-t_1) A  }}{A}  \right)_{x_2y}  + \left( \frac{e^{- (t-t_1) A  }}{A}  \right)_{xx_1}  ( e^{-(t-t_1)A})_{x_2y } 
 \Bigg] \crcr
& = 6 \left(  \frac{1}{A}  \right)^3_{x_1x_2}  
\int_{-\infty}^t dt_1  \frac{d}{dt_1}
\Bigg[ \left( \frac{e^{- (t-t_1) A  }}{A}  \right)_{xx_1} 
\left( \frac{e^{- (t-t_1) A  }}{A}  \right)_{x_2y} 
 \Bigg]  =6 \left(  \frac{1}{A}  \right)_{xx_1}  \left(  \frac{1}{A}  \right)^3_{x_1x_2}  \left(  \frac{1}{A}  \right)_{x_2y} \;,
\end{split}
\]
reproducing the sunset Feynman amplitude, including the combinatorial factor 6 (see Fig.~\ref{fig:firstord}).

\paragraph{Connected contribution to $\Braket{\prod_{i=1}^{q+1} \Phi_{(t,x_i)}  }_{\xi}$ in the $\phi^{q+1}$ model at order $g$.}
As a final example we consider the first order contribution to the connected $(q+1)$-point function in the $\phi^{q+1}$ models. In the Feynman expansion this is:
\[
 q! \int_x \prod_{i=1}^{q+1} C_{x_ix} \;.
\]
In the stochastic approach we have $q$ trees contributing, each with a factor $q!$ from the $\xi$ contractions yielding ($x$ is integrated):
\[
\begin{split}
& q!  \sum_{i=1}^{q+1} \int_{\pmb{x}} \pmb{C}_{\pmb{x}_i\pmb{x}}
  \prod_{j\neq i} \left[ \int_{\pmb{z}_j} \pmb{C}_{\pmb{x} \pmb{z}_j} 2  \pmb{C}_{\pmb{x}_j \pmb{z}_j} \right]
  = q!  \sum_{i=1}^{q+1} \int_{-\infty}^t
    dt' \;\int_x \; (e^{-(t-t') A})_{x_i x} \prod_{j\neq i}
    \left( \frac{e^{-(t-t') A}}{A} \right)_{x_j x}  \crcr
&  \qquad =  q! \int_x \int_{-\infty}^t
    dt'  \frac{d}{dt'} \left[ \prod_{j}
    \left( \frac{e^{-(t-t') A}}{A} \right)_{x_j x} \right]  = q! \int_x \prod_{j} \left(\frac{1}{A} \right)_{x_j x} \; .
\end{split}
\]

\section{Low order examples of Taylor interpolation}
\label{app:lowordex}

We present several explicit examples of the Taylor interpolation of the Feynman amplitude of a graph. At each step, we draw the edges of the graph that have not yet been assigned into tree or noise edges as dashed edges, the candidate tree edges as blue edges, the tree edges as black edges and the noise edges as black edges with a red cross.

\paragraph{Order $0$.} At this order, depicted in Fig.~\ref{fig:placeholder}, there is no interpolation to be made and we get
$ (\frac{1}{A})_{xy}$ which turns out to be equal to the stochastic amplitude of an edge with a $\xi$ insertion:
\[
  \left( \frac{1}{A} \right)_{xy} =  \left( \frac{e^{- |t-s| A  }}{A}  \right)_{xy} \bigg|_{t=s} 
\]

\begin{figure}[ht!]
    \centering
\begin{tikzpicture}[scale=0.4]
\draw[black, dashed,thick] (0,0) -- (2,0); 

\draw (3,0) node {=};

\draw[black,thick] (4.5,0) -- (6.5,0);

\draw (4.5,0) node[left] {$t$};
\draw (5.5,0) node[cross=3pt,red] {};
\draw (6.5,0) node[right] {$t$};

\end{tikzpicture}
    \caption{The one-edge graph.}
    \label{fig:placeholder}
\end{figure}

\paragraph{Order $g$.} At order $g$, we have $3$ tadpoles and we perform a first level expansion for every tadpole, as depicted in Fig.~\ref{fig:taylortadpole}:
\[
\begin{split}
& \left( \frac{1}{A} \right)_{x x_1} 
\left( \frac{1}{A} \right)_{x_1 x_1}  
\left( \frac{1}{A} \right)_{x_1 y}  
 = \left( \frac{1}{A} \right)_{x_1 x_1} \int_{-\infty}^{t} dt_1
  \frac{d}{dt_1} \left[ \left( \frac{e^{-(t-t_1)A}}{A}\right)_{xx_1} 
  \left(  \frac{e^{-(t-t_1)A}}{A} \right)_{x_1y}  \right] \crcr
&  = \int_{-\infty}^{t} dt_1 \left( e^{-(t-t_1)A} \right)_{xx_1} 
  \left( \frac{1}{A} \right)_{x_1 x_1} \left(  \frac{e^{-(t-t_1)A}}{A} \right)_{x_1y} + 
  \int_{-\infty}^{t} dt_1 \left(  \frac{e^{-(t-t_1)A}}{A} \right)_{x x_1}  
  \left( \frac{1}{A} \right)_{x_1 x_1} \left( e^{-(t-t_1)A} \right)_{x_1y}  \;.
\end{split}
\]

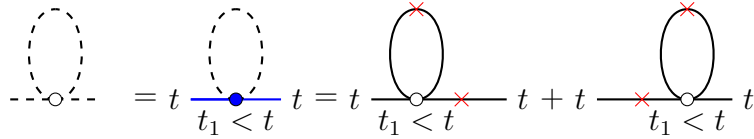
\begin{figure}[ht!]
    \centering
\begin{tikzpicture}[scale=0.6]

\draw[black,dashed, thick] (0,0) -- (1,0);
\draw[black,dashed, thick] (1,0) -- (2,0);
\draw[black,dashed, thick] (1,0) to[out=0,in=0] (1,2);
\draw[black, dashed, thick] (1,0) to[out=180,in=180] (1,2);
\filldraw[fill=white] (1,0) circle (4pt);
\draw (3,0) node {=};

\draw[black,blue, thick] (4,0) -- (5,0);
\draw[black,blue, thick] (4,0) -- (6,0);
\draw[black,dashed, thick] (5,0) to[out=0,in=0] (5,2);
\draw[black, dashed, thick] (5,0) to[out=180,in=180] (5,2);
\filldraw[fill=blue] (5,0) circle (4pt)  node[below] {$t_1<t$};
\draw (4,0) node [left] {$t$};
\draw (6,0) node [right] {$t$};

\draw (7,0) node {=};

\draw[black, thick] (8,0) -- (9,0);
\draw (8,0) node [left] {$t$};

\draw[black,thick] (9,0) -- (11,0);
\draw[black,thick] (9,0) to[out=0,in=0] (9,2);
\draw[black,thick] (9,0) to[out=180,in=180] (9,2);

\draw (9,2) node[cross=3pt,red] {};
\draw (10,0) node[cross=3pt,red] {};
\draw (11,0) node [right] {$t$};

\filldraw[fill=white] (9,0) circle (4pt)  node[below] {$t_1<t$};

\draw (11.5,0) node[right] {+};

\draw[black, thick] (13,0) -- (15,0)  node[below] {$t_1<t$};
\draw (13,0) node [left] {$t$};

\draw[black,thick] (15,0) -- (16,0);
\draw (16,0) node [right] {$t$};

\draw[black,thick] (15,0) to[out=0,in=0] (15,2);
\draw[black,thick] (15,0) to[out=180,in=180] (15,2);

\draw (15,2) node[cross=3pt,red] {};
\draw (14,0) node[cross=3pt,red] {};

\filldraw[fill=white] (15,0) circle (4pt);

\end{tikzpicture}
    \caption{Taylor interpolation of tadpoles. This is the inverse of what depicted in Fig.~\ref{fig:SPDEtad}.}
    \label{fig:taylortadpole}
\end{figure}

\paragraph{Sunset at $g^2$.} The successive interpolations for the sunset are depicted in Fig.~\ref{fig:candidatesunset}.

\begin{figure}[ht!]

\begin{center}
    
\begin{tikzpicture}[scale=0.4]

\draw[black,dashed, thick] (0,0) -- (1,0);
\draw[black,dashed, thick] (1,0) to[out=90,in=180] (3,2);
\draw[black,dashed, thick] (3,2) to[out=0,in=90] (5,0);
\draw[black,dashed, thick] (1,0) to[out=-90,in=180] (3,-2);
\draw[black,dashed, thick] (3,-2) to[out=0,in=-90] (5,0);
\draw[black,dashed, thick] (1,0) -- (5,0);
\draw[black,dashed, thick] (5,0) -- (6,0) ;
\filldraw[fill=white]  (1,0) circle (4pt);
\filldraw[fill=white]  (5,0) circle (4pt);
\draw (6,0) node[right] {=};

\draw[black,blue, thick] (8,0) -- (9,0);
\draw (8,0) node[left] {$t$};
\draw[black,dashed, thick] (9,0) to[out=90,in=180] (11,2);
\draw[black,dashed, thick] (11,2) to[out=0,in=90] (13,0);
\draw[black,dashed, thick] (9,0) to[out=-90,in=180] (11,-2);
\draw[black,dashed, thick] (11,-2) to[out=0,in=-90] (13,0);
\draw[black,dashed, thick] (9,0) -- (13,0);
\draw[black,blue, thick] (13,0) -- (14,0);
\draw (14,0) node[right] {$t$};

\filldraw[fill=blue]  (9,0) circle (4pt);
\draw (9,0) node[below] {$t_1<t$};

\filldraw[fill=blue]  (13,0) circle (4pt);
\draw (13,0) node[below] {$t_1<t$};

\draw (14.5,0) node[right] {=};

\draw[black, thick] (17,0) -- (18,0);
\draw (17,0) node[left] {$t$};
\draw[black,blue, thick] (18,0) to[out=90,in=180] (20,2);
\draw[black,blue, thick] (20,2) to[out=0,in=90] (22,0);
\draw[black,blue, thick] (18,0) to[out=-90,in=180] (20,-2);
\draw[black,blue, thick] (20,-2) to[out=0,in=-90] (22,0);
\draw[black,blue, thick] (18,0) -- (22,0);
\draw[black,blue, thick] (22,0) -- (23,0);
\draw (23,0) node[right] {$t$};

\filldraw[fill=black]  (18,0) circle (6pt);
\draw (18,0) node[below] {$t_1<t$};

\filldraw[fill=blue]  (22,0) circle (4pt);
\draw (22,0) node[below] {$t_2<t_1$};

\draw (24,0) node[right] {+};

\draw[black,blue, thick] (26.5,0) -- (27.5,0);
\draw (26.5,0) node[left] {$t$};
\draw[black,blue, thick] (27.5,0) to[out=90,in=180] (29.5,2);
\draw[black,blue, thick] (29.5,2) to[out=0,in=90] (31.5,0);
\draw[black,blue, thick] (27.5,0) to[out=-90,in=180] (29.5,-2);
\draw[black,blue, thick] (29.5,-2) to[out=0,in=-90] (31.5,0);
\draw[black,blue, thick] (27.5,0) -- (31.5,0);
\draw[black,black, thick] (31.5,0) -- (32.5,0);
\draw (32.5,0) node[right] {$t$};

\filldraw[fill=blue]  (27.5,0) circle (4pt);
\draw (27.5,0) node[below] {$t_2<t_1$};

\filldraw[fill=black]  (31.5,0) circle (6pt);
\draw (31.5,0) node[below] {$t_1<t$};

\end{tikzpicture}

\end{center}
    \caption{Taylor interpolation for the sunset graph.}
    \label{fig:candidatesunset}
\end{figure}
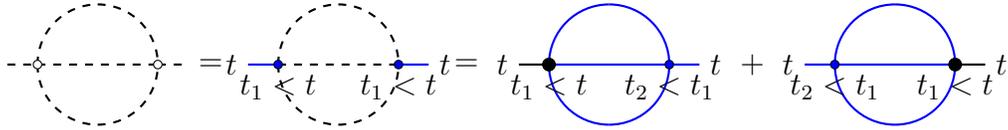

The first interpolation yields two possible tree edges: the left external edge or the right external edge. Once one of the two external external edges is picked, the second interpolation has 4 choices for choosing a new tree edge: either the second external edge, or any of the three internal edges. This is displayed in Fig.~\ref{fig:susnshiefinal}.
The factors $3 \times$ count which internal edge is a tree edge (hence does not have a cross). The terms with crosses on all the internal edges consist in two recursive forests, which we combine into two forests which respect the tree partial order (that $t_1$ and $t_2$ are not ordered as they belong to different trees).

\begin{figure}
\centering
\begin{subfigure}
        \centering
 
\begin{tikzpicture}[scale=0.4]

\draw (-1,0) node {$3 \times $};

\draw[black, thick] (1,0) -- (2,0);
\draw (1,0) node[left] {$t$};

\draw[black, thick] (2,0) to[out=90,in=180] (4,2);
\draw (4,2) node[cross=3pt,red] {};
\draw[black,thick] (4,2) to[out=0,in=90] (6,0);
\draw[black,thick] (2,0) to[out=-90,in=180] (4,-2);
\draw[black, thick] (4,-2) to[out=0,in=-90] (6,0);
\draw (4,-2) node[cross=3pt,red] {};

\draw[black, thick] (2,0) -- (6,0);
\draw[black, thick] (6,0) -- (8,0);
\draw (7,0) node[cross=3pt,red] {};
\draw (8,0) node[right] {$t$};

\filldraw[fill=black]  (2,0) circle (6pt);
\draw (2,0) node[below] {$t_1<t$};

\filldraw[fill=black]  (6,0) circle (6pt);
\draw (6,0) node[below] {$t_2<t_1$};

\draw (9,0) node[right] {+};

\draw[black, thick] (11.5,0) -- (12.5,0);
\draw (11.5,0) node[left] {$t$};

\draw[black, thick] (12.5,0) to[out=90,in=180] (14.5,2);
\draw (14.5,2) node[cross=3pt,red] {};
\draw[black,thick] (14.5,2) to[out=0,in=90] (16.5,0);
\draw[black,thick] (12.5,0) to[out=-90,in=180] (14.5,-2);
\draw[black, thick] (14.5,-2) to[out=0,in=-90] (16.5,0);
\draw (14.5,-2) node[cross=3pt,red] {};

\draw[black, thick] (12.5,0) -- (16.5,0);
\draw (14.5,0) node[cross=3pt,red] {};

\draw[black, thick] (16.5,0) -- (17.5,0);
\draw (17.5,0) node[right] {$t$};

\filldraw[fill=black]  (12.5,0) circle (6pt);
\draw (12.5,0) node[below] {$t_1<t$};

\filldraw[fill=black]  (16.5,0) circle (6pt);
\draw (16.5,0) node[below] {$t_2<t_1$};

\end{tikzpicture}
    \end{subfigure}\hspace{0.05\textwidth}%
\begin{subfigure} 
        \centering
\begin{tikzpicture}[scale=0.4]

\draw (-1,0) node {$3 \times $};

\draw[black, thick] (1,0) -- (3,0);
\draw (1,0) node[left] {$t$};

\draw (2,0) node[cross=3pt,red] {};

\draw[black, thick] (3,0) to[out=90,in=180] (5,2);
\draw (5,2) node[cross=3pt,red] {};
\draw[black,thick] (5,2) to[out=0,in=90] (7,0);
\draw[black,thick] (3,0) to[out=-90,in=180] (5,-2);
\draw[black, thick] (5,-2) to[out=0,in=-90] (7,0);
\draw (5,-2) node[cross=3pt,red] {};

\draw[black, thick] (3,0) -- (7,0);
\draw[black, thick] (7,0) -- (8,0);
\draw (8,0) node[right] {$t$};

\filldraw[fill=black]  (3,0) circle (6pt);
\draw (3,0) node[below] {$t_2<t_1$};

\filldraw[fill=black]  (7,0) circle (6pt);
\draw (7,0) node[below] {$t_1<t$};

\draw (9,0) node[right] {+};

\draw[black, thick] (11.5,0) -- (12.5,0);
\draw (11.5,0) node[left] {$t$};

\draw[black, thick] (12.5,0) to[out=90,in=180] (14.5,2);
\draw (14.5,2) node[cross=3pt,red] {};
\draw[black,thick] (14.5,2) to[out=0,in=90] (16.5,0);
\draw[black,thick] (12.5,0) to[out=-90,in=180] (14.5,-2);
\draw[black, thick] (14.5,-2) to[out=0,in=-90] (16.5,0);
\draw (14.5,-2) node[cross=3pt,red] {};

\draw[black, thick] (12.5,0) -- (16.5,0);
\draw (14.5,0) node[cross=3pt,red] {};

\draw[black, thick] (16.5,0) -- (17.5,0);
\draw (17.5,0) node[right] {$t$};

\filldraw[fill=black]  (12.5,0) circle (6pt);
\draw (12.5,0) node[below] {$t_2<t_1$};

\filldraw[fill=black]  (16.5,0) circle (6pt);
\draw (16.5,0) node[below] {$t_1<t$};

\end{tikzpicture}
    \end{subfigure}
    
    \caption{The four forests in the Taylor expansion of the sunset. The two trees in each forest are rooted at the external vertices. To be compared with Fig.~\ref{fig:SPDEsunsets}.
    }
    \label{fig:susnshiefinal}
\end{figure}
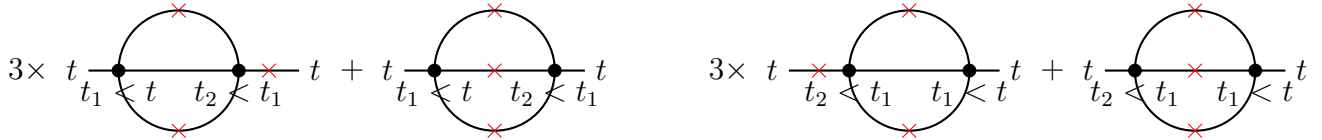

In equations this writes as:
\[
\begin{split}
& \left( \frac{1}{A} \right)_{xx_1}
\left( \frac{1}{A} \right)^3_{x_1y_1}
\left( \frac{1}{A} \right)_{y_1y} = 
\left( \frac{1}{A} \right)^3_{x_1y_1} \int_{-\infty}^{t} dt_1  \;
 \frac{d}{dt_1} \left[  \left( \frac{e^{-(t-t_1)A}}{A}\right)_{xx_1} 
  \left(  \frac{e^{-(t-t_1)A}}{A} \right)_{y_1y}  \right] \crcr
& =\int_{-\infty}^{t} dt_1  \; 
   \left( e^{-(t-t_1)A} \right)_{xx_1} \left( \frac{1}{A} \right)^3_{x_1y_1}
  \left(  \frac{e^{-(t-t_1)A}}{A} \right)_{y_1y}  +  
 \int_{-\infty}^{t} dt_1  \; 
   \left(  \frac{e^{-(t-t_1)A}}{A} \right)_{xx_1}   
  \left( \frac{1}{A} \right)^3_{x_1y_1}   \left( e^{-(t-t_1)A} \right)_{y_1y}  \crcr
& = \int_{-\infty}^{t} dt_1  \; 
   \left( e^{-(t-t_1)A} \right)_{xx_1} \int_{-\infty}^{t_1} dt_2 
   \frac{d}{dt_2} \left[ 
   \left( \frac{ e^{-(t_1-t_2)A} }{A} \right)^3_{x_1y_1}
  \left(  \frac{e^{-(t-t_2)A}}{A} \right)_{y_1y}  \right] \crcr
  & +  
 \int_{-\infty}^{t} dt_1  \;  \int_{-\infty}^{t_1} dt_2 
   \frac{d}{dt_2} \left[ 
   \left(  \frac{e^{-(t-t_2)A}}{A} \right)_{xx_1}   
  \left( \frac{ e^{-(t_1-t_2)A} }{A} \right)^3_{x_1y_1} \right]   \left( e^{-(t-t_1)A} \right)_{y_1y}  \;.
\end{split}
\]
Computing the second order interpolation we obtain the sum over four terms:
\[
\begin{split}
& \int_{-\infty}^{t} dt_1  \; 
   \left( e^{-(t-t_1)A} \right)_{xx_1} \int_{-\infty}^{t_1} dt_2 \crcr
& \qquad \left[ 3    \left( e^{-(t_1-t_2)A}  \right)_{x_1y_1}
   \left( \frac{ e^{-(t_1-t_2)A} }{A} \right)^2_{x_1y_1}
  \left(  \frac{e^{-(t-t_2)A}}{A} \right)_{y_1y} +     \left( \frac{ e^{-(t_1-t_2)A} }{A} \right)^3_{x_1y_1} \left( e^{-(t-t_2)A} \right)_{y_1y} 
 \right] 
 \crcr
& + \int_{-\infty}^{t} dt_1  \;  \int_{-\infty}^{t_1} dt_2  \crcr
& \qquad \left[ 
 \left(  \frac{e^{-(t-t_2)A}}{A} \right)_{xx_1}  
3    \left( e^{-(t_1-t_2)A}  \right)_{x_1y_1}
   \left( \frac{ e^{-(t_1-t_2)A} }{A} \right)^2_{x_1y_1}
 +   
  \left( e^{-(t-t_2)A} \right)_{xx_1} 
  \left( \frac{ e^{-(t_1-t_2)A} }{A} \right)^3_{x_1y_1} 
 \right]   \left( e^{-(t-t_1)A} \right)_{y_1y}      \; .
\end{split}
\]
The last terms in the two parentheses add up and relabeling $t_1\leftrightarrow t_2$ they reconstitute:
\[
\int_{-\infty}^{t} dt_1  \; \int_{-\infty}^{t} dt_2
 \;  \left( e^{-(t-t_1)A} \right)_{xx_1}  
   \left( \frac{ e^{- | t_1-t_2| A} }{A} \right)^3_{x_1y_1} \left( e^{-(t-t_2)A} \right)_{y_1y}  \; .
\]

%

\providecommand{\href}[2]{#2}\begingroup\raggedright\endgroup


\end{document}